\documentclass[journal,onecolumn]{IEEEtran}
\usepackage[colorlinks, linkcolor=red, citecolor=blue]{hyperref}
\usepackage{amssymb,amsthm,amsmath,latexsym,cite}
\usepackage{lineno,hyperref,mathtools}
\usepackage{tikz}
\usepackage{pdflscape}
\usepackage{longtable}
\usepackage{microtype}
\usepackage{setspace}
\usepackage{amsmath}
\usepackage{amssymb}
\usepackage{bm}
\usepackage{amssymb}
\usepackage{amsthm}
\usepackage{multirow,booktabs}
\usepackage{cases}
\usepackage{tabularx}
\usepackage{adjustbox}
\usepackage[figuresright]{rotating}
\usepackage{longtable}
\usepackage{supertabular}
\usepackage{threeparttable}
\usepackage{booktabs}
\usepackage{multirow}
\usepackage{color}
\usepackage{cite}
\usepackage{multirow,booktabs}
\usepackage{cases}
\usepackage{tabularx}
\usepackage{adjustbox}
\usepackage[figuresright]{rotating}
\usepackage{longtable}
\usepackage{booktabs}
\usepackage{multirow}
\usepackage{color}
\usepackage{lineno,hyperref,mathtools}
\usepackage{soul}
\usepackage[ruled,linesnumbered]{algorithm2e}
\usepackage{algpseudocode}
\usepackage{pifont}
\usepackage{booktabs}
\usepackage{graphicx}
\usepackage{orcidlink}
\usepackage{enumitem}
\allowdisplaybreaks
\newtheorem{theorem}{Theorem}

\newtheorem{problem}{Problem}
\newtheorem{statement}{Statement}

\newtheorem{lemma}[theorem]{Lemma}
\newtheorem{corollary}[theorem]{Corollary}

\theoremstyle{definition} 
\newtheorem{remark}{Remark}
\newtheorem{example}{Example}
\newtheorem{definition}{Definition}
\theoremstyle{remark}

\newcommand{\F}{\mathbb{F}}

\newcommand {\ccc}{{\mathbf{c}}}
\newcommand{\C}{\mathcal{C}}
\newcommand{\D}{{\mathcal{D}}}

\newcommand{\Hom}{{\mathrm{Hom}}}

\newcommand{\Hull}{{\mathrm{Hull}}}
\newcommand{\Col}{{\mathrm{Col}}}

\newcommand{\wt}{{{\rm{wt}}}}

\newcommand{\rank}{{\rm rank}}

\definecolor{mustfixcolor}{rgb}{0.95,0.00,0.85}
\definecolor{prooffixcolor}{rgb}{0.45,0.00,0.80}
\definecolor{polishfixcolor}{rgb}{0.00,0.45,0.18}
\definecolor{datafixcolor}{rgb}{0.90,0.35,0.00}
\definecolor{oldtextcolor}{rgb}{0.45,0.45,0.45}

\newcommand{\ignore}[1]{}
\makeatletter
\newcommand{\rmnum}[1]{\romannumeral #1}
\newcommand{\Rmnum}[1]{\expandafter\@slowromancap\romannumeral #1@}
\makeatother

\begin{document}
\title{Generalized Extended Codes with Applications in Entanglement-Assisted Qubit and Qutrit Codes}
\author{Yang Li${}^{\orcidlink{0000-0003-0286-9263}}$, 
Martianus Frederic Ezerman${}^{\orcidlink{0000-0002-5851-2717}}$, 
Shitao Li${}^{\orcidlink{0000-0002-9716-6212}}$, 
San Ling${}^{\orcidlink{0000-0002-1978-3557}}$, 
and Zhonghua Sun${}^{\orcidlink{0000-0002-8975-1163}}$
\thanks{
This research is supported by the Nanyang Technological University Research Grant under No. 04INS000047C230GRT01.}
\thanks{Yang Li is with the School of Physical and Mathematical Sciences, 
Nanyang Technological University, 21 Nanyang Link, Singapore 637371, Singapore
(email: yanglimath@163.com).} 
\thanks{Martianus Frederic Ezerman is with the School of Physical and Mathematical Sciences, Nanyang Technological University, 21 Nanyang Link, Singapore 637371, Singapore
(email: fredezerman@ntu.edu.sg).} 
\thanks{Shitao Li is with the School of Internet, Anhui University, Hefei, Anhui 230039, China (email: lishitao0216@163.com).}
\thanks{San Ling is with the School of Physical and Mathematical Sciences, 
Nanyang Technological University, 21 Nanyang Link, Singapore 637371, Singapore
(email: lingsan@ntu.edu.sg). 
He is also with VinUniversity, Vinhomes Ocean Park, Gia Lam, Hanoi 100000, Vietnam (email: ling.s@vinuni.edu.vn).} 
\thanks{Zhonghua Sun is with the School of Mathematics, Hefei University of Technology, Hefei,
230601, Anhui, China (email: sunzhonghuas@163.com).}}

\maketitle
\begin{abstract}
Given a linear code $\mathcal{C}$ of length $n$ over $\mathbb{F}_{q^2}$ and a nonzero vector ${\bf u}\in \mathbb{F}_{q^2}^n$,  
Sun, Ding, and Chen introduced the second kind of extended construction resulting in the code $\overline{\C}({\bf u})$. Their construction generalizes the standard extended construction $\overline{\C}(-{\bf 1})$, that is, when ${\bf u}$ is fixed to $-{\bf 1}$. They further showed that every $[n,k,d]_{q^2}$ linear code with $d \geq 2$ can be obtained from this construction for a suitable choice of the initial code $\mathcal{C}$ and the extension vector ${\bf u}$, up to permutation equivalence. To construct entanglement-assisted quantum error-correcting codes (EAQECCs) with more flexible and potentially better parameters, we consider the Hermitian dual of such extended linear codes and apply monomial variations on them. The resulting family of linear codes consist of $q^2$-ary codes $\C$ with respect to vector ${\bf u}\in \F_{q^2}^n \setminus \C$ and the scalar $a\in \F_{q^2}^*$. We call them generalized extended codes, with the code denoted by $\C({\bf u},a)$.

We prove that any generalized extended code is monomially equivalent to the Hermitian dual of a code which is closely related to a second kind of extended code of $\C^{\perp_{\rm H}}$. Every $[n+1,k+1]_{q^2}$ linear code $\D$ with $d(\D^{\perp_{\rm H}})>1$ is monomially equivalent to the generalized extended code $\C({\bf u},a)$ of an $[n,k]_{q^2}$ linear code $\C$ for a fixed $a\in\F_{q^2}^{*}$ and some ${\bf u}\in\F_{q^2}^{n}$. We then characterize the Hermitian hull and Hermitian dual distance of $\C({\bf u},a)$ in terms of the position of ${\bf u}$ relative to $\C+\C^{\perp_{\rm H}}$ and the interaction between ${\bf u}$ and the minimum weight codewords of $\C^{\perp_{\rm H}}$, respectively. We obtain explicit criteria to independently control the expected Hermitian hull dimension and Hermitian dual distance of $\C({\bf u},a)$. In particular, several conditions for simultaneously increasing the Hermitian hull dimension and the Hermitian dual distance of $\C({\bf u},a)$ are derived. 

Applying these results to the Hermitian construction for EAQECCs gives us $267$ new EA qubit codes of lengths $n \leq 40$ and  
$14$ new EA qutrit codes of lengths $n \leq 25$ compared to the best-known codes in Grassl's code tables and the imporvements recorded in very recent works in the literature. Among the new parameter sets, we confirm improvements for $236$ qubit and $8$ qutrit codes.
\end{abstract}

\begin{IEEEkeywords}
Entanglement-assisted, quantum code, Hermitian dual distance, Hermitian hull, generalized extended code.
\end{IEEEkeywords}

\section{Introduction}
\IEEEPARstart{T}{hroughout} this paper, let $\F_q$ denote the {\em finite field} with $q$ elements and let
$\F_q^*=\F_q\setminus\{0\}$, where $q$ is a prime power. An
$[n,k,d]_{q^2}$ {\em linear code} $\C$ is a $k$-dimensional linear subspace of $\F_{q^2}^n$ with minimum Hamming distance $d:=d(\C)$. 
We use ${\bf 0}$ and ${\bf 1}$ to denote appropriate all zero and all one vectors, respectively. 

\subsection{Hermitian Hulls and Equivalent Linear Codes}

The {\em Hermitian inner product} of any two vectors 
${\bf x}=(x_1,x_2,\ldots,x_n)$ and ${\bf y}=(y_1,y_2,\ldots,y_n)$ in $\F_{q^2}^n$ is $\langle {\bf x}, {\bf y} \rangle_{\rm H}=\sum_{i=1}^n x_i y_i^q$. 
The {\em Hermitian dual} and the {\em Hermitian hull} of $\C$ are defined, respectively, by 
\[
\C^{\perp_{\rm H}}= \{{\bf y}\in \F_{q^2}^n \, : \, 
\langle {\bf x}, {\bf y} \rangle_{\rm H}=0 \mbox{ for all }{\bf x}\in \C\} 
\mbox{ and } 
\Hull_{\rm H}(\C)=\C\cap \C^{\perp_{\rm H}}.
\]
The code $\C$ is {\em Hermitian self-orthogonal} if $\Hull_{\rm H}(\C)=\C$. The code is {\em Hermitian dual-containing} if $\Hull_{\rm H}(\C)=\C^{\perp_{\rm H}}$. It is Hermitian self-orthogonal if and only if $\C^{\perp_{\rm H}}$ is Hermitian dual-containing.
If $\C$ is an $[n,k]_{q^2}$ linear code, 
then $\C^{\perp_{\rm H}}$ is an $[n,n-k]_{q^2}$ linear code 
and $\Hull_{\rm H}(\C)$ is an $[n,\ell]_{q^2}$ Hermitian self-orthogonal code, with $\ell$ being the hull dimension.

A linear code $\C$ can be completely described by a generator matrix $G:=G(\C)$ whose rows form a basis. A generator matrix $H$ of $\C^{\perp_{\rm H}}$ is a \emph{Hermitian parity-check matrix} of $\C$. Two linear codes $\C_1$ and $\C_2$ that can be generated, respectively, by $G_1$ and $G_2$ are {\em permutation equivalent} if there exists a permutation matrix $P$ such that $G_1 \, P$ is a generator matrix of $\C_2$. The two codes are {\em monomially equivalent} if there exists a monomial matrix $M$ such that $G_1 \, M$ is a generator matrix of $\C_2$. We note that $G_2$ is not required to be equal to $G_1P$ or $G_1M$. 
Equivalent codes have the same length, dimension, minimum distance, and Hermitian dual distance. Luo {\it et al.} in \cite{LEGL2023} have shown that permutation equivalent codes have the same Hermitian hull dimension. Luo {\it et al.} in \cite{LSEL2024} and Chen in \cite{Chen2023-DCC} established that monomial equivalence may either increase or decrease the hull dimension. 

\subsection{EAQECCs from the Second Kind of Extended Codes}

Quantum error-correcting codes (QECCs), or simply quantum codes, provide a fundamental mechanism that can protect quantum information from decoherence and other sources of quantum noise \cite{Shor1995,Steane1996,CalderbankShor1996,KKKS2006}. 
The entanglement-assisted variant, with the codes often called EAQECCs in short, utilizes pre-shared entanglement between the sender and the receiver. This allows one to construct quantum codes from arbitrary classical linear codes without requiring self-orthogonality or dual containment \cite{BrunDevetakHsieh2006,GHMR2019}. 

An $[[n,\kappa,\delta;c]]_q$ EAQECC encodes $\kappa$ logical qudits into $n$ physical qudits. Its (quantum) minimum distance is $\delta$, provided that $c$ maximally entangled pairs have been shared. We typically assume $c > 0$ to distinguish the EAQECC from a QECC as the latter does not require a pre-shared entanglement. To evaluate the performance of an $[[n,\kappa,\delta;c]]_q$ EAQECC, we consider its {\em rate}, {\em net rate}, and {\em error-correcting capacity}, which are defined, respectively, as 
\[
\rho:=\frac{\kappa}{n}, \quad 
\overline{\rho}:=\frac{\kappa-c}{n}, \quad \mbox{ and } \quad
e:= \left\lfloor \frac{\delta-1}{2} \right\rfloor. 
\]
Here, $\rho$ measures the information transmission efficiency, 
$\overline{\rho}$ reflects the effective coding rate after accounting for the entanglement, and $e$ is the number of arbitrary qudit errors that the EAQECC can correct. It has also been confirmed in \cite{LSEL2024} that EAQECCs can indeed outperform QECCs in terms of 
error-correcting capacity when the rate is fixed. EAQECCs with larger values of $\rho$, $\overline{\rho}$, and $e$ are preferred. An EAQECC $\mathcal{Q}_1$ is {\em better than} $\mathcal{Q}_2$, or 
$\mathcal{Q}_1$ has {\em improved parameters} compared to $\mathcal{Q}_2$, if $\mathcal{Q}_1$ has at least an improvement among the parameters $\rho$, {$\overline{\rho}$}, and $e$ when the other parameters are fixed. Equivalently, an $[[n_1,\kappa_1,\delta_1;c_1]]_q$ EAQECC $\mathcal{Q}_1$ is better 
than an $[[n_2,\kappa_2,\delta_2;c_2]]_q$ EAQECC $\mathcal{Q}_2$ 
if $n_1\leq n_2$, $\kappa_1\geq \kappa_2$, $\delta_1\geq \delta_2$, and $c_1\leq c_2$, 
with at least one of these inequalities being strict. 

The symplectic, CSS, and Hermitian constructions \cite{BrunDevetakHsieh2006,GHMR2019} form several general construction methods to derive EAQECCs from classical linear codes with respect to the symplectic, Euclidean, and Hermitian inner products, respectively. 
Compared to the other two construction methods, the Hermitian construction in \cite[Theorem 3]{GHMR2019} is recognized 
as having greater potential in constructing EAQECCs. We rephrase it in the next lemma. 
\begin{lemma}{\rm (\!\! Hermitian construction)}\label{lem.Hermitian-Construction}
If $\C$ is an $[n,k]_{q^2}$ linear code whose Hermitian dual is $\C^{\perp_{\rm H}}$, then there exists an $[[n,\kappa,\delta;c]]_q$ EAQECC $\mathcal{Q}$ with parameters
\[
c=k-\dim\bigl(\Hull_{\rm H}(\C)\bigr), \quad \kappa=n-2k+c, \quad 
\delta=
\begin{cases}
d^{\perp_{\rm H}},
& \text{if } \C^{\perp_{\rm H}}\subseteq \C,\\[0.6ex]
\wt\bigl(\C^{\perp_{\rm H}}\setminus \Hull_{\rm H}(\C)\bigr)\geq d^{\perp_{\rm H}},
& \text{otherwise}.
\end{cases}
\]
The code $\mathcal{Q}$ is {\em pure} if $\delta=d^{\perp_{\rm H}}$ and is {\em impure} otherwise.
\end{lemma}

Many classes of EAQECCs with good parameters have been found by applying the Hermitian construction in Lemma \ref{lem.Hermitian-Construction} to specific families of classical linear codes. The works done, \textit{e.g.}, in \cite{LSEL2023,LEGL2023,FFLZ2019,Cao2021,LSEL2024,Chen2023-DCC} and many references therein can be consulted for further details. Notably, by using optimal and best-known linear codes documented in \cite{G2026-linear}, Luo {\it et al.} systematically constructed EAQECCs, particularly EA qubit and EA qutrit codes, of a wide range of lengths in \cite{LEGL2023}. Together with the codes obtained from \cite{G2021-PRA}, EA qubit codes of lengths up to $64$ and EA qutrit codes of lengths up to $36$ have been collected in \cite{G-binary-2025} and \cite{G-ternary-2025}, respectively. By using linear codes from the building-up construction \cite{LLS2025,K2023}, concatenated codes \cite{FCLLC2025}, the coordinate-wise juxtaposition of codes \cite{LLZS2024,LZ2024-DM,LSL2023}, nested GRS codes and related punctured codes \cite{CLZ2025}, and cyclic codes and punctured symplectic self-orthogonal quasi-cyclic codes \cite{ZK2026,LZ2024-JAMC}, some {\em sporadic examples} with improved parameters have also been discovered.

Sun, Ding, and Chen introduced and investigated the second kind of extended code $\overline{\C}({\bf u})$ of a linear code $\C$ in \cite{SDC2024FFA} and \cite{SDC2024DM} as a generalization of the standard extension technique. For an $[n,k]_{q^2}$ linear code $\C$ and a vector ${\bf u}=(u_1,u_2,\ldots,u_n)\in \F_{q^2}^n$, 
{\em the second kind of extended code} of $\C$ with respect to ${\bf u}$ is 
\begin{align}\label{eq.intro-extended-code}
\overline{\C}({\bf u})
=
\left\{
(c_1,c_2,\ldots,c_n,c_{n+1}):\
(c_1,c_2,\ldots,c_n)\in \C,\
c_{n+1}=\sum_{i=1}^{n}c_i u_i
\right\}.
\end{align}
If $\C$ has a generator matrix $G$ and a Hermitian parity-check matrix $H$, 
then the respective generator and Hermitian parity-check matrices of $\overline{\C}({\bf u})$ are 
\begin{align}\label{eq.extended-matrix}
\overline{G}:= 
\begin{pmatrix}
    G & G {\bf u}^{\top}
\end{pmatrix} \mbox{ and } 
\overline{H}=
\begin{pmatrix}
H & {\bf 0}^{\top}\\
{\bf u}^{q} & -1
\end{pmatrix}.   
\end{align}
In particular, the extended code $\overline{\C}({\bf u})$ become the standard extended code if ${\bf u}=-{\bf 1}$. A remarkable feature of this construction is its universality. As Sun {\it et al.} have proven in \cite[Theorem 12]{SDC2024FFA}, every $[n,k,d]_q$ linear code with $d>1$ is permutation equivalent to the extended code $\overline{\C'}({\bf u})$ of an $[n-1,k]_q$ linear code $\C'$ for some ${\bf u}\in \F_q^{n-1}$. In light of \eqref{eq.extended-matrix}, any linear code with minimum distance larger than $1$ can be expressed algebraically in terms of the second kind of extended code. 

\subsection{Our Motivations and Contributions} 

The preceding discussion leads to the following observations.
\begin{enumerate}
\item Classical linear codes with flexible parameters and explicit hull structures are likely to yield EAQECCs with parameters that correspond to better trade-offs. Lemma \ref{lem.Hermitian-Construction} deals directly with the key values of $\dim(\Hull_{\rm H}(\C))$ and $d^{\perp_{\rm H}}(\C)$.

\item The family of second kind of extended codes can indeed provide candidates of classical linear codes with flexible parameters in terms of \cite[Theorem 12]{SDC2024FFA} and explicit algebraic structures in terms of \eqref{eq.intro-extended-code} and \eqref{eq.extended-matrix}. 

\item Permutation equivalence preserves the Hermitian dual distance and the Hermitian hull dimension. Monomial equivalence preserves the Hermitian dual distance but may change the Hermitian hull dimension. To obtain EAQECCs with more flexible and better parameters, 
it is therefore natural to add monomial variations to second kind of extended codes and use the resulting codes as classical ingredients in the Hermitian construction in Lemma \ref{lem.Hermitian-Construction}.
\end{enumerate}

Motivated by the above observations, we introduce a modified version of the second kind of extended code and call it a
{\em generalized extended code}. Let $\C$ be an $[n,k]_{q^2}$ linear code.
Given a vector $\mathbf{u}\in \F_{q^2}^n$ and an element
$a\in \F_{q^2}^*$, we define the generalized extended code of $\C$ with respect to 
the extension vector $\mathbf{u}\in \F_{q^2}^n\setminus \C$ and the scalar $a\in \F_{q^2}^*$ by 
\begin{align}\label{eq.modified}
  \C({\bf u},a):=\{ {\bf c}_{({\bf u},a)}=({\bf c}+b {\bf u},~ ab) :~{\bf c}\in \C,~ b\in \F_{q^2}\},  
\end{align}
where the condition $\mathbf{u}\in \F_{q^2}^n\setminus \C$ 
excludes the trivial case in which $\C(\mathbf{u},a)$ contains the word $(0,0,\ldots,0,1)$. 
By Theorem \ref{th:gec-dual-extension-universality}, $\C({\bf u},a)$ is monomially
equivalent to the Hermitian dual of the second kind of extended code of $\overline{\C^{\perp_{\rm H}}}({\bf u}^q)$ 
and every $[n+1,k+1]_{q^2}$ linear code $\D$ with
$d(\D^{\perp_{\rm H}})>1$ is monomially equivalent to a generalized
extended code. 
Although the Hermitian hulls of a linear code and its Hermitian dual coincide,
the generalized extended code gives more flexibility in controlling the
Hermitian hull dimension, while retaining the advantages of second-kind
extended codes in terms of parameter flexibility and explicit algebraic
structure. We state the central question that this paper and answer it affirmatively.
\begin{problem}\label{prob.intro}
Can generalized extended codes $\C({\bf u},a)$ be used to produce improved EAQECCs?
\end{problem}

Our {\em main contributions} are summarized as follows, with the notation explained soon after in Subsection \ref{subsec.notation}. 

\begin{enumerate}
\item Theorem \ref{th:gec-dual-extension-universality} states the advantages of the generalized extended codes $\C({\bf u},a)$ over the second kind of extended codes. Remark \ref{remark:why-gec} provides further explanation.

To determine the number of required maximally entangled pairs in an EAQECC obtained from $\C({\bf u},a)$, Theorem \ref{th.hull} gives a complete characterization of
$\dim(\Hull_{\rm H}(\C({\bf u},a)))$ in terms of
$\dim(\Hull_{\rm H}(\C))$, according to the position of the extension vector
${\bf u}$ relative to the space $\C+\C^{\perp_{\rm H}}$. 
Although the general formula involves a $\{1\}$-inverse of $G \, G^\dagger$, where
$G$ is a generator matrix of $\C$, we prove that the resulting scalar is
independent of the chosen inverse.

With the help of the key reduction in Statement \ref{state.1}, 
Theorem \ref{th.hull222} gives a more explicit characterization of
$\dim(\Hull_{\rm H}(\C({\bf u},a)))$ and determines the explicit structures of
$\Hull_{\rm H}(\C({\bf u},a))$ in the three possible cases. 
The difference between Theorems \ref{th.hull} and \ref{th.hull222} is illustrated
in Remark \ref{rem:eaqecc-meaning-reduction} and Figure \ref{fig:reduced-search-space}. 

Lemma \ref{lem.uu} is on the existence of extension vectors
${\bf u}\in \C^{\perp_{\rm H}}\setminus\Hull_{\rm H}(\C)$ satisfying
${\bf u}{\bf u}^\dagger=0$ or ${\bf u}{\bf u}^\dagger\neq 0$. 
Based on the lemma, Theorem \ref{th:necessary_condition_hull} gives necessary
and sufficient conditions for $\dim(\Hull_{\rm H}(\C({\bf u},a)))=\dim(\Hull_{\rm H}(\C))$ 
and for $\dim(\Hull_{\rm H}(\C({\bf u},a)))=\dim(\Hull_{\rm H}(\C))+1$. 
As a consequence, Corollary \ref{coro.so} supplies necessary and sufficient
conditions to obtain Hermitian self-orthogonal codes and linear codes with large Hermitian hulls from $\C({\bf u},a)$.

\item To evaluate the error-correcting capability of an EAQECC obtained from $\C({\bf u},a)$, 
we determine in Theorem \ref{th.distance} that the Hermitian dual distance of $\C({\bf u},a)$ is 
either $d^{\perp_{\rm H}}$ or $d^{\perp_{\rm H}}+1$, where $d^{\perp_{\rm H}}$ is the Hermitian dual distance of the initial code $\C$. In the cases of ${\bf u}\in \F_{q^2}^n\setminus \C$ and ${\bf u}\in \C^{\perp_{\rm H}}\setminus\Hull_{\rm H}(\C)$, Theorems \ref{th.subcode_characterization} and \ref{th:dualspace_characterization} give two equivalent characterizations of the condition for $d(\C({\bf u},a)^{\perp_{\rm H}})=d^{\perp_{\rm H}}+1$ in terms of the existence of maximal subcodes of $\C^{\perp_{\rm H}}$ that satisfy specific properties. Moreover, we obtain a sufficient existence condition for an extension vector ${\bf u}$ in Theorem \ref{th.sufficient_increment} by an approach from finite geometry. 
Remarks \ref{rem:feasibility_conditions} and \ref{rem:sufficient_increment} give further discussions on the feasibility of the sufficient condition.

To obtain better EAQECCs, Theorem \ref{th.increasing_both} identifies when the Hermitian hull dimension and the Hermitian dual distance of $\C({\bf u},a)$ can be \emph{increased simultaneously}. Corollary \ref{cor:constructive_choice_u} confirms the existence of $a$ to achieve this goal for any given 
\[
{\bf u} \in \C^{\perp_{\rm H}} \setminus
\left(
\mathcal Z_{\rm H}
\cup
\bigcup_{{\bf e}\in \mathbf W_{\min}(\C^{\perp_{\rm H}})} S_{\bf e}
\right)\subseteq \C^{\perp_{\rm H}}\setminus \Hull_{\rm H}(\C),
\] 
where $\mathcal Z_{\rm H}$ is the set of self-orthogonal codewords in
$\C^{\perp_{\rm H}}$ and $S_{\bf e}$ is the maximal subcode of $\C^{\perp_{\rm H}}$ 
orthogonal to ${\bf e}\in \mathbf W_{\min}(\C^{\perp_{\rm H}})$. 
Example \ref{ex.classical} illustrates the operational meaning of Theorem \ref{th.increasing_both} and Corollary \ref{cor:constructive_choice_u} in the classical setting, 
where we show that, even when the base code $\C$ is not optimal, our generalized extended codes can produce optimal linear codes with larger Hermitian hulls compared to the records in Grassl's online tables \cite{G2026-linear}. 
\item Applying the above theoretical results, Theorem \ref{th:application-search-principle} and
Table \ref{tab:parameter-transformations} then translate the possible changes
$(\Delta_{\Hull_{\rm H}},\Delta_{d^{\perp_{\rm H}}})$ into six families of EAQECCs. 
Examples \ref{ex:binary-MI-example} and \ref{ex:improved-eaqecc-small} illustrate 
concrete ways to apply our results to obtain new and improved EA qubit and qutrit codes from $\C({\bf u},a)$.

Compared to the best-known EA qubit and qutrit codes in Grassl's online tables \cite{G-binary-2025,G-ternary-2025} and the latest records in \cite{LLS2025,K2023,FCLLC2025,LLZS2024,LZ2024-DM,CLZ2025,LZ2024-JAMC}, 
we obtain $267$ new qubit codes of lengths up to $40$ and $14$ qutrit codes of lengths up to $25$. Among these entanglement-assisted codes, $236$ qubit and $8$ qutrit are confirmed to have improved parameters. In fact, $180$ of them have \emph{simultaneous improvements in two or more parameters}. Such improvements are rather uncommon in the literature. 
Tables \ref{tab:breakthroughs2140} and \ref{tab:breakthroughs0118} list the $281$ new and improved EAQECCs obtained from $\C({\bf u},a)$. Table \ref{tab:change-distribution} summarizes their distribution. In addition, we find $10$ known EAQECCs with improved parameters compared to \cite{G-binary-2025,G-ternary-2025} in Table \ref{tab:literature-separated-eaqeccs}. 
As stated in Remark \ref{rem.known-dual}, these $10$ extra improvements have already been reported in the literature based on differing methods, while we obtain them in a unified manner.
\end{enumerate}

The rest of this paper is organized as follows. In Section \ref{sec.pre}, we recall basic
notation used in this paper and preliminary results on generalized extended codes. 
In Section \ref{sec.hull}, we study the Hermitian hulls of generalized extended codes. 
In Section \ref{sec.improvment}, we characterize when the Hermitian dual distance can be increased, and when the Hermitian hull
dimension and the Hermitian dual distance can be improved simultaneously. In
Section \ref{sec.application}, we present improved EA qubit and qutrit codes. Finally, Section \ref{sec.conclusion} concludes the paper and discusses possible research directions.
\section{Notation and preliminary discussions}
\label{sec.pre}
\subsection{Basic Notation and Two Useful Lemmas}\label{subsec.notation}
\begin{itemize}
    \item The \emph{sum} of two linear codes $\C_1$ and $\C_2$ of the same length is  
    $
    \C_1+\C_2:=\{\ccc_1+\ccc_2:~\ccc_1 \in \C_1,~\ccc_2\in \C_2\}. 
    $
    If $\C_1\cap \C_2=\{{\bf 0}\}$, then $\C_1+\C_2$ is called the {\em direct sum} of $\C_1$ and $\C_2$, denoted by $\C_1\oplus \C_2$.
 

    \item We use $G^\dagger$ and ${\bf u}^\dagger$ to denote the {\em conjugate transpose} of 
a matrix $G$ and a vector ${\bf u}$ over $\F_{q^2}$, respectively. 

    \item For any two vectors ${\bf u}$ and ${\bf v}$ in $\F_{q^2}^n$, 
    we say that ${\bf u}$ is {\em Hermitian orthogonal} to ${\bf v}$ if $\langle {\bf u}, {\bf v}\rangle_{\rm H}=0$.
    
    \item $\Col(G)$ denotes the {\em column space} of $G$.  
    
    \item $\mathbf{W}_{\min}(\C^{\perp_{\rm H}}):=\{{\bf e}\in \C^{\perp_{\rm H}}:~ \wt({\bf e})=d^{\perp_{\rm H}}\}$, 
    where $d^{\perp_{\rm H}}:=d(\C^{\perp_{\rm H}})$. 
    In general, $\mathbf{W}_{\min}(\C^{\perp_{\rm H}})$ cannot form a linear subspace. 

    \item For a finite-dimensional vector space $V$ over $\mathbb{F}_{q^2}$, 
    $\Hom(V, \mathbb{F}_{q^2})$ denotes the {\em space of all $\F_{q^2}$-linear functionals} from $V$ to $\mathbb{F}_{q^2}$. It is clear that $\dim(\Hom(V, \mathbb{F}_{q^2}))=\dim(V)$.  

    \item A subcode $S \subsetneq \C$ is called a {\em maximal subcode} if its {\em codimension} in $\C$ is $1$, that is, $\dim(\C) - \dim(S) = 1$.
    
    \item For a linear map $f$, $\ker(f)$ and $\operatorname{Im}(f)$ denote the {\em kernel} and the {\em image} of $f$, respectively.
    
    \item If $W$ is a subspace of a vector space $V$, then $V/W$ denotes the {\em quotient space} of $V$ modulo $W$. 
    
    \item We use $\phi$ to denote a {\em canonical map} or isomorphism. 
    In our context, the canonical linear map 
    \[
    \phi: \mathbb{F}_{q^2}^n \to \Hom(\C^{\perp_{\rm H}},\F_{q^2})
    \]
    is defined by $\phi(\mathbf{u})(\mathbf{x}) = \mathbf{x}\mathbf{u}^\dagger$ 
    for all $\mathbf{x} \in \C^{\perp_{\rm H}}$.

    \item For an $[n,k,d]_{q^2}$ linear code $\C$, the set 
    \[
    \C\times\{0\}:=\{({\bf c},0):~{\bf c}\in \C\}
    \]
    yields an $[n+1,k,d]_{q^2}$ linear code, 
    which is a subcode of $\C({\bf u},a)$ for any ${\bf u}\in \F_{q^2}^n$ and $a\in \F_{q^2}^*$.
\end{itemize}

\begin{lemma}{\rm (\!\! \cite[Lemma 4]{LEGL2023})}\label{lem:hull-dim-gram}
If $\C$ is an $[n,k]_{q^2}$ linear code with a generator matrix $G$, then
\[
\dim(\Hull_{\rm H}(\C))=k-\rank\!\left(G \, G^\dagger\right).
\]
\end{lemma}

\begin{lemma}{\rm (\!\! \cite[Theorem 12]{SDC2024FFA})}\label{lem.Sun-FFA}
Every $[n,k,d]_q$ linear code with $d>1$ is permutation equivalent to the extended code 
$\overline{\C'}({\bf u})$ of an $[n-1,k]_q$ linear code $\C'$ for some ${\bf u}\in \F_q^{n-1}$. 
\end{lemma}

\subsection{Preliminary Discussions on Generalized Extended Codes}
\begin{lemma}\label{lem.generator_matrix}
Let $\C$ be an $[n,k]_{q^2}$ linear code with generator matrix $G$ and Hermitian parity-check matrix $H$. For any $\mathbf{u}\in \mathbb{F}_{q^2}^n\setminus \C$ and $a\in \F_{q^2}^*$, the linear code $\C({\bf u},a)$ has parameters $[n+1,k+1]_{q^2}$ and respective generator and Hermitian parity-check matrices
\begin{align}\label{eq.generator-parity-matrix}
G_{({\bf u},a)}=
\begin{pmatrix}G & {\bf 0}\\
{\bf u} & a 
\end{pmatrix} \mbox{ and }
H_{({\bf u},a)}=
\begin{pmatrix}
H & -a^{-q}H{\bf u}^\dagger
\end{pmatrix}.
\end{align}
\end{lemma}
\begin{IEEEproof}
Since  $\mathbf{u}\in \mathbb{F}_{q^2}^n\setminus \C$, we can use \eqref{eq.modified} to confirm that $\C({\bf u},a)$ has parameters $[n+1,k+1]_{q^2}$ and generator matrix $G_{({\bf u},a)}$. To confirm that $H_{({\bf u},a)}$ is a parity-check matrix of $\C({\bf u},a)$, 
we use two facts, namely, $\rank(H_{({\bf u},a)})=n-k$ and
\[
\begin{pmatrix}
H & -a^{-q}H{\bf u}^\dagger
\end{pmatrix} 
\begin{pmatrix}
G & {\bf 0}\\
{\bf u} & a
\end{pmatrix}^\dagger 
= 
\begin{pmatrix}
H & -a^{-q}H{\bf u}^\dagger
\end{pmatrix}
\begin{pmatrix}
G^\dagger & {\bf u}^\dagger\\
{\bf 0} & a^q
\end{pmatrix}
=
\begin{pmatrix}
HG^\dagger & H{\bf u}^\dagger-H{\bf u}^\dagger
\end{pmatrix}
= {\bf 0}_{(n-k)\times (k+1)}.
\]
Thus, the desired conclusion is confirmed.
\end{IEEEproof}

\begin{theorem}\label{th:gec-dual-extension-universality}
Let $\overset{\rm p}{\backsimeq}$ and $\overset{\rm m}{\backsimeq}$ denote, respectively, permutation equivalence and monomial equivalence. If $\C$ is an $[n,k]_{q^2}$ linear code, then the following statements hold.
\begin{enumerate}
\item For any ${\bf u}\in\F_{q^2}^{n}$ and $a\in\F_{q^2}^{*}$, the generalized extended code $\C({\bf u},a)$ is \emph{permutation equivalent} to the Hermitian dual of a second kind of extended code of $\C^{\perp_{\rm H}}$, that is,
\[
\C({\bf u},a)
\overset{\rm p}{\backsimeq}
\left(
\overline{\C^{\perp_{\rm H}}}({\bf u}^q)
\right)^{\perp_{\rm H}}. 
\]
\item Every $[n+1,k+1]_{q^2}$ linear code $\D$ with $d(\D^{\perp_{\rm H}})>1$ is \emph{monomially equivalent} to the generalized extended code $\C({\bf u},a)$ of an $[n,k]_{q^2}$ linear code $\C$ for a fixed $a\in\F_{q^2}^{*}$ and some ${\bf u}\in\F_{q^2}^{n}$.
\end{enumerate}
\end{theorem}
\begin{IEEEproof}
\begin{enumerate}
\item If $G$ is a generator matrix of $\C$, then $G$ is a Hermitian parity-check matrix of $\C^{\perp_{\rm H}}$. By \eqref{eq.extended-matrix}, the second kind of extended code $\overline{\C^{\perp_{\rm H}}}({\bf u}^{q})$ has a Hermitian parity-check matrix
\[
\begin{pmatrix}
G & {\bf 0}^{\top}\\
{\bf u} & -1
\end{pmatrix}
\rightsquigarrow
\begin{pmatrix}
G & {\bf 0}^{\top}\\
{\bf u} & a
\end{pmatrix}
=
\begin{pmatrix}
G & {\bf 0}^{\top}\\
{\bf u} & -1
\end{pmatrix}
\, 
\operatorname{diag}(1,1,\ldots,1,-a),
\]
where $\operatorname{diag}(1,1,\ldots,1,-a)$ is an $(n+1)\times(n+1)$
monomial matrix with the specified elements along the diagonal and $\rightsquigarrow$ denotes the corresponding monomial
transformation. Taking Lemma~\ref{lem.generator_matrix} into account, the two matrices above are generator matrices of 
$\left(\overline{\C^{\perp_{\rm H}}}({\bf u}^{q})\right)^{\perp_{\rm H}}$ and $\C({\bf u},a)$, respectively. Thus, 
$\C({\bf u},a)
\overset{\rm p}{\backsimeq} 
\left(
\overline{\C^{\perp_{\rm H}}}({\bf u}^{q})
\right)^{\perp_{\rm H}}$.

\item Applying Lemma~\ref{lem.Sun-FFA} over $\F_{q^2}$ to
$\D^{\perp_{\rm H}}$, there exist an $[n,n-k]_{q^2}$ linear code $\D'$ and
a vector ${\bf v}\in\F_{q^2}^{n}$ such that $\D^{\perp_{\rm H}}$ is
permutation equivalent to $\overline{\D'}({\bf v})$. Since taking Hermitian dual preserves permutation equivalence, $\D$ is permutation equivalent to $\left(\overline{\D'}({\bf v})\right)^{\perp_{\rm H}}$. If $\C=(\D')^{\perp_{\rm H}}$, then $\C$ is an $[n,k]_{q^2}$ linear code and
$\D'=\C^{\perp_{\rm H}}$. By Statement 1), we have
\[
\D
\overset{\rm p}{\backsimeq}
\left(\overline{\D'}({\bf v})\right)^{\perp_{\rm H}}
=
\left(\overline{\C^{\perp_{\rm H}}}({\bf v})\right)^{\perp_{\rm H}}
\overset{\rm m}{\backsimeq}
\C({\bf v}^{q},a).
\]
Thus, $\D$ is monomially equivalent to $\C({\bf u},a)$ with ${\bf u}={\bf v}^{q}$ and any fixed $a\in\F_{q^2}^{*}$. 
\end{enumerate}
\end{IEEEproof}

\begin{remark} 
Keeping the notation in Theorem \ref{th:gec-dual-extension-universality}, the monomial equivalence in Statement 1) can
be sharpened after changing the extension vector in the second kind of extended code. More precisely,
\[
\C({\bf u},a)=
\left(\overline{\C^{\perp_{\rm H}}}(-a^{-q}{\bf u}^q)\right)^{\perp_{\rm H}}.
\]
Consequently, every $[n+1,k+1]_{q^2}$ linear code $\D$ with
$d(\D^{\perp_{\rm H}})>1$ is permutation equivalent to a generalized extended
code $\C({\bf u},a)$ of some $[n,k]_{q^2}$ linear code $\C$.
\end{remark}

The next remark explains why we introduce and work with generalized extended code $\C({\bf u},a)$ instead of working directly with the second kind of extended code $\overline{\C}({\bf u})$ in constructing EAQECCs with improved parameters.

\begin{remark}\label{remark:why-gec}
The Hermitian construction in Lemma~\ref{lem.Hermitian-Construction} identifies linear codes whose Hermitian hull dimension and Hermitian dual distance can be controlled effectively as the main classical ingredient. The hull dimension determines both the number of encoded qudits and the amount of required pre-shared entanglement. The
Hermitian dual distance gives a lower bound on the minimum distance of the resulting EAQECC. The generalized extended code $\C({\bf u},a)$ is designed to provide more flexibility than the second kind of extended code $\overline{\C}({\bf u})$ in controlling these two quantities.

\begin{enumerate}
\item We start with the Hermitian hull dimension. If $\C$ has
generator matrix $G$, then, by Lemma~\ref{lem:hull-dim-gram}, the dimension is determined by the \emph{rank} of the corresponding Gram matrix. For the second kind of extended code $\overline{\C}({\bf u})$, we can use \eqref{eq.extended-matrix} to infer
\begin{align*}
\overline{G} \, \overline{G}^{\dagger}
=
G \, G^{\dagger}+ G \, {\bf u}^{\top} \, {\bf u}^q \, G^{\dagger}
=
G \, \left(I_n+{\bf u}^{\top}{\bf u}^q \right) \, G^{\dagger},
\end{align*}
Thus, the Hermitian hull dimension of $\overline{\C}({\bf u})$ is governed by a rank-one perturbation of $G \, G^\dagger$. Although this form is useful, the influence of $I_n+{\bf u}^{\top}{\bf u}^q$ on the original Gram matrix $G \, G^\dagger$ is not easy to characterize.

In contrast, for the generalized extended code $\C({\bf u},a)$, it follows from \eqref{eq.generator-parity-matrix} that 
\begin{align}\label{eq.gram-block}
G_{({\bf u},a)} \, G_{({\bf u},a)}^\dagger
=
\begin{pmatrix}
G \, G^\dagger & G \, {\bf u}^\dagger\\
{\bf u} \, G^\dagger & {\bf u} \, {\bf u}^\dagger + a^{q+1}
\end{pmatrix}.
\end{align}
This block form separates the old Gram matrix $G \, G^\dagger$, the anti-diagonal terms determined by ${\bf u}$, and the adjustable scalar ${\bf u}{\bf u}^\dagger+a^{q+1}$. Hence, once the extension vector ${\bf u}$ is fixed, the nonzero element $a$ still provides an additional degree of freedom for adjusting the Hermitian hull dimension. This explicit block structure is the main reason why we can characterize the Hermitian hull of $\C({\bf u},a)$ effectively.

\item Regarding the Hermitian dual distance, both the second kind of extended codes and our generalized extended codes provide possible ways to control it through the choice of the extension vector ${\bf u}$. An equivalent relation in Theorem~\ref{th:gec-dual-extension-universality} shows that $\C({\bf u},a)$ is closely connected with the Hermitian dual of a second kind of extended code of $\C^{\perp_{\rm H}}$. Hence, such a code can still be useful for the purpose. In the case of our generalized extended code, once ${\bf u}$ is fixed, the Hermitian dual code has the simple form described in Lemma~\ref{lem.generator_matrix}. As will be shown later in Theorem~\ref{th.distance}, this leads to a direct criterion, which is independent of the choice of $a\in\F_{q^2}^{*}$, for determining whether the Hermitian dual distance of $\C({\bf u},a)$ increases. We can thus use ${\bf u}$ to control the Hermitian dual distance, while the scalar $a$ remains available for adjusting the Hermitian hull dimension.
\end{enumerate}
\end{remark}

\section{Hermitian hulls of generalized extended codes}\label{sec.hull}

This section is devoted to the study of Hermitian hulls of generalized extended codes. 
In Subsection~\ref{Subsec.hull1}, we characterize the dimension and structure of 
$\Hull_{\rm H}(\C({\bf u},a))$. In Subsection~\ref{Subsec.hull2}, we further derive explicit conditions under which suitable choices of ${\bf u}$ and $a$ allow us to control the dimension of $\Hull_{\rm H}(\C({\bf u},a))$.  

\subsection{Characterization of the Hermitian Hulls}\label{Subsec.hull1}
With the help of Lemma \ref{lem:hull-dim-gram} and \eqref{eq.gram-block}, 
we can now characterize the dimension of $\Hull_{\rm H}(\C({\bf u},a))$. 
We begin with the column space of $G \, G^\dagger$.

\begin{lemma}\label{lem:column-space-characterization}
Let $\C$ be an $[n,k]_{q^2}$ linear code with a generator matrix $G$ and let $\mathbf{u}\in \mathbb{F}_{q^2}^n$. 
We have $G \mathbf{u}^\dagger\in \Col\!\left(G \, G^\dagger\right)$ if and only if 
$\mathbf{u}\in \C+\C^{\perp_{\rm H}}$.
\end{lemma}
\begin{IEEEproof}
We first consider the necessary condition. Since
$G \mathbf{u}^\dagger\in \Col\!\left(G \, G^\dagger\right)$, there is a column vector $\mathbf{y}$ of length $k$ over $\F_{q^2}$ such that
$G\mathbf{u}^\dagger = G \, G^\dagger \, \mathbf{y}$. If $\mathbf{c} = \mathbf{y}^\dagger \, G\in \C$, where $\mathbf{y}^\dagger$ is a row vector of length $k$ over $\F_{q^2}$, then 
\[
G\mathbf{c}^\dagger = G \, G^\dagger\mathbf{y} = G\mathbf{u}^\dagger,
\]
which implies that $G(\mathbf{u}-\mathbf{c})^\dagger ={\bf 0}$ and, therefore, $
\mathbf{u}-\mathbf{c} \in \C^{\perp_{\rm H}}$. Thus, 
\[
\mathbf{u}=\mathbf{c}+(\mathbf{u}-\mathbf{c})\in \C+\C^{\perp_{\rm H}}.
\]

Conversely, since $\mathbf{u}\in \C+\C^{\perp_{\rm H}}$, there exist $\mathbf{c}\in \C$ and $\mathbf{e}\in \C^{\perp_{\rm H}}$ such that $\mathbf{u} = \mathbf{c}+\mathbf{e}$. Since $\mathbf{c}\in \C$, we can write $\mathbf{c} =\mathbf{x} \, G$ for some vector $\mathbf{x} \in \F_{q^2}^k$. Since $\mathbf{e}\in \C^{\perp_{\rm H}}$, 
we conclude the proof by verifying that  
\[
G\mathbf{u}^\dagger = G\mathbf{c}^\dagger+G\mathbf{e}^\dagger = G \, G^\dagger \mathbf{x}^\dagger \in \Col\left(G \, G^\dagger\right).
\]
\end{IEEEproof}

Let a linear code $\C$ with generator matrix $G$ be given. It is well-known that the corresponding Gram matrix $G \, G^\dagger$ is nonsingular if and only if 
$\dim(\Hull_{\rm H}(\C)) =0$. When $\dim(\Hull_{\rm H}(\C))>0$, the Gram matrix $G \, G^\dagger$ is singular and, hence, does not have an inverse. In this case, we need to use the $\{1\}$-inverse of $G \, G^\dagger$ to characterize the Hermitian hull of $\C({\bf u},a)$.

\begin{definition}
Let $A$ be an $m \times n$ matrix over a field $\mathbb{F}$. 
A matrix $X$ of size $n \times m$ over $\mathbb{F}$ is called a $\{1\}$-\emph{inverse} (or a \emph{generalized inverse}) 
of $A$ if it satisfies the condition
\begin{equation*}
A \, X \,  A = A.
\end{equation*}
The set of all $\{1\}$-inverses of $A$ is denoted by $A\{1\}$. A matrix $X \in A\{1\}$ is typically denoted by $A^{-}$.
\end{definition}

Every matrix over $\mathbb{F}_{q^2}$ has at least one $\{1\}$-inverse. The $\{1\}$-inverse of $A$ is uniquely determined and coincides with the inverse $A^{-1}$ if and only if $A$ is a square matrix with full rank. Although the $\{1\}$-inverse of $G \, G^\dagger$ is \emph{not} unique, it can still be used to characterize the Hermitian hull of $\C({\bf u},a)$, as stated in the following theorem. 

\begin{theorem}\label{th.hull}
Let $\C$ be an $[n,k]_{q^2}$ linear code with 
$\dim(\Hull_{\rm H}(\C))=\ell$. Let $S=\mathbf{u} G^\dagger (G \, G^\dagger)^- \, G\mathbf{u}^\dagger$, 
where $(G \, G^\dagger)^-$ denotes the $\{1\}$-inverse of $G \, G^\dagger$. For any $\mathbf{u}\in \mathbb{F}_{q^2}^n\setminus \C$ and $a\in \mathbb{F}_{q^2}^*$, the linear code $\C(\mathbf{u},a)$ has parameters $[n+1,k+1]_{q^2}$, with 
\begin{align}\label{eq.hull111}
\dim(\Hull_{\rm H}(\C(\mathbf{u},a)))
=\begin{cases} 
\ell -1 , & \mbox{if } \mathbf{u}\notin \C+\C^{\perp_{\rm H}}, \\
\ell, & \mbox{if } \mathbf{u}\in (\C+\C^{\perp_{\rm H}})\setminus \C \mbox{ and } \mathbf{u}\mathbf{u}^\dagger + a^{q+1} \neq S, \\
\ell + 1 , & \mbox{if } \mathbf{u}\in (\C+\C^{\perp_{\rm H}})\setminus \C \mbox{ and } \mathbf{u}\mathbf{u}^\dagger + a^{q+1} = S.  
\end{cases}
\end{align}
Moreover, the value of $S$ is independent of the choice of the $\{1\}$-inverse $(G \, G^\dagger)^-$.
\end{theorem}
\begin{IEEEproof}
By Lemmas \ref{lem:hull-dim-gram} and \ref{lem.generator_matrix}, we have 
\begin{align*}
G_{(\mathbf{u},a)} \, G_{(\mathbf{u},a)}^\dagger &= 
\begin{pmatrix} 
G \, G^\dagger & G\mathbf{u}^\dagger \\ 
\mathbf{u}G^\dagger & \mathbf{u}\mathbf{u}^\dagger + a^{q+1} 
\end{pmatrix}
\mbox{ and}\\
\dim(\Hull_{\rm H}(\C(\mathbf{u},a))) &= k+1 - \rank \left(G_{(\mathbf{u},a)} G_{(\mathbf{u},a)}^\dagger\right).
\end{align*}
We proceed to consider the following two cases. 
\begin{enumerate}[wide=0pt, labelindent=0pt,label=\textbf{Case \arabic*}:]
\item $\mathbf{u}\notin \C+\C^{\perp_{\rm H}}$. 
By Lemma \ref{lem:column-space-characterization}, we know that $G\mathbf{u}^\dagger\notin \Col(G \, G^\dagger)$, which implies $\rank((G \, G^\dagger \, G \mathbf{u}^\dagger))=\rank(G \, G^\dagger)+1$.  
Since $\left(G_{(\mathbf{u},a)} \, G_{(\mathbf{u},a)}^\dagger\right)^\dagger = G_{(\mathbf{u},a)} \, G_{(\mathbf{u},a)}^\dagger$, if the last row is a linear combination of the first $k$ rows, then there exists a row vector $\mathbf{e}$ of length $k$ such that $\mathbf{e}G \, G^\dagger = \mathbf{u}G^\dagger$. Taking the Hermitian conjugate yields $G \, G^\dagger \mathbf{e}^\dagger = G\mathbf{u}^\dagger$, which contradicts $G\mathbf{u}^\dagger \notin \Col(G \, G^\dagger)$. Hence, $\rank \left(G_{(\mathbf{u},a)} \, G_{(\mathbf{u},a)}^\dagger\right) = \rank(G \, G^\dagger) + 2$. 
Substituting the rank value into the formula for $\dim(\Hull_{\rm H}(\C(\mathbf{u},a)))$ gives 
\[
\dim(\Hull_{\rm H}(\C(\mathbf{u},a))) = k + 1 - (\rank(G \, G^\dagger) + 2) = (k - \rank(G \, G^\dagger)) - 1 = \ell -1. \]

\item $\mathbf{u}\in \C+\C^{\perp_{\rm H}}$. We perform the following block congruence transformation 
on $G_{(\mathbf{u},a)} G_{(\mathbf{u},a)}^\dagger$. 
\[
\begin{pmatrix} 
I_k & \mathbf{0} \\ 
-\mathbf{u}G^\dagger (G \, G^\dagger)^- & 1 
\end{pmatrix} 
G_{(\mathbf{u},a)} \, G_{(\mathbf{u},a)}^\dagger 
\begin{pmatrix} 
I_k & \mathbf{0} \\ 
-\mathbf{u}G^\dagger (G \, G^\dagger)^- & 1 
\end{pmatrix}^\dagger = 
\begin{pmatrix} 
G \, G^\dagger & \mathbf{0} \\ 
\mathbf{0} & \mathbf{u}\mathbf{u}^\dagger + a^{q+1} - S
\end{pmatrix},
\]
with $S:=\mathbf{u}G^\dagger (G \, G^\dagger)^- G\mathbf{u}^\dagger $ and $(G \, G^\dagger)^-$ being the generalized inverse matrix of $G \, G^\dagger$. We divide the analysis into two. 
\begin{enumerate}[wide=0pt,labelindent=0pt,label=\textbf{Subcase \arabic*}:]
    \item If $\mathbf{u}\mathbf{u}^\dagger + a^{q+1} \neq S$, then 
    $\rank \left(G_{(\mathbf{u},a)} \, G_{(\mathbf{u},a)}^\dagger\right) = \rank(G \, G^\dagger) + 1$. Thus,
    \[
    \dim(\Hull_{\rm H}(\C(\mathbf{u},a))) = k + 1 - (\rank(G \, G^\dagger) + 1) = \ell.
    \]
   \item If $\mathbf{u}\mathbf{u}^\dagger + a^{q+1} = S$, then 
    $\rank \left(G_{(\mathbf{u},a)} \, G_{(\mathbf{u},a)}^\dagger \right) = \rank(G \, G^\dagger)$. Thus, 
    \[
    \dim(\Hull_{\rm H}(\C(\mathbf{u},a))) = k + 1 - \rank(G \, G^\dagger) = \ell + 1.
    \]
\end{enumerate}
\end{enumerate}

Now that all cases have been covered, we establish that the value of $S$ is independent of the $\{1\}$-inverse $(G \, G^\dagger)^-$. If there are two row vectors ${\bf y}_1,~{\bf y}_2\in \F_{q^2}^k$ such that 
$G \, G^\dagger {\bf y}_1^\dagger=G{\bf u}^\dagger$ and $G \, G^\dagger {\bf y}_2^\dagger=G{\bf u}^\dagger$, then
\begin{align}\label{eq.S}
    S={\bf y}_i G \, G^\dagger (G \, G^\dagger)^{-}G \, G^\dagger {\bf y}_i^\dagger={\bf y}_iG \, G^\dagger{\bf y}_i^\dagger \mbox{ for } i=1,2.
\end{align}
Observing that $G \, G^\dagger ({\bf y}_1-{\bf y}_2)^\dagger={\bf 0}$, we let ${\bf z}={\bf y}_1-{\bf y}_2$. 
We use \eqref{eq.S} to infer that 
\begin{align*}
\begin{split}
S =  {\bf y}_1G \, G^\dagger{\bf y}_1^\dagger  
=  ({\bf z}+{\bf y}_2)G \, G^\dagger({\bf z}+{\bf y}_2)^\dagger  
=  {\bf z}G \, G^\dagger{\bf z}^\dagger+{\bf z}G \, G^\dagger{\bf y}_2^\dagger+
{\bf y}_2G \, G^\dagger{\bf z}^\dagger+{\bf y}_2G \, G^\dagger{\bf y}_2^\dagger 
=  {\bf y}_2G \, G^\dagger{\bf y}_2^\dagger. 
\end{split}
\end{align*}
Thus, $S$ is independent of the choice of $(G \, G^\dagger)^-$.
\end{IEEEproof}

\begin{remark}
Given an $[n,k]_{q^2}$ linear code with generator matrix $G$ and $\ell$-dimensional Hermitian hull, two facts are well-known. 
\begin{enumerate}
\item There exists a $k\times k$ nonsingular matrix $D$ such that 
\[
    D \, G \, G^\dagger \, D^\dagger=
    \begin{pmatrix}
    I_{k-\ell} & {\bf 0} \\ 
    {\bf 0} & {\bf 0}
    \end{pmatrix}.
\]
\item For any matrices $U$, $V$, and $W$ over $\F_{q^2}$, the $\{1\}$-inverse of $G \, G^\dagger$ must be of the form 
\[
(D^{\dagger})^{-1} 
\begin{pmatrix}
    I_{k-\ell} & U \\ 
    V & W
\end{pmatrix} 
D^{-1}.
\]
\end{enumerate}
Hence, there are $q^{2\ell(2k - \ell)}$ possible $\{1\}$-inverse matrices $(G \, G^\dagger)^{-}$ of $G \, G^\dagger$. 
We have just shown in Theorem \ref{th.hull} that the value of $S$ is \emph{independent of the choice of $(G \, G^\dagger)^{-}$}.
\end{remark}

We obtain $S={\bf y}_iG \, G^\dagger{\bf y}_i^\dagger~{\rm for}~i=1,2$ in \eqref{eq.S}. Since ${\bf y}_iG$ gives a codeword ${\bf c}_i\in \C$, if we let ${\bf e}_i :={\bf u}-{\bf c}_i$, then 
\[
G{\bf e}_i^{\dagger}
=
G{\bf u}^\dagger - G{\bf c}_i^{\dagger}
=
G \, G^\dagger{\bf y}_i^\dagger - G \, G^\dagger{\bf y}_i^\dagger
=
{\bf 0}
\]
and, hence, ${\bf e}_{i}\in \C^{\perp_{\rm H}}$ for $i=1,2$. 
Moreover, since ${\bf c}_i\in\C$ and ${\bf e}_i\in\C^{\perp_{\rm H}}$, we have
${\bf c}_i{\bf e}_i^\dagger={\bf e}_i{\bf c}_i^\dagger=0$. Therefore,
\[
{\bf u}{\bf u}^\dagger-S
=
({\bf c}_i+{\bf e}_i)({\bf c}_i+{\bf e}_i)^\dagger
-
{\bf y}_iG \, G^\dagger{\bf y}_i^\dagger
=
{\bf e}_i{\bf e}_i^\dagger.
\]
Thus, in this case, ${\bf u}{\bf u}^\dagger+a^{q+1}-S
=
{\bf e}_i{\bf e}_i^\dagger+a^{q+1}$, which is independent of the chosen representative ${\bf e}_i\in\C^{\perp_{\rm H}}$ in the decomposition ${\bf u}={\bf c}_i+{\bf e}_i$.

We highlight that this reduction follows directly from the definition of generalized extended codes. For any ${\bf c}_0\in\C$ and $a\in\F_{q^2}^{*}$,
we have
\[
\begin{aligned}
\C({\bf u}+{\bf c}_0,a)
&=\{({\bf c}+b({\bf u}+{\bf c}_0),ab):{\bf c}\in\C,\ b\in\F_{q^2}\}\\
&=\{(({\bf c}+b{\bf c}_0)+b{\bf u},ab):{\bf c}\in\C,\ b\in\F_{q^2}\} =\C({\bf u},a).
\end{aligned}
\]
The generalized extended code depends on the extension vector ${\bf u}$ only through its coset modulo $\C$. In particular, if ${\bf u}\in(\C+\C^{\perp_{\rm H}})\setminus\C$, then one may replace ${\bf u}$ by any representative in $\C^{\perp_{\rm H}}\setminus\Hull_{\rm H}(\C)$ without changing the resulting generalized extended code. The following statement formalizes the reduction.

\begin{statement}\label{state.1}
To study the generalized extended code $\C({\bf u},a)$ when ${\bf u}\in(\C+\C^{\perp_{\rm H}})\setminus\C$, it
suffices to consider ${\bf u} \in \C^{\perp_{\rm H}}\setminus\Hull_{\rm H}(\C)$ instead of $
{\bf u}\in(\C+\C^{\perp_{\rm H}})\setminus\C.$
\end{statement}

The statement provides the motivation to rewrite the second and third cases of \eqref{eq.hull111} in Theorem~\ref{th.hull} by taking the representative
${\bf u}\in\C^{\perp_{\rm H}}\setminus\Hull_{\rm H}(\C)$. In the next theorem, we determine the explicit form of $\Hull_{\rm H}(\C({\bf u},a))$. To do so, we need the following lemma.

\begin{lemma}\label{lem.maximal_subcode_general}
Let $\C$ be an $[n,k]_{q^2}$ linear code and let
${\bf u}\in \F_{q^2}^n$. Let
\[
\C_{\bf u}:=\left\{{\bf x}\in \C \, : \ {\bf u}{\bf x}^\dagger=0\right\}.
\]
If there exists an ${\bf x}\in \C$ such that ${\bf u}{\bf x}^\dagger\neq 0$, then $\C_{\bf u}$ is a maximal subcode of $\C$. In particular, $\dim(\C_{\bf u})=k-1$.
\end{lemma}
\begin{IEEEproof}
Consider the map $f_{\bf u} : \C \to \F_{q^2}$ that sends $
{\bf x} \mapsto {\bf u} \, {\bf x}^\dagger$. Under the assumption, $f_{\bf u}$ is a nonzero $\F_{q^2}$-linear map and $\C_{\bf u}\subsetneq \C$. 
Since $\dim_{\F_{q^2}}(\F_{q^2})=1$, we have $\dim(\operatorname{Im}(f_{\bf u}))=1$. 
By the rank-nullity theorem in linear algebra,
\[
\dim(\C_{\bf u})=\dim(\ker(f_{\bf u}))=\dim(\C)-\dim(\operatorname{Im}(f_{\bf u}))=k-1.
\]
Hence, $\C_{\bf u}$ has codimension one in $\C$ and, thus, is a maximal subcode of $\C$. 
\end{IEEEproof}

\begin{theorem}\label{th.hull222}
Let $\C$ be an $[n,k]_{q^2}$ linear code. 
For any $\mathbf{u}\in \mathbb{F}_{q^2}^n\setminus \C$, 
and $a\in \mathbb{F}_{q^2}^*$, the following statements hold. 
\begin{enumerate}
\item If $\dim(\Hull_{\rm H}(\C))=\ell$, then $\C(\mathbf{u},a)$ is an $[n+1,k+1]_{q^2}$ linear code with 
\begin{align}\label{eq.hull222}
\dim(\Hull_{\rm H}(\C(\mathbf{u},a)))
=
\begin{cases}
\ell-1, & \mbox{if } \mathbf{u}\notin \C+\C^{\perp_{\rm H}},\\
\ell, & \mbox{if } \mathbf{u}\in \C^{\perp_{\rm H}} \setminus \Hull_{\rm H}(\C) \mbox{ and } {\bf u}{\bf u}^\dagger+a^{q+1}\neq 0, \\
\ell+1, & \mbox{if } \mathbf{u}\in \C^{\perp_{\rm H}}\setminus \Hull_{\rm H}(\C) \mbox{ and } 
{\bf u}{\bf u}^\dagger+a^{q+1}=0.
\end{cases}
\end{align}
\item The Hermitian hull of $\C(\mathbf{u},a)$ is given by 
\begin{align}\label{eq.hull333}
\Hull_{\rm H}(\C(\mathbf{u},a)) =
\begin{cases}
\Hull_{\rm H}(\C)_{\bf u}\times \{0\}, & \mbox{if } \mathbf{u}\notin \C+\C^{\perp_{\rm H}},\\
\Hull_{\rm H}(\C)\times \{0\}, & \mbox{if } \mathbf{u}\in \C^{\perp_{\rm H}}\setminus \Hull_{\rm H}(\C) \mbox{ and } {\bf u}{\bf u}^\dagger+a^{q+1}\neq 0,\\
\Hull_{\rm H}(\C)({\bf u},a), & \mbox{if } \mathbf{u}\in \C^{\perp_{\rm H}}\setminus \Hull_{\rm H}(\C)\mbox{ and }{\bf u}{\bf u}^\dagger+a^{q+1}=0,
\end{cases}
\end{align}
where $\Hull_{\rm H}(\C)_{\bf u}:=\{{\bf c} \, : \, 
{\bf c}\in \Hull_{\rm H}(\C),\ {\bf c}{\bf u}^\dagger=0\}$ is a maximal subcode of $\Hull_{\rm H}(\C)$.
\end{enumerate}
\end{theorem}
\begin{IEEEproof}
Confirming the first assertion is straightforward by Theorem \ref{th.hull} and Statement \ref{state.1}. 

For the second assertion, we select any codeword 
${\bf c}_{({\bf u},a)} = ({\bf c}+b{\bf u},ab) \in \C({\bf u},a)$ with ${\bf c}\in \C$ and $b\in \F_{q^2}$. Then ${\bf c}_{({\bf u},a)}\in \Hull_{\rm H}(\C({\bf u},a))$ if and only if ${\bf c}_{({\bf u},a)}$ is Hermitian orthogonal
to every codeword of $\C({\bf u},a)$. Since $\C({\bf u},a)$ is generated by $({\bf c}',0)$, with ${\bf c}'\in \C$, and $({\bf u},a)$, we get
\begin{align}
& ({\bf c}+b{\bf u}) \, {\bf c}'^\dagger =0 \mbox{ for all }
{\bf c}'\in \C \mbox{ and} \label{eq:hull-structure-1}\\
& ({\bf c}+b{\bf u}){\bf u}^\dagger+a^{q+1}b =0. \label{eq:hull-structure-2}
\end{align}

We have two cases to consider. 
\begin{enumerate}[wide=0pt, labelindent=0pt,label=\textbf{Case \arabic*}:]
\item ${\bf u}\notin \C+\C^{\perp_{\rm H}}$. 
By \eqref{eq:hull-structure-1}, we have ${\bf c}+b{\bf u}\in \C^{\perp_{\rm H}}$. 
If $b\neq 0$, then
\[
{\bf u}=b^{-1}\bigl(({\bf c}+b{\bf u})-{\bf c}\bigr)\in \C+\C^{\perp_{\rm H}},
\]
which is a contradiction. Hence, $b=0$ and \eqref{eq:hull-structure-1} yields ${\bf c}\in \Hull_{\rm H}(\C)$. We rewrite \eqref{eq:hull-structure-2} as 
${\bf c}{\bf u}^\dagger=0$ and infer that 
\[
\Hull_{\rm H}(\C({\bf u},a))
=
\{({\bf c},0) \, : \, {\bf c}\in \Hull_{\rm H}(\C),\ {\bf c}{\bf u}^\dagger=0\} = \Hull_{\rm H}(\C)_{\bf u}\times \{0\}.
\]
To verify that $\Hull_{\rm H}(\C)_{\bf u}$ is a maximal subcode of $\Hull_{\rm H}(\C)$, 
it suffices to confirm the existence of some $\mathbf{c} \in \Hull_{\rm H}(\C)$ such that $\mathbf{c}\mathbf{u}^\dagger \neq 0$ according to Lemma \ref{lem.maximal_subcode_general}. 
If $\mathbf{c}\mathbf{u}^\dagger = 0$ for all $\mathbf{c} \in \Hull_{\rm H}(\C)$, then 
\[
\mathbf{u} \in (\Hull_{\rm H}(\C))^{\perp_{\rm H}} = (\C \cap \C^{\perp_{\rm H}})^{\perp_{\rm H}} = \C + \C^{\perp_{\rm H}},\]
which contradicts the assumption that $\mathbf{u} \notin \C + \C^{\perp_{\rm H}}$. Since such a $\mathbf{c}$ exists, we conclude that $\Hull_{\rm H}(\C)_{\mathbf{u}}$ is a maximal subcode of $\Hull_{\rm H}(\C)$. 

\item ${\bf u}\in \C^{\perp_{\rm H}}\setminus \Hull_{\rm H}(\C)$. The equality in \eqref{eq:hull-structure-1} reduces to $\ccc \ccc'^\dagger=0$, which means that ${\bf c} \in \Hull_{\rm H}(\C) \subseteq \C.$ Combining with the fact that ${\bf u}\in \C^{\perp_{\rm H}}$, we write \eqref{eq:hull-structure-2} as ${\bf c}{\bf u}^\dagger+ b\bigl({\bf u}{\bf u}^\dagger + a^{q+1}\bigr) = b\bigl({\bf u}{\bf u}^\dagger+a^{q+1}\bigr)=0$.

If ${\bf u}{\bf u}^\dagger+a^{q+1}\neq 0$, then $b=0$ and, hence,
\[
\Hull_{\rm H}(\C({\bf u},a))=\{(\ccc, 0) \, : \, \ccc\in \Hull_{\rm H}(\C)\}=\Hull_{\rm H}(\C)\times\{0\}.
\]

If ${\bf u}{\bf u}^\dagger+a^{q+1}=0$, then $b$ is arbitrary. Thus, 
\[
\Hull_{\rm H}(\C({\bf u},a))
=
\{({\bf c}+b{\bf u},ab) \, : \, 
{\bf c}\in \Hull_{\rm H}(\C),\ b\in \F_{q^2}\}
=
\Hull_{\rm H}(\C)({\bf u},a).
\]
\end{enumerate}

Having covered the two cases, the proof is now complete.
\end{IEEEproof}

\begin{remark}\label{rem:eaqecc-meaning-reduction}
The reduction in Statement~\ref{state.1}, which Figure~\ref{fig:reduced-search-space} illustrates, is useful for the construction of EAQECCs from generalized extended codes. By the Hermitian construction in Lemma~\ref{lem.Hermitian-Construction}, the
Hermitian hull dimension determines the dimension and the required number of pre-shared entanglement in the resulting EAQECCs. Hence, in selecting extension
vectors that preserve or increase the hull dimension, it is enough to work in the reduced space
$\C^{\perp_{\rm H}}\setminus\Hull_{\rm H}(\C)$ instead of the whole space $(\C+\C^{\perp_{\rm H}})\setminus\C$. This reduction eliminates equivalent choices of extension vectors that result in the same generalized extended code. It
also provides a simpler setting in which the influence of the scalar $a\in\F_{q^2}^{*}$ on the Hermitian hull dimension can be described explicitly. Consequently, Statement~\ref{state.1} simplifies the process of controlling the Hermitian hulls by reducing it to the study of representatives within $\C^{\perp_{\rm H}}\setminus\Hull_{\rm H}(\C)$.
\end{remark}

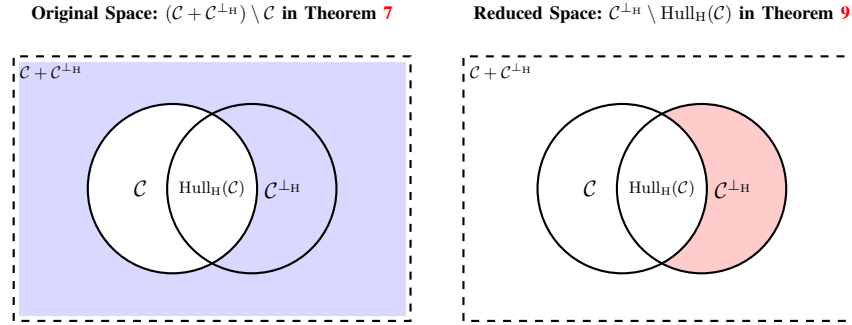
\begin{figure}[htbp]
\centering
\begin{tikzpicture}[scale=0.7, every node/.style={transform shape}]

  \begin{scope}[shift={(0,0)}]
    \node[above, font=\bfseries] at (0.75, 3)
    {Original Space: $(\C+\C^{\perp_{\rm H}})\setminus\C$ in Theorem \ref{th.hull}};
    
    \fill[blue!15] (-2.9,-2.4) rectangle (4.4,2.4);
    \fill[white] (0,0) circle (1.6);
    
    \draw[thick, dashed] (-3,-2.5) rectangle (4.5,2.5);
    \node[below right, font=\small] at (-3, 2.5)
    {$\C+\C^{\perp_{\rm H}}$};
    \draw[thick] (0,0) circle (1.6);
    \draw[thick] (1.5,0) circle (1.6);
    
    \node[font=\large] at (-0.6, 0) {$\C$};
    \node[font=\large] at (2.1, 0) {$\C^{\perp_{\rm H}}$};
    \node[font=\small] at (0.75, 0) {$\Hull_{\rm H}(\C)$};
  \end{scope}

  \begin{scope}[shift={(8.5,0)}]
    \node[above, font=\bfseries] at (0.75, 3)
    {Reduced Space: $\C^{\perp_{\rm H}}\setminus\Hull_{\rm H}(\C)$ in Theorem \ref{th.hull222}};
    
    \begin{scope}
      \clip (1.5,0) circle (1.6);
      \fill[red!20] (-3,-2.5) rectangle (4.5,2.5);
      \fill[white] (0,0) circle (1.6);
    \end{scope}
    
    \draw[thick, dashed] (-3,-2.5) rectangle (4.5,2.5);
    \node[below right, font=\small] at (-3, 2.5)
    {$\C+\C^{\perp_{\rm H}}$};
    \draw[thick] (0,0) circle (1.6);
    \draw[thick] (1.5,0) circle (1.6);
    
    \node[font=\large] at (-0.6, 0) {$\C$};
    \node[font=\large] at (2.1, 0) {$\C^{\perp_{\rm H}}$};
    \node[font=\small] at (0.75, 0) {$\Hull_{\rm H}(\C)$};
  \end{scope}
\end{tikzpicture}
\caption{Reduction of the search space for extension vectors within $(\C+\C^{\perp_{\rm H}})\setminus\C$. 
By Statement~\ref{state.1}, every vector
in the original space can be replaced, without changing the generalized
extended code, by a representative in
$\C^{\perp_{\rm H}}\setminus\Hull_{\rm H}(\C)$.}
\label{fig:reduced-search-space}
\end{figure}

\subsection{Controlling the Hermitian Hull Dimension}\label{Subsec.hull2}
Based on \eqref{eq.hull222}, there must be an ${\bf u}\in \F_{q^2}^n\setminus \C$ such that $\dim(\Hull_{\rm H}(\C(\mathbf{u}, a))) = \dim(\Hull_{\rm H}(\C)) - 1$ for any $a\in \F_{q^2}^*$, provided that $\C\oplus \C^{\perp_{\rm H}}\neq \F_{q^2}^n$. This is equivalent to saying that $\C$ is not a Hermitian LCD code. On the other hand, we can also establish necessary and sufficient conditions that prevent $\dim(\Hull_{\rm H}(\C(\mathbf{u}, a)))$ from decreasing. We reproduce a useful lemma before proving another one.

\begin{lemma}\label{lem.Hermitian}
Let $\C$ be a linear code over $\mathbb{F}_{q^2}$. 
If $\mathbf{x}\mathbf{x}^\dagger = 0$ for all $\mathbf{x} \in \C$, then $\C$ is Hermitian self-orthogonal.    
\end{lemma}
\begin{IEEEproof}
For any $\mathbf{x},\mathbf{y}\in\C$, since $\C$ is linear, we have
$\mathbf{x}+\mathbf{y}\in\C$. By the assumption, we have 
$
0=(\mathbf{x}+\mathbf{y})(\mathbf{x}+\mathbf{y})^\dagger
=\mathbf{x}\mathbf{y}^\dagger+\mathbf{y}\mathbf{x}^\dagger .
$ 
Similarly, taking $\lambda\in\mathbb F_{q^2}\setminus\mathbb F_q$, from
$\mathbf{x}+\lambda\mathbf{y}\in\C$ we get
$
0=\lambda^q\mathbf{x}\mathbf{y}^\dagger
+\lambda\mathbf{y}\mathbf{x}^\dagger .
$
Since $\lambda^q\ne\lambda$, combining this two equations yields that
$\mathbf{x}\mathbf{y}^\dagger=0$ for any $\mathbf{x},\mathbf{y}\in\C$. 
It follows that $\C\subseteq\C^{\perp_H}$. This completes the proof. 
\end{IEEEproof}

\begin{lemma}\label{lem.uu}
If $\C$ is an $[n,k]_{q^2}$ linear code with $\dim(\Hull_{\rm H}(\C)) = \ell$, then the following statements hold. 
\begin{enumerate}
    \item If $\ell = n-k-1$, then $\mathbf{u}\mathbf{u}^\dagger \neq 0$ for all $\mathbf{u} \in \C^{\perp_{\rm H}} \setminus \Hull_{\rm H}(\C)$.
    \item If $\ell < n-k-1$, then there exists a vector $\mathbf{u} \in \C^{\perp_{\rm H}} \setminus \Hull_{\rm H}(\C)$ such that $\mathbf{u}\mathbf{u}^\dagger = 0$.
    \item If $\ell \leq n-k-1$, then there exists a vector $\mathbf{w} \in \C^{\perp_{\rm H}} \setminus \Hull_{\rm H}(\C)$ such that $\mathbf{w} \mathbf{w}^\dagger \neq 0$.
\end{enumerate}
\end{lemma}
\begin{IEEEproof}
We define a quantity 
\[
\Delta := \dim(\C^{\perp_{\rm H}}) - \dim(\Hull_{\rm H}(\C)) = n - k - \ell.
\]
\begin{enumerate}
\item Since $\ell = n-k-1$, we have $\Delta = 1$. Let $\mathbf{u}$ be any vector in $\C^{\perp_{\rm H}} \setminus \Hull_{\rm H}(\C)$. By definition, $\Hull_{\rm H}(\C)$ consists of vectors in $\C^{\perp_{\rm H}}$ that are Hermitian orthogonal to the entire space $\C^{\perp_{\rm H}}$. Since $\Delta=1$ and $\mathbf{u} \in \C^{\perp_{\rm H}} \setminus \Hull_{\rm H}(\C)$, we can express $\C^{\perp_{\rm H}}$ as the direct sum 
\begin{align}\label{eq.direct-sum1}
    \C^{\perp_{\rm H}} = \Hull_{\rm H}(\C) \oplus \langle \mathbf{u} \rangle,
\end{align}
where $\langle \mathbf{u} \rangle$ denotes the $1$-dimensional subspace spanned by $\mathbf{u}$. Hence, any vector $\mathbf{x} \in \C^{\perp_{\rm H}}$ can be uniquely decomposed as $\mathbf{x} = c\mathbf{u} + \mathbf{h}$ for some element $c \in \mathbb{F}_{q^2}$ 
and a vector $\mathbf{h} \in \Hull_{\rm H}(\C)$. For a contradiction, let 
$\mathbf{u}\mathbf{u}^\dagger = 0$. Since ${\bf u}\in \C^{\perp_{\rm H}}$ and ${\bf h}\in \Hull_{\rm H}(\C)$, it is clear that ${\bf u}{\bf h}^\dagger=0$. 
Therefore, $\mathbf{u}(c\mathbf{u} + \mathbf{h})^\dagger=0$ for any $c \in \mathbb{F}_{q^2}$, which implies that $\mathbf{u}$ is Hermitian orthogonal to every vector $\mathbf{x} \in \C^{\perp_{\rm H}}$. Since $\mathbf{u} \in \Hull_{\rm H}(\C)$ obviously contradicts $\mathbf{u} \in \C^{\perp_{\rm H}} \setminus \Hull_{\rm H}(\C)$, we conclude that $\mathbf{u}\mathbf{u}^\dagger \neq 0$.

\item Since $\ell < n-k-1$, we have $\Delta = n - k - \ell \geq 2$. 
Hence, there exist at least two vectors 
$\mathbf{v}_1, \mathbf{v}_2 \in \C^{\perp_{\rm H}} \setminus \Hull_{\rm H}(\C)$ such that $\mathbf{v}_1, \mathbf{v}_2$, and $\Hull_{\rm H}(\C)$ are linearly independent. We consider a linear combination $\mathbf{u} = x\mathbf{v}_1 + y\mathbf{v}_2$, with $x, y \in \mathbb{F}_{q^2}$, and confirm that 
\begin{align}\label{eq.uu111}
\mathbf{u}\mathbf{u}^\dagger 
&= (x\mathbf{v}_1 + y\mathbf{v}_2)(x\mathbf{v}_1 + y\mathbf{v}_2)^\dagger 
= x^{q+1} \, \mathbf{v}_1\mathbf{v}_1^\dagger + x \, y^q \, \mathbf{v}_1\mathbf{v}_2^\dagger + x^q \, y \,  \mathbf{v}_2\mathbf{v}_1^\dagger + y^{q+1} \, \mathbf{v}_2\mathbf{v}_2^\dagger.
\end{align}
Let $\alpha = \mathbf{v}_1\mathbf{v}_1^\dagger$, 
$\beta = \mathbf{v}_1\mathbf{v}_2^\dagger$, and 
$\gamma = \mathbf{v}_2\mathbf{v}_2^\dagger$. Hence, $\alpha, \gamma \in \mathbb{F}_q$ and $\beta \in \mathbb{F}_{q^2}$. Since $\mathbf{v}_2\mathbf{v}_1^\dagger = \beta^q$, we can use \eqref{eq.uu111} to infer that 
$\mathbf{u}\mathbf{u}^\dagger = \alpha \, x^{q+1} + \beta \,  x \, y^q + \beta^q \, x^q \, y + \gamma \, y^{q+1}$. 
To find a vector satisfying $\mathbf{u}\mathbf{u}^\dagger = 0$, it suffices to 
solve the following homogeneous polynomial equation of degree $q+1$.
\begin{align}\label{eq.homogeneous}
\alpha \, x^{q+1} + \beta \, x \, y^q + \beta^q \, x^q \, y + \gamma \, y^{q+1} = 0
\end{align}
in variables $x$ and $y$ over $\mathbb{F}_{q^2}$, with the evaluation being restricted to the subfield $\mathbb{F}_q$. We verify that \eqref{eq.homogeneous} always admits a non-trivial solution 
$(x_0, y_0) \in (\F_{q^2}^*)^2$. By setting $(x, y)$ to be this non-trivial solution $(x_0, y_0)$, we immediately obtain a vector $\mathbf{u} = x_0 \mathbf{v}_1 + y_0 \mathbf{v}_2$ that satisfies $\mathbf{u}\mathbf{u}^\dagger = 0$. Furthermore, since $\mathbf{v}_1$ and $\mathbf{v}_2$ are linearly independent modulo $\Hull_{\rm H}(\C)$, 
their non-trivial linear combination cannot belong to $\Hull_{\rm H}(\C)$. This ensures that $\mathbf{u} \in \C^{\perp_{\rm H}} \setminus \Hull_{\rm H}(\C)$. 
\item By extending the direct sum decomposition in \eqref{eq.direct-sum1} to 
$\C^{\perp_{\rm H}} = \Hull_{\rm H}(\C) \oplus \text{span}\{\mathbf{v}_{\ell+1}, \mathbf{v}_{\ell+2}, \dots, \mathbf{v}_{n-k}\}$, 
the third assertion can be proved similarly by contradiction, where the 
$n-k-\ell$ vectors $\mathbf{v}_{\ell+1}, \mathbf{v}_{\ell+2}, \dots, \mathbf{v}_{n-k} \in \C^{\perp_{\rm H}} \setminus \Hull_{\rm H}(\C)$ 
are linearly independent. In short, if $\mathbf{w}\mathbf{w}^\dagger = 0$ holds for all $\mathbf{w} \in \C^{\perp_{\rm H}} \setminus \Hull_{\rm H}(\C)$, 
then $\mathbf{x} \mathbf{x}^\dagger = 0$ for all $\mathbf{x} \in \C^{\perp_{\rm H}}$. By Lemma \ref{lem.Hermitian}, the code $\C^{\perp_{\rm H}}$ is Hermitian self-orthogonal, that is, $\C^{\perp_{\rm H}} \subseteq \C$ and, hence, $\ell = n-k$, contradicting the hypothesis $\ell \leq n-k-1$. 
\end{enumerate}
The proof is now complete.
\end{IEEEproof}

Using Lemma \ref{lem.uu}, we can establish necessary and sufficient condition for the existence of extension vectors that either preserve or increase the Hermitian hull dimension.

\begin{theorem}\label{th:necessary_condition_hull}
If $\C$ is an $[n,k]_{q^2}$ linear code with $\dim(\Hull_{\rm H}(\C)) = \ell$, then the following assertions hold. 
\begin{enumerate}
\item There exists a vector $\mathbf{u}_1 \in \mathbb{F}_{q^2}^n \setminus \C$ 
and an element $a_1 \in \mathbb{F}_{q^2}^*$ such that $\dim(\Hull_{\rm H}(\C(\mathbf{u}_1, a_1))) = \ell$ 
if and only if $\ell < n-k$ and $(q, \ell)\neq (2,n-k-1)$.

\item There exists a vector $\mathbf{u}_2 \in \mathbb{F}_{q^2}^n \setminus \C$ 
and an element $a_2 \in \mathbb{F}_{q^2}^*$ such that $\dim(\Hull_{\rm H}(\C(\mathbf{u}_2, a_2))) = \ell+1$ 
if and only if $\ell < n-k$.  
\end{enumerate}
\end{theorem}
\begin{IEEEproof}
We prove the necessity and then the sufficiency, each time for both assertions.

For the \emph{necessary conditions}, let us assume the existence of $\mathbf{u}_i \in \mathbb{F}_{q^2}^n \setminus \C$ and $a_i \in \mathbb{F}_{q^2}^*$ 
such that $\dim(\Hull_{\rm H}(\C(\mathbf{u}_i, a_i))) \in \{\ell, \ell+1\}$ for $i \in \{1, 2\}$. In both assertions, the new Hermitian hull dimension is $\geq \ell$. For a contradiction, suppose that $\ell = n-k$, which implies that $\C^{\perp_{\rm H}} \subseteq \C$. For either statement under consideration, $\dim(\Hull_{\rm H}(\C(\mathbf{u}_i, a_i))) \geq \ell = n-k$. Since $\dim(\C(\mathbf{u}_i, a_i)^{\perp_{\rm H}})=n-k$ and $\Hull_{\rm H}(\C(\mathbf{u}_i, a_i)) \subseteq \C(\mathbf{u}_i, a_i)^{\perp_{\rm H}}$, then $\C(\mathbf{u}_i, a_i)^{\perp_{\rm H}} \subseteq \C(\mathbf{u}_i, a_i)$. We recall that any codeword in $\C(\mathbf{u}_i, a_i)^{\perp_{\rm H}}$ takes the form $(\mathbf{x}, \, -a_i^{-q} \, \mathbf{x} \, \mathbf{u}_i^\dagger)$, with $\mathbf{x} \in \C^{\perp_{\rm H}}$. Since $\C(\mathbf{u}_i, a_i)^{\perp_{\rm H}} \subseteq \C(\mathbf{u}_i, a_i)$, the vector $(\mathbf{x}, \, -a_i^{-q}\mathbf{x}\mathbf{u}_i^\dagger)$ can also be generated by $\C(\mathbf{u}_i, a_i)$. In other words, there exist $\mathbf{c} \in \C$ and $b \in \mathbb{F}_{q^2}$ such that 
\[
\mathbf{x} = \mathbf{c} + b \, \mathbf{u}_i \mbox{ and } 
-a_i^{-q} \, \mathbf{x} \, \mathbf{u}_i^\dagger = a_i \, b.
\]
Since $\mathbf{x} \in \C^{\perp_{\rm H}} \subseteq \C$ and $\mathbf{c} \in \C$, we infer that $b \, \mathbf{u}_i = \mathbf{x} - \mathbf{c} \in \C$. However, $\mathbf{u}_i \in \mathbb{F}_{q^2}^n \setminus \C$ implies $b = 0$. Consequently, the last coordinate is always $-a_i^{-q} \, \mathbf{x} \, \mathbf{u}_i^\dagger=0 \, a_i = 0$, leading to $\mathbf{x} \, \mathbf{u}_i^\dagger = 0$ for all $\mathbf{x} \in \C^{\perp_{\rm H}}$. Hence, $\mathbf{u}_i \in (\C^{\perp_{\rm H}})^{\perp_{\rm H}} = \C$, which is a contradiction. Thus, $\ell \leq n-k-1$, which means that $\ell < n-k$. We are done with the necessary condition of $\ell < n-k$ for both statements.

To establish the necessity of the condition $(q, \ell) \neq (2, n-k-1)$ in the first assertion, we assume towards a contradiction that $\ell < n-k$, $q = 2$, and $\ell = n-k-1$. Let there be $\mathbf{u}_1 \in \mathbb{F}_{q^2}^n \setminus \C$ and $a_1 \in \mathbb{F}_{q^2}^*$ such that $\dim(\Hull_{\rm H}(\C(\mathbf{u}_1, a_1))) = \ell$. Based on Statement \ref{state.1} and Theorem \ref{th.hull222}, we can further assume, without loss of generality, that 
$\mathbf{u}_1 \in \C^{\perp_{\rm H}} \setminus \Hull_{\rm H}(\C)$. 
Since $\ell = n-k-1$, we have $\dim(\C^{\perp_{\rm H}}) - \dim(\Hull_{\rm H}(\C)) = 1$. It then follows from Lemma \ref{lem.uu} Item 1) that $0\neq \mathbf{u}_1\mathbf{u}_1^\dagger \in \F_2^*$ for any ${\bf u}_1\in \C^{\perp_{\rm H}} \setminus \Hull_{\rm H}(\C)$. 
This forces $\mathbf{u}_1\mathbf{u}_1^\dagger = 1$. For any $a_1 \in \mathbb{F}_4^*$, however, its $(q+1)^{\rm th}$ power is uniquely determined as $a_1^{2+1} = a_1^3 = 1$. Hence, $\mathbf{u}_1\mathbf{u}_1^\dagger + a_1^{q+1} = 1 + 1 = 0$ always holds over $\mathbb{F}_2$. By Theorem \ref{th.hull222}, we get $\dim(\Hull_{\rm H}(\C(\mathbf{u}_1, a_1))) = \ell+1$ for any ${\bf u}_1\in \C^{\perp_{\rm H}}\setminus \Hull_{\rm H}(\C)$ and $a_1\in \F_4^*$, 
which is a contradiction. This completes the justification that the condition $(q, \ell) \neq (2, n-k-1)$ is necessary.

On the \emph{sufficient conditions}, we first work on the condition for the second assertion before moving on to the first assertion. For the second assertion, since $\ell < n-k$, we have $\dim(\C^{\perp_{\rm H}}) - \dim(\Hull_{\rm H}(\C)) = n - k - \ell \geq 1$. By Lemma \ref{lem.uu} Item 3), there must exist a vector $\mathbf{u}_2 \in \C^{\perp_{\rm H}} \setminus \Hull_{\rm H}(\C)\subseteq \F_{q^2}^n\setminus \C$ 
such that $\mathbf{u}_2\mathbf{u}_2^\dagger \neq 0$, which implies $-\mathbf{u}_2\mathbf{u}_2^\dagger \in \mathbb{F}_q^*$. For any finite field $\mathbb{F}_{q^2}$, the $(q+1)^{\rm th}$ power map from $\mathbb{F}_{q^2}^*$ to $\mathbb{F}_q^*$ is surjective. Hence, there exists at least one element $a_2 \in \mathbb{F}_{q^2}^*$ such that $(a_2)^{q+1} = -\mathbf{u}_2 \mathbf{u}_2^\dagger$, which means that $\mathbf{u}_2\mathbf{u}_2^\dagger + (a_2)^{q+1} = 0$. Based on Theorem \ref{th.hull222}, we get $\dim(\Hull_{\rm H}(\C(\mathbf{u}_2,a_2))) = \ell+1$.

For the first assertion, we assume $\ell < n-k$ and $(q, \ell) \neq (2, n-k-1)$. We select $\mathbf{u}_1 \in (\C^{\perp_{\rm H}} \setminus \Hull_{\rm H}(\C)) \subseteq (\F_{q^2}^n\setminus \C)$ and $a_1 \in \mathbb{F}_{q^2}^*$ such that $\mathbf{u}_1\mathbf{u}_1^\dagger + (a_1)^{q+1} \neq 0$. We continue our analysis by considering two cases, namely when $q>2$ and $q=2$.

\begin{enumerate}[wide=0pt, labelindent=0pt,label=\textbf{Case \arabic*}:]
\item $q > 2$. We choose the same vector $\mathbf{u}_1 = \mathbf{u}_2 \in \C^{\perp_{\rm H}} \setminus \Hull_{\rm H}(\C)$ as in the proof of the second assertion earlier, where $\mathbf{u}_1\mathbf{u}_1^\dagger \neq 0$. Since the subfield $\mathbb{F}_q^*$ contains $q-1 \geq 2$ nonzero elements, and the $(q+1)^{\rm th}$ power map is surjective, there exists at least one element $a_1 \in \mathbb{F}_{q^2}^*$ such that $(a_1)^{q+1} \neq -\mathbf{u}_1 \mathbf{u}_1^\dagger$, ensuring $\mathbf{u}_1\mathbf{u}_1^\dagger + a_1^{q+1} \neq 0$.

\item $q = 2$ and $n - k - \ell \geq 2$. Since $\dim(\C^{\perp_{\rm H}}) - \dim(\Hull_{\rm H}(\C)) \geq 2$, we know from Lemma \ref{lem.uu} Item 2) that 
there exists a vector $\mathbf{u}_1 \in \C^{\perp_{\rm H}}\setminus \Hull_{\rm H}(\C)$ such that $\mathbf{u}_1\mathbf{u}_1^\dagger = 0$. For any $a_1 \in \mathbb{F}_4^*$, its $(q+1)^{\rm th}$ power is uniquely determined as $(a_1)^{2+1} = (a_1)^3 = 1$. Consequently, we obtain $\mathbf{u}_1\mathbf{u}_1^\dagger + (a_1)^{q+1} = 0 + 1 = 1 \neq 0$.
\end{enumerate}
In both cases, we have successfully found $\mathbf{u}_1 \in \C^{\perp_{\rm H}} \setminus \Hull_{\rm H}(\C)$ and $a_1 \in \mathbb{F}_{q^2}^*$ satisfying $\mathbf{u}_1\mathbf{u}_1^\dagger + (a_1)^{q+1} \neq 0$. Thus, by Theorem \ref{th.hull222}, we conclude that $\dim(\Hull_{\rm H}(\C(\mathbf{u}_1,a_1))) = \ell$.
\end{IEEEproof}

Applying the sequence of extensions $\C_i = \C_{i-1}({\bf u}_i, a_i)$ with $i\geq 0$  from Theorem \ref{th.hull222} Item 1) does \emph{not} yield a Hermitian self-orthogonal code for any choice of ${\bf u}_i$ and $a_i$, unless the initial code $\C_0=\C$ with $\dim(\Hull_{\rm H}(\C))=\ell$ is already Hermitian self-orthogonal. Indeed, in each iterative step, the code dimension always increases by $1$ while the Hermitian hull dimension increases by at most $1$. Thus, the dimensional gap 
\begin{align}\label{eq.gap}
    \Delta_i = \dim(\C_i) - \dim(\Hull_{\rm H}(\C_i))
\end{align}
is a \emph{non-decreasing sequence} throughout the process. If $\Delta_0 > 0$ initially, then after $m$ iterations we have $\Delta_m \geq \Delta_0 > 0$. 
Hence, it is impossible to achieve the Hermitian self-orthogonal condition $\Delta_m = 0$ through this iterative construction. 

Conversely, by meticulously controlling the gap sequence $\Delta_i$ in \eqref{eq.gap}, there are precise iterative pathways to construct $[n+m, k+m]_{q^2}$ linear codes with a $(k+m-1)$-dimensional Hermitian hull.  
Such codes have the second largest possible dimension of the Hermitian hull and some references, \textit{e.g.}, \cite{PZYG2023,Chen2024,Cao2021,LZM2023}, call them {\em codes with large Hermitian hulls} or {\em almost Hermitian self-orthogonal codes}. 

We formalize these iterative structural conditions as a corollary. 

\begin{corollary}\label{coro.so}
Let $\C_0$ be an $[n,k]_{q^2}$ code. If $\C_i = \C_{i-1}({\bf u}_i, a_i)$ for $1 \leq i \leq m$, where ${\bf u}_i \in \mathbb{F}_{q^2}^{n+i-1} \setminus \C_{i-1}$ and $a_i \in \mathbb{F}_{q^2}^*$, form a sequence of iteratively generalized extended codes, then $\C_m$ is an $[n+m, k+m]_{q^2}$ code and the following statements hold.
\begin{enumerate}
\item The code $\C_m$ is Hermitian self-orthogonal if and only if $\C_0$ is Hermitian self-orthogonal \emph{and} the vector ${\bf u}_i \in \C_{i-1}^{\perp_{\rm H}} \setminus \C_{i-1}$ is such that $(a_i)^{q+1} = -{\bf u}_i{\bf u}_i^\dagger$ for all $1 \leq i \leq m$.
        
\item The code $\C_m$ has a $(k+m-1)$-dimensional Hermitian hull if and only if ${\bf u}_i \in \C_{i-1}^{\perp_{\rm H}} \setminus \Hull_{\rm H}(\C_{i-1})$ for all $1 \leq i \leq m$  and \emph{exactly one} of the following conditions holds.
\begin{itemize}
\item [\rmnum{1})] The code $\C_0$ is Hermitian self-orthogonal and there exists exactly one index $j \in \{1, 2, \dots, m\}$ such that $(a_j)^{q+1} \neq -{\bf u}_j{\bf u}_j^\dagger$ while $(a_i)^{q+1} = -{\bf u}_i{\bf u}_i^\dagger$ for all $i \neq j$.
\item [\rmnum{2})] The code $\C_0$ is such that $\dim(\Hull_{\rm H}(\C_0)) = k-1$ and $(a_i)^{q+1} = -{\bf u}_i{\bf u}_i^\dagger$ for all $1 \leq i \leq m$.
\end{itemize}
\end{enumerate}
\end{corollary}
\begin{IEEEproof}
Let $\Delta_i = \dim(\C_i) - \dim(\Hull_{\rm H}(\C_i))$ denote the dimensional gap at the $i^{\rm th}$ step. Since $\dim(\C_i) = \dim(\C_{i-1}) + 1$, the increment $\Delta_i - \Delta_{i-1}$ depends only on the variation of the Hermitian hull dimension. By Statement~\ref{state.1}, whenever an extension vector lies in $(\C_{i-1}+\C_{i-1}^{\perp_{\rm H}})\setminus\C_{i-1}$, we may replace it by a representative in $\C_{i-1}^{\perp_{\rm H}}\setminus\Hull_{\rm H}(\C_{i-1})$ without changing the resulting generalized extended code. Three rules of increment can therefore be stated for these representatives.
\begin{enumerate}[wide=0pt, labelindent=0pt,label=Rule \arabic*:]
\item $\Delta_i - \Delta_{i-1} = 0$ if and only if ${\bf u}_i \in \C_{i-1}^{\perp_{\rm H}} \setminus \Hull_{\rm H}(\C_{i-1})$ and $(a_i)^{q+1} = -{\bf u}_i{\bf u}_i^\dagger$.
\item $\Delta_i - \Delta_{i-1} = 1$ if and only if ${\bf u}_i \in \C_{i-1}^{\perp_{\rm H}} \setminus \Hull_{\rm H}(\C_{i-1})$ and $(a_i)^{q+1} \neq -{\bf u}_i{\bf u}_i^\dagger$.
\item For any other possible choices of ${\bf u}_i\in \mathbb{F}_{q^2}^{n+i-1} \setminus \C_{i-1}$, the dimension $\dim(\Hull_{\rm H}(\C_i))$ either remains unchanged or decreases, leading to $\Delta_i - \Delta_{i-1} \geq 2$.
\end{enumerate} 
Thus, the sequence $(\Delta_0, \Delta_1, \dots, \Delta_m)$ is non-decreasing. We now proceed to justify the statements listed above. 
\begin{enumerate}
    \item The code $\C_m$ is Hermitian self-orthogonal if and only if $\Delta_m = 0$. Since $\Delta_i$ is a non-decreasing sequence of non-negative integers, the condition $\Delta_m = 0$ is equivalent to $\Delta_0 = 0$ and $\Delta_i -\Delta_{i-1} = 0$ for all $1 \leq i \leq m$. On the one hand, the initial condition $\Delta_0 = 0$ holds if and only if $\C_0$ is Hermitian self-orthogonal. On the other hand, since $\C_{i-1}$ inherits the Hermitian self-orthogonality, the condition $\Delta_i - \Delta_{i-1} = 0$ requires ${\bf u}_i \in \C_{i-1}^{\perp_{\rm H}} \setminus \Hull_{\rm H}(\C_{i-1})= \C_{i-1}^{\perp_{\rm H}} \setminus \C_{i-1}$ and $(a_i)^{q+1} = -{\bf u}_i{\bf u}_i^\dagger$ for all $1 \leq i \leq m$.
    \item The code $\C_m$ has a $(k+m-1)$-dimensional Hermitian hull if and only if $\Delta_m = 1$, which is further equivalent to exactly one of the following two cases.
    \begin{enumerate}[wide=0pt, labelindent=0pt,label=\textbf{Case \arabic*}:]
    \item $\Delta_0 = 0$. The code $\C_0$ is Hermitian self-orthogonal and $\Delta_j - \Delta_{j-1} = 1$ for exactly one index $j$, with $\Delta_i - \Delta_{i-1} = 0$ for all $i \neq j$. By above increment rules, this claim holds if and only if $(a_j)^{q+1} \neq -{\bf u}_j{\bf u}_j^\dagger$ for exactly one index $j\in \{1,2,\ldots,m\}$ and $(a_i)^{q+1} = -{\bf u}_i{\bf u}_i^\dagger$ for all $1\leq i \leq m$ such that $i \neq j$. 
    \item $\Delta_0 = 1$. Hence, $\dim(\Hull_{\rm H}(\C_0)) = k-1$ and $\Delta_i - \Delta_{i-1} = 0$ for all $1 \leq i \leq m$. This holds if and only if $(a_i)^{q+1} = -{\bf u}_i{\bf u}_i^\dagger$ for all $1 \leq i \leq m$. 
    \end{enumerate}
    In addition, to ensure that $\Delta_i - \Delta_{i-1} \leq 1$ for all $1 \leq i \leq m$, the added vectors satisfy ${\bf u}_i \in \C_{i-1}^{\perp_{\rm H}} \setminus \Hull_{\rm H}(\C_{i-1})$ in both cases. The second statement is, thus, justified.
\end{enumerate}
\end{IEEEproof}

\section{Hermitian dual distances of generalized extended codes}\label{sec.improvment}

This section studies the Hermitian dual distances of generalized extended codes
$\C({\bf u},a)$. In Subsection~\ref{subsec.distance1}, we characterize the
Hermitian dual distance of $\C({\bf u},a)$ and give explicit conditions under
which it increases. Some of these conditions are necessary and sufficient.
We derive conditions that increase the Hermitian dual distance and the Hermitian hull dimension simultaneously in Subsection~\ref{subsec.distance3}.

\subsection{Characterization and Control of the Hermitian Dual Distances}\label{subsec.distance1}

We keep the notation $\mathbf{W}_{\min}(\C^{\perp_{\rm H}})=\{{\bf e}\in \C^{\perp_{\rm H}}:~\wt({\bf e})=d^{\perp_{\rm H}}\}$ and give an easy characterization on the minimum distance of $\C({\bf u},a)^{\perp_{\rm H}}$ 
based on the minimum distance of $\C^{\perp_{\rm H}}$. 

\begin{theorem}\label{th.distance}
Let $\C$ be an $[n,k,d]_{q^2}$ linear code with Hermitian dual distance $d^{\perp_{\rm H}}$.  
For any $\mathbf{u}\in \mathbb{F}_{q^2}^n\setminus \C$ and $a\in \mathbb{F}_{q^2}^*$, we have 
\begin{align*}
d\bigl(\C({\bf u},a)^{\perp_{\rm H}}\bigr)=
\begin{cases}
d^{\perp_{\rm H}}, &  \mbox{if }{\bf e}{\bf u}^\dagger=0~\text{for~some~} {\bf e}\in \mathbf{W}_{\min}(\C^{\perp_{\rm H}}), \\
d^{\perp_{\rm H}}+1, & \mbox{if } {\bf e}{\bf u}^\dagger\neq 0~\text{for~any~} {\bf e}\in \mathbf{W}_{\min}(\C^{\perp_{\rm H}}).
\end{cases}
\end{align*}  
\end{theorem}
\begin{IEEEproof}
We use \eqref{eq.generator-parity-matrix} to derive 
\[
\C({\bf u},a)^{\perp_{\rm H}}=\{{\bf e}_{({\bf u},a)}=({\bf e}, -a^{-q}{\bf e}{\bf u}^\dagger):~{\bf e}\in \C^{\perp_{\rm H}} \},  
\]
which implies that 
\begin{align*}
\wt({\bf e}_{({\bf u},a)}) &=
\begin{cases}
\wt({\bf e}), & \mbox{if } {\bf e}{\bf u}^\dagger=0, \\
\wt({\bf e})+1, & \mbox{if } {\bf e}{\bf u}^\dagger \neq 0,
\end{cases}\mbox{ and}\\
d\bigl(\C({\bf u},a)^{\perp_{\rm H}}\bigr) &=
\begin{cases}
d^{\perp_{\rm H}}, & \mbox{if there exists } {\bf e}\in \mathbf{W}_{\min}(\C^{\perp_{\rm H}}) \mbox{ such that } {\bf e}{\bf u}^\dagger=0, \\
d^{\perp_{\rm H}}+1 & \mbox{for all } {\bf e}\in \mathbf{W}_{\min}(\C^{\perp_{\rm H}}) \mbox{ for which } {\bf e}{\bf u}^\dagger\neq 0.
\end{cases}
\end{align*}  
This completes the proof. 
\end{IEEEproof}

\begin{remark}\label{rem.distance}
Before Statement~\ref{state.1}, we have written ${\bf u}={\bf c}_i+{\bf e}_i$ with ${\bf c}_i\in \C$ and ${\bf e}_i\in \C^{\perp_{\rm H}}$ 
when ${\bf u}\in\C+\C^{\perp_{\rm H}}$. For any ${\bf e}\in \mathbf{W}_{\min}(\C^{\perp_{\rm H}})$, we have
\[
{\bf e}{\bf u}^\dagger
=
{\bf e}({\bf c}_i+{\bf e}_i)^\dagger
=
{\bf e}{\bf c}_i^\dagger+{\bf e}{\bf e}_i^\dagger
=
{\bf e}{\bf e}_i^\dagger.
\]
Hence, the condition ${\bf e}{\bf u}^\dagger=0$ in Theorem~\ref{th.distance} 
depends only on the component ${\bf e}_i \in \C^{\perp_{\rm H}}$ of ${\bf u}$. 
Furthermore, if ${\bf e}\in \Hull_{\rm H}(\C)$, then ${\bf e}\in \C$ and, therefore, the choice of ${\bf e}$ does not produce a new nontrivial representative modulo $\C$. This fact reinforces our motivation for simplifying ${\bf u}\in\C+\C^{\perp_{\rm H}}$ to 
${\bf u}\in\C^{\perp_{\rm H}}\setminus\Hull_{\rm H}(\C)$ in Statement~\ref{state.1}.
\end{remark}

The next two theorems give necessary and sufficient conditions for the
Hermitian dual distance of $\C({\bf u},a)$ to increase by $1$ compared with
that of $\C$. Combined with Statement \ref{state.1} and Remark \ref{rem.distance}, these conditions are stated with respect to the two search spaces $\F_{q^2}^n\setminus\C$ and
$\C^{\perp_{\rm H}}\setminus\Hull_{\rm H}(\C)$, respectively. By
Theorem~\ref{th.distance}, the Hermitian dual distance of $\C({\bf u},a)$
is either $d^{\perp_{\rm H}}$ or $d^{\perp_{\rm H}}+1$. Hence, by taking the
negations of the following two criteria, we immediately obtain necessary and
sufficient conditions for $d\bigl((\C({\bf u},a))^{\perp_{\rm H}}\bigr) =d^{\perp_{\rm H}}$.

\begin{theorem}\label{th.subcode_characterization}
Let $\C$ be an $[n,k]_{q^2}$ linear code with Hermitian dual distance
$d^{\perp_{\rm H}}$. For any $a\in \F_{q^2}^*$, 
there is a vector $\mathbf{u}\in \F_{q^2}^n\setminus \C$ such that 
$d\bigl((\C(\mathbf{u},a))^{\perp_{\rm H}}\bigr)=d^{\perp_{\rm H}}+1$ 
if and only if 
there exists a maximal subcode $S\subsetneq \C^{\perp_{\rm H}}$
such that
\[
S\cap \mathbf{W}_{\min}(\C^{\perp_{\rm H}})=\emptyset.
\]
\end{theorem}
\begin{IEEEproof} 
By Theorem~\ref{th.distance}, the condition
$
d\bigl((\C(\mathbf{u},a))^{\perp_{\rm H}}\bigr)=d^{\perp_{\rm H}}+1
$
is equivalent to
$
\mathbf{e}\mathbf{u}^\dagger\neq 0
$
for all $\mathbf{e}\in \mathbf{W}_{\min}(\C^{\perp_{\rm H}})$. We define the set 
\[
S_{\mathbf{u}}
:=
\{\mathbf{x}\in \C^{\perp_{\rm H}}:~\mathbf{x}\mathbf{u}^\dagger=0\}.
\]
Since $
\mathbf{u}\notin \C=(\C^{\perp_{\rm H}})^{\perp_{\rm H}}$, there exists ${\bf x}\in \C^{\perp_{\rm H}}$ such that $\mathbf{x}\mathbf{u}^\dagger\neq 0$. 
Applying Lemma~\ref{lem.maximal_subcode_general} to the code $\C^{\perp_{\rm H}}$, we confirm that $S_{\bf u}$ is a maximal subcode of $\C^{\perp_{\rm H}}$. The above non-orthogonality condition $\mathbf{e}\mathbf{u}^\dagger\neq 0$ is equivalent to $S_{\mathbf{u}}\cap \mathbf{W}_{\min}(\C^{\perp_{\rm H}})=\emptyset$. 

Conversely, let us assume that there exists a maximal subcode $ S\subsetneq \C^{\perp_{\rm H}}$ such that $S \cap \mathbf{W}_{\min}(\C^{\perp_{\rm H}})=\emptyset$. Since $S$ has codimension $1$ in $\C^{\perp_{\rm H}}$, there exists a nonzero linear functional $f$ from $\C^{\perp_{\rm H}}$ to $\F_{q^2}$ with $\ker(f)=S$. Now, for each $\mathbf{u}\in \F_{q^2}^n$, we define the map
\[
\phi:\C^{\perp_{\rm H}} \to \F_{q^2}\mbox{, sending } \mathbf{x} \mapsto \mathbf{x} \mathbf{u}^\dagger.
\]
It is straightforward to verify that $\phi$ is an $\F_{q^2}$-linear functional on $\C^{\perp_{\rm H}}$. Hence, we obtain a linear map
\[
\Phi:~\F_{q^2}^n \to \Hom(\C^{\perp_{\rm H}},\F_{q^2}) \mbox{, sending }
\mathbf{u} \mapsto \phi.
\]
Since the Hermitian inner product is nondegenerate, 
\[
\ker(\Phi)
=
\{\mathbf{u}\in \F_{q^2}^n:~\phi=0\}
=
(\C^{\perp_{\rm H}})^{\perp_{\rm H}}
=
\C.
\]
By the first isomorphism theorem, we get $ \F_{q^2}^n/\C \cong \operatorname{Im}(\Phi)$. Due to the fact that  
\[
\dim(\F_{q^2}^n/\C)=n-k
=\dim(\C^{\perp_{\rm H}})=
\dim\bigl(\Hom(\C^{\perp_{\rm H}},\F_{q^2})\bigr),
\]
we have $\operatorname{Im}(\Phi) = \Hom(\C^{\perp_{\rm H}},\F_{q^2})$ and, hence, $\F_{q^2}^n/\C \cong \Hom(\C^{\perp_{\rm H}},\F_{q^2})$. We can also verify the existence of $\mathbf{u}\in \F_{q^2}^n$ such that
$f=\phi$, that is,
\[
f(\mathbf{x})=\phi({\bf x})=\mathbf{x}\mathbf{u}^\dagger \mbox{ for every }
\mathbf{x}\in \C^{\perp_{\rm H}}.
\]
Since $f \neq 0$, we have $\phi\neq 0$ and $\mathbf{u}\notin \ker(\Phi)=\C$.
Thus, we have found a vector ${\bf u}\in \F_{q^2}^n\setminus \C$ such that 
\begin{align*}
S=\ker(f)=\{\mathbf{x}\in \C^{\perp_{\rm H}} \, : \, \mathbf{x}\mathbf{u}^\dagger=0\}.    
\end{align*}
By the assumption that $S\cap \mathbf{W}_{\min}(\C^{\perp_{\rm H}}) = \emptyset$, we get 
$\mathbf{e}\mathbf{u}^\dagger\neq 0$ for each $\mathbf{e}\in \mathbf{W}_{\min}(\C^{\perp_{\rm H}})$. To complete the proof, we use Theorem~\ref{th.distance} to establish the desired conclusion that $d\bigl((\C(\mathbf{u},a))^{\perp_{\rm H}}\bigr) = d^{\perp_{\rm H}}+1$ for every $a\in \F_{q^2}^*$. 
\end{IEEEproof}

\begin{theorem}\label{th:dualspace_characterization}
Let $\C$ be an $[n,k]_{q^2}$ linear code with Hermitian dual distance
$d^{\perp_{\rm H}}$. For any $a\in \F_{q^2}^*$, there is a vector $\mathbf{u}\in \C^{\perp_{\rm H}}\setminus \Hull_{\rm H}(\C)$ such that 
$d\bigl((\C(\mathbf{u},a))^{\perp_{\rm H}}\bigr)=d^{\perp_{\rm H}}+1$ 
if and only if there exists a maximal subcode $S\subsetneq \C^{\perp_{\rm H}}$
such that
\[
\Hull_{\rm H}(\C)\subseteq S
\mbox{ and } S\cap \mathbf{W}_{\min}(\C^{\perp_{\rm H}})=\emptyset.
\]
\end{theorem}

\begin{IEEEproof}
We assume that there exists
$\mathbf{u}\in \C^{\perp_{\rm H}}\setminus \Hull_{\rm H}(\C)$
such that $d\bigl((\C(\mathbf{u},a))^{\perp_{\rm H}}\bigr)=d^{\perp_{\rm H}}+1$. Let  
\[
S_{\mathbf{u}}
:=
\{\mathbf{x}\in \C^{\perp_{\rm H}} \, : \, 
\mathbf{x}\mathbf{u}^\dagger=0\}.
\]
Since $\mathbf{u}\notin \Hull_{\rm H}(\C)\subseteq \C$, it is clear that 
$\mathbf{u}\notin \C$. Hence, there exists ${\bf x}\in \C^{\perp_{\rm H}}$ such that $\mathbf{x}\mathbf{u}^\dagger\neq 0$. By Lemma~\ref{lem.maximal_subcode_general}, we know that $S_{\bf u}$ is a maximal subcode of $\C^{\perp_{\rm H}}$ and $S_{\mathbf{u}}\cap \mathbf{W}_{\min}(\C^{\perp_{\rm H}})=\emptyset$. For any $\mathbf{h}\in \Hull_{\rm H}(\C)$, since $\mathbf{h}\in \Hull_{\rm H}(\C)\subseteq \C$ and $\mathbf{u} \in \C^{\perp_{\rm H}}$, we have $\mathbf{h}\mathbf{u}^\dagger=0$, which implies that $\mathbf{h}\in S_{\mathbf{u}}$. Therefore, $\Hull_{\rm H}(\C)\subseteq S_{\mathbf{u}}$

Conversely, we start with a maximal subcode
$S\subsetneq \C^{\perp_{\rm H}}$ such that $\Hull_{\rm H}(\C)\subseteq S$ and 
$S\cap \mathbf{W}_{\min}(\C^{\perp_{\rm H}})=\emptyset$. Similar to the second part of the proof of Theorem~\ref{th.subcode_characterization}, it suffices to find a linear functional $f\in \Hom(\C^{\perp_{\rm H}}, \F_{q^2})$ and 
a vector ${\bf u}\in \C^{\perp_{\rm H}}\setminus \Hull_{\rm H}(\C)$ such that 
\begin{align*}
   S=\ker(f)=\{\mathbf{x}\in \C^{\perp_{\rm H}} \, : \, \mathbf{x}\mathbf{u}^\dagger=0\}.
\end{align*}
Since $\Hull_{\rm H}(\C) = \Hull_{\rm H}(\C^{\perp_{\rm H}}) = \C^{\perp_{\rm H}}\cap (\C^{\perp_{\rm H}})^{\perp_{\rm H}}$, we can choose a basis
$\{\mathbf{h}_1,\mathbf{h}_2,\ldots,\mathbf{h}_\ell,\mathbf{v}_{\ell+1},\mathbf{v}_{\ell+2},\ldots,\mathbf{v}_{n-k}\}$ of $\C^{\perp_{\rm H}}$ such that 
\begin{align}\label{eq.hull_depo111}
\C^{\perp_{\rm H}}=\Hull_{\rm H}(\C^{\perp_{\rm H}})\oplus {\rm span}\{\mathbf{v}_{\ell+1},\mathbf{v}_{\ell+2},\ldots,\mathbf{v}_{n-k}\} 
= \Hull_{\rm H}(\C)\oplus {\rm span}\{\mathbf{v}_{\ell+1},\mathbf{v}_{\ell+2},\ldots,\mathbf{v}_{n-k}\}    
\end{align}
and, for all admissible indices $i$ and $j$,
\begin{align}\label{eq.hull_depo222}
  \mathbf{h}_i\mathbf{h}_j^\dagger=0, ~
\mathbf{h}_i\mathbf{v}_j^\dagger=0 \mbox{, and } 
\mathbf{v}_i\mathbf{v}_j^\dagger = \delta_{ij},  
\end{align}
where $0\leq \ell=\dim(\Hull_{\rm H}(\C))\leq \min\{k,n-k\}$ and $\delta_{ij}$ is the Kronecker delta. Since $S$ is a maximal subcode of $\C^{\perp_{\rm H}}$ containing $\Hull_{\rm H}(\C)$, there is a nonzero linear functional
$f\in \Hom(\C^{\perp_{\rm H}},\F_{q^2})$ such that $S=\ker(f)$. Since $\Hull_{\rm H}(\C)\subseteq S=\ker(f)$, we have
\begin{align}\label{eq.hull_depo333}
  f(\mathbf{h}_i)=0 \mbox{ for } 1\leq i\leq \ell.  
\end{align}
We write
\begin{align}\label{eq.hull_depo444}
f(\mathbf{v}_j)=b_j \mbox{ for } \ell+1\leq j\leq n-k.
\end{align}
Since $f\neq 0$, not all $b_j$ are zero. Letting
\begin{align}\label{eq.uuu}
    \mathbf{u}=\sum_{j=\ell+1}^{n-k} b_j^{q} \, \mathbf{v}_j\in 
{\rm span}\{\mathbf{v}_{\ell+1},\mathbf{v}_{\ell+2},\ldots,\mathbf{v}_{n-k}\}\setminus \{{\bf 0}\} 
\subseteq \C^{\perp_{\rm H}}\setminus \{{\bf 0}\},
\end{align}
we infer that ${\bf u}\notin \Hull_{\rm H}(\C)$ 
and, hence, ${\bf u} \in \C^{\perp_{\rm H}}\setminus \Hull_{\rm H}(\C)$. 
Otherwise,  
\[
{\bf u}\in \Hull_{\rm H}(\C) \, \cap \,  {\rm span}\{\mathbf{v}_{\ell+1},\mathbf{v}_{\ell+2},\ldots,\mathbf{v}_{n-k}\}.
\]
By \eqref{eq.hull_depo111}, we get ${\bf u}={\bf 0}$, which is a contradiction. Moreover, for any
\[
\mathbf{x}=\sum_{i=1}^{\ell}x_i\mathbf{h}_i+\sum_{j=\ell+1}^{n-k}x_j\mathbf{v}_j
\in \C^{\perp_{\rm H}},
\]
it follows from \eqref{eq.hull_depo222}, \eqref{eq.hull_depo333}, and \eqref{eq.hull_depo444} that 
\[
\mathbf{x}\mathbf{u}^\dagger
=
\sum_{j=\ell+1}^{n-k} x_j \, b_j
=\sum_{i=1}^{\ell} x_i \, f(\mathbf{h}_i) + \sum_{j=\ell+1}^{n-k} x_j \, f(\mathbf{v}_j)= f(\mathbf{x}).
\]
Thus, $S=\ker(f) =\{\mathbf{x}\in \C^{\perp_{\rm H}} \, : \, \mathbf{x}\mathbf{u}^\dagger=0\}$ for the vector $\mathbf{u}\in \C^{\perp_{\rm H}}\setminus \Hull_{\rm H}(\C)$ given in \eqref{eq.uuu}.
\end{IEEEproof}

The following is a direct corollary to Theorem \ref{th:dualspace_characterization}. 

\begin{corollary}\label{cor:degeneration_hull}
Let $\C$ be an $[n,k]_{q^2}$ linear code with Hermitian dual distance
$d^{\perp_{\rm H}}$. For any $a\in \F_{q^2}^*$, 
if there is a vector $\mathbf{u}\in \C^{\perp_{\rm H}}\setminus \Hull_{\rm H}(\C)$ such that 
$d\bigl((\C(\mathbf{u},a))^{\perp_{\rm H}}\bigr)=d^{\perp_{\rm H}}+1$, 
then $\Hull_{\rm H}(\C)\cap
\mathbf{W}_{\min}(\C^{\perp_{\rm H}}) = \emptyset$.
\end{corollary}

The next theorem gives a uniform sufficient condition for the existence of an
extension vector that increases the Hermitian dual distance. Since our proof uses a projective-geometric covering argument, we state it for $\dim(\C^{\perp_{\rm H}})\geq 3$. 

\begin{theorem}\label{th.sufficient_increment}
Let $\C$ be an $[n,k]_{q^2}$ linear code with Hermitian dual distance
$d^{\perp_{\rm H}}$ and $\dim(\C^{\perp_{\rm H}})\geq 3.$
If the two conditions, namely,  
\begin{enumerate}
    \item [\rmnum{1})] $\mathbf{W}_{\min}(\C^{\perp_{\rm H}})\cup\{\mathbf{0}\}$ contains no $2$-dimensional subspace and  
    \item [\rmnum{2})] $A_{d^{\perp_{\rm H}}}\leq (q^2-1)(q^2+q)$
\end{enumerate}  
are both met, then there is a vector $\mathbf{u}\in \F_{q^2}^n\setminus \C$ such that $d\bigl((\C(\mathbf{u},a))^{\perp_{\rm H}}\bigr)=d^{\perp_{\rm H}} +1$ for every $a\in \F_{q^2}^*$.
\end{theorem}
\begin{IEEEproof}
By Theorem~\ref{th.subcode_characterization}, it suffices to show that there
exists a maximal subcode $S\subsetneq \C^{\perp_{\rm H}}$ 
such that $S\cap \mathbf{W}_{\min}(\C^{\perp_{\rm H}})=\emptyset$. 
Assume, to the contrary, that every maximal subcode of $\C^{\perp_{\rm H}}$
contains at least one minimum weight codeword of $\C^{\perp_{\rm H}}$.
Let $\mathcal{P}_{\min} := \{\langle \mathbf{e}\rangle:~\mathbf{e}\in \mathbf{W}_{\min}(\C^{\perp_{\rm H}})\}$, 
where $\langle \mathbf{e}\rangle$ denotes the $1$-dimensional subspace generated by $\mathbf{e}$. Since scalar multiplication does not change the Hamming weight, each element of $\mathcal{P}_{\min}$ contributes exactly $q^2-1$ nonzero minimum weight codewords. Therefore,
\begin{align}\label{eq.Pmin}
|\mathcal{P}_{\min}|=\frac{A_{d^{\perp_{\rm H}}}}{q^2-1}\leq q^2+q.    
\end{align}

The $1$-dimensional subspaces of $\C^{\perp_{\rm H}}$ are precisely the points of the projective space $PG(\dim(\C^{\perp_{\rm H}})-1,q^2)=PG(n-k-1,q^2)$ with $n-k-1\geq 2$. The maximal subcodes of $\C^{\perp_{\rm H}}$ are precisely its hyperplanes. By the assumption, every maximal subcode of $\C^{\perp_{\rm H}}$ contains at least one minimum weight codeword and, therefore, it contains at least one point of $\mathcal{P}_{\min}$. This is equivalent to stating that $\mathcal{P}_{\min}$ meets every hyperplane of $PG(n-k-1,q^2)$. Since a projective line in $PG(n-k-1,q^2)$ corresponds to the set of all $1$-dimensional subspaces which are contained in a $2$-dimensional subspace of $\C^{\perp_{\rm H}}$ and since $\mathbf{W}_{\min}(\C^{\perp_{\rm H}})\cup\{\mathbf{0}\}$ contains no $2$-dimensional subspace, the set $\mathcal{P}_{\min}$ does not contain all points of any projective line.

In summary, we regard $\mathcal{P}_{\min}$ as a point set in $PG(\dim(\C^{\perp_{\rm H}})-1,q^2)$ that meets every hyperplane but does not contain all points of any projective line. By Bruen Theorem for projective planes \cite{Bruen1971} and its generalization in higher-dimensional spaces \cite{Ebert1978}, we can conclude that $|\mathcal{P}_{\min}|\geq q^2+q+1$, which contradicts \eqref{eq.Pmin}. This completes the proof.
\end{IEEEproof}

\begin{remark}\label{rem:feasibility_conditions}
At first glance, the structural and numerical prerequisites in Theorem \ref{th.sufficient_increment} seem restrictive. A deeper algebraic and combinatorial analysis, however, reveals that these conditions are in fact remarkably mild and naturally met by a vast majority of linear codes. 
We highlight this feasibility by stating two theoretical insights.
\begin{enumerate}
\item The first condition is equivalent to saying that $\C^{\perp_{\rm H}}$ contains no $2$-dimensional constant-weight subcode whose nonzero codewords all have weight $d^{\perp_{\rm H}}$. We know from \cite{B1984} that every $q^2$-ary constant-weight linear code is, up to monomial equivalence, a repetition of a simplex code, possibly with additional zero coordinates. Furthermore, every $2$-dimensional $q^2$-ary constant-weight code has nonzero weight divisible by $q^2$. Thus, the first condition in the theorem is automatically satisfied if $d^{\perp_{\rm H}}\not\equiv 0 \pmod{q^2}$.

\item If $\mathbf{W}_{\min}(\C^{\perp_{\rm H}})\cup\{\mathbf{0}\}$ contains at least one $2$-dimensional subspace, then there are at least $q^4-1$ nonzero vectors with Hamming weights $d^{\perp_{\rm H}}$ in $\mathbf{W}_{\min}(\C^{\perp_{\rm H}})$. Hence, $A_{d^{\perp_{\rm H}}}\geq q^4-1$. Since $(q^2-1)(q^2+q)>q^4-1$ for any $q \geq 2$, both conditions that Theorem~\ref{th.sufficient_increment} requires are met if $A_{d^{\perp_{\rm H}}}< q^4-1$. 
\end{enumerate}
\end{remark}

\begin{remark}\label{rem:sufficient_increment}
The assumption $\dim(\C^{\perp_{\rm H}})\geq 3$ in Theorem~\ref{th.sufficient_increment} excludes only low-dimensional
dual spaces. When $\dim(\C^{\perp_{\rm H}})=1$, the available nonzero
directions are too limited for the above projective covering argument. When $\dim(\C^{\perp_{\rm H}})=2$, the behavior depends on the concrete
distribution of the minimum weight codewords in $\C^{\perp_{\rm H}}$. These
exceptional cases can be treated separately, checked directly, and are not needed for the uniform criterion in Theorem~\ref{th.sufficient_increment}.

On the other hand, Theorem~\ref{th.sufficient_increment} gives a uniform sufficient condition valid for all
$\dim(\C^{\perp_{\rm H}})\geq 3$.
The above bound is sharp in the plane case $\dim(\C^{\perp_{\rm H}})=3$ \cite{Bruen1971},
where the associated projective space is $PG(2,q^2)$. 
When $\dim(\C^{\perp_{\rm H}})>3$, stronger lower bounds for $\mathcal{P}_{\min}$ are available \cite{Ebert1978}. Consequently, the numerical condition $A_{d^{\perp_{\rm H}}}\leq (q^2-1)(q^2+q)$ in Theorem~\ref{th.sufficient_increment} can be replaced by weaker bounds that lead to sharper sufficient criteria. We keep the present formulation because it is dimension-free and easy to apply.
\end{remark}

\subsection{Increasing Hermitian Dual Distances and Hull Dimensions Simultaneously}\label{subsec.distance3}

To design EAQECCs with better parameters, we focus on the case where 
\[
d(\C({\bf u},a)^{\perp_{\rm H}})=d(\C^{\perp_{\rm H}})+1 \mbox{ and }
\dim(\Hull_{\rm H}(\C(\mathbf{u}, a)))=\dim(\Hull_{\rm H}(\C))+1.
\]
This allows us to obtain higher error handling capabilities while requiring fewer number of entangled pairs. Based on Statement \ref{state.1}, we only need to select ${\bf u}$ from $\C^{\perp_{\rm H}} \setminus \Hull_{\rm H}(\C)$ properly. To do this, we use Theorem \ref{th.hull222} Part 1), Theorem \ref{th:dualspace_characterization}, and two subsets of $\C^{\perp_{\rm H}}$, namely,  
\begin{itemize}
    \item the set $\mathcal{Z}_{\rm H}=\{\mathbf x\in \C^{\perp_{\rm H}} \, : \, \mathbf x\mathbf x^\dagger=0\}\subseteq \C^{\perp_{\rm H}}$ 
    of self-orthogonal codewords in $\C^{\perp_{\rm H}}$ and
    \item the maximal subcode $S_{{\bf e}}=\{\mathbf x\in \C^{\perp_{\rm H}} \, : \, \mathbf x{\bf e}^\dagger=0\}$ of $\C^{\perp_{\rm H}}$ which is orthogonal to ${\bf e}\in \mathbf W_{\min}(\C^{\perp_{\rm H}})$.
\end{itemize} 

\begin{theorem}\label{th.increasing_both}
Let $\C$ be an $[n,k]_{q^2}$ linear code with $\dim(\Hull_{\rm H}(\C))=\ell$ 
and Hermitian dual distance $d^{\perp_{\rm H}}$. For a fixed $\mathbf{u}\in \C^{\perp_{\rm H}}\setminus \Hull_{\rm H}(\C)$, there exists $a\in \F_{q^2}^*$ such that $\dim(\Hull_{\rm H}(\C(\mathbf{u}, a)))=\ell+1$ and $d(\C(\mathbf{u},a)^{\perp_{\rm H}})=d^{\perp_{\rm H}}+1$ if and only if 
\begin{align}\label{eq.u_space}
    {\bf u}\notin \mathcal{Z}_{\rm H} \bigcup 
    \left(\bigcup_{{\bf e}\in \mathbf W_{\min}(\C^{\perp_{\rm H}})} 
    S_{{\bf e}}\right). 
\end{align}
In particular, the existence of such a vector ${\bf u}$ implies 
\begin{equation}\label{eq:vectoru}
\C^{\perp_{\rm H}}\not\subseteq \mathcal Z_{\rm H} \bigcup 
\left(\bigcup_{{\bf e}\in \mathbf W_{\min}(\C^{\perp_{\rm H}})} 
S_{{\bf e}}\right).
\end{equation}
\end{theorem}
\begin{IEEEproof} 
Let ${\bf u}\in \C^{\perp_{\rm H}}\setminus \Hull_{\rm H}(\C)$ and $a\in \F_{q^2}^*$ such that $\dim(\Hull_{\rm H}(\C({\bf u},a))) = \ell+1$ and $d(\C({\bf u},a)^{\perp_{\rm H}})=d^{\perp_{\rm H}}+1$ be given.
By Theorem~\ref{th.distance}, Statement~\ref{state.1}, and Theorem~\ref{th.hull222}, we have ${\bf u}{\bf e}^\dagger\neq 0$ 
for all ${\bf e}\in \mathbf W_{\min}(\C^{\perp_{\rm H}})$ and ${\bf u}{\bf u}^\dagger+a^{q+1}=0$. Hence, ${\bf u}\notin S_{{\bf e}}$ for all ${\bf e}\in \mathbf W_{\min}(\C^{\perp_{\rm H}})$ and ${\bf u}\notin \mathcal{Z}_{\rm H}$, yielding \eqref{eq.u_space}. Moreover, since $\mathbf{u}\in \C^{\perp_{\rm H}}$, we arrive at \eqref{eq:vectoru}.

Conversely, let $\mathbf{u}\in \C^{\perp_{\rm H}}\setminus \Hull_{\rm H}(\C)$ be such that ${\bf u}\notin \mathcal{Z}_{\rm H}$ and 
${\bf u}\notin S_{{\bf e}}$ for all ${\bf e}\in \mathbf W_{\min}(\C^{\perp_{\rm H}})$. Since ${\bf u}\notin S_{{\bf e}}$, we know that ${\bf u}{\bf e}^\dagger\neq 0$ for all ${\bf e}\in \mathbf W_{\min}(\C^{\perp_{\rm H}})$. Applying Theorem~\ref{th.distance} results in $d(\C({\bf u},a)^{\perp_{\rm H}})=d^{\perp_{\rm H}}+1$.
In addition, ${\bf u}\notin \mathcal{Z}_{\rm H}$ implies that ${\bf u}{\bf u}^\dagger\in \F_q^*$. We now define the norm map 
\[
{\rm Norm}:~\F_{q^2}^* \to \F_q^* \mbox{, sending }
a \mapsto a^{q+1}.
\]
Since ${\rm Norm}(\cdot)$ is surjective, there is always an element $a\in \F_{q^2}^*$ such that 
$a^{q+1}=-{\bf u}{\bf u}^\dagger$. By Theorem~\ref{th.hull222}, we conclude that 
\[
\dim(\Hull_{\rm H}(\C({\bf u},a)))=\ell+1.
\]
\end{IEEEproof}

\begin{corollary}\label{cor:constructive_choice_u}
Let $\C$ be an $[n,k]_{q^2}$ linear code with $\dim(\Hull_{\rm H}(\C))=\ell$ 
and Hermitian dual distance $d^{\perp_{\rm H}}$. 
For any vector 
\begin{align}\label{eq.u_space222}
{\bf u}\in
\C^{\perp_{\rm H}}
\setminus
\left(
\mathcal Z_{\rm H}
\bigcup
 \left(\bigcup_{{\bf e}\in \mathbf W_{\min}(\C^{\perp_{\rm H}})} S_{\bf e}
\right) \right) 
\subseteq \C^{\perp_{\rm H}}\setminus \Hull_{\rm H}(\C),
\end{align}
there exists an element $a\in \F_{q^2}^*$ such that
$\dim(\Hull_{\rm H}(\C({\bf u},a)))=\ell+1$ and $d(\C({\bf u},a)^{\perp_{\rm H}})=d^{\perp_{\rm H}}+1$.
\end{corollary}
\begin{IEEEproof}
Since every codeword in $\Hull_{\rm H}(\C)$ is self-orthogonal, $\Hull_{\rm H}(\C)\subseteq \mathcal{Z}_{\rm H}$, which implies that 
\[
\C^{\perp_{\rm H}} \setminus
\left(\mathcal Z_{\rm H} \bigcup
\left(\bigcup_{{\bf e} \in \mathbf W_{\min}(\C^{\perp_{\rm H}})} S_{\bf e}
\right) \right) 
\subseteq \C^{\perp_{\rm H}}\setminus \Hull_{\rm H}(\C).
\]
By \eqref{eq.u_space222}, we get  
\[
{\bf u}\in \C^{\perp_{\rm H}}\setminus \Hull_{\rm H}(\C)~{\rm and}~
{\bf u}\notin \mathcal Z_{\rm H} \bigcup 
\left(\bigcup_{{\bf e}\in \mathbf W_{\min}(\C^{\perp_{\rm H}})} S_{\bf e}\right).
\]
The desired result then follows from Theorem \ref{th.increasing_both}. 
\end{IEEEproof}

\begin{remark} 
According to Corollary \ref{cor:degeneration_hull} and Theorem \ref{th:necessary_condition_hull}, 
it is possible to find an appropriate vector ${\bf u}\in \C^{\perp_{\rm H}}\setminus \Hull_{\rm H}(\C)$ and 
an element $a\in \F_{q^2}^*$ such that 
$d(\C({\bf u},a)^{\perp_{\rm H}})=d(\C^{\perp_{\rm H}})+1$ and 
$\dim(\Hull_{\rm H}(\C(\mathbf{u}, a)))\in \{\dim(\Hull_{\rm H}(\C)),~ \dim(\Hull_{\rm H}(\C))+1\}$ 
only if 
\begin{align}\label{eq.necessary_conditions}
    \mathbf W_{\min}(\C^{\perp_{\rm H}})\cap \Hull_{\rm H}(\C)=\emptyset~{\rm and}~\dim(\Hull_{\rm H}(\C))\leq n-k-1.
\end{align}
The two conditions in \eqref{eq.necessary_conditions} are satisfied 
if \eqref{eq.u_space} holds. This is why we omit them from Theorem \ref{th.increasing_both}. 
\end{remark}

We end this section with an example illustrating the operational meaning of
Theorem~\ref{th.increasing_both} in the classical setting. Even when the original code is not optimal, the generalized extension can produce
optimal linear codes with larger Hermitian hulls.

\begin{example}\label{ex.classical}
Let $\mathbb{F}_4 = \{0, 1, \omega, \omega^2\}$ with $\omega^2 + \omega + 1 = 0$. Let $\C$ be the $[6, 3, 3]_4$ NMDS code with generator matrix 
\begin{equation*}
G = \begin{pmatrix} 
1 & 0 & 0 & \omega^2 & \omega & 0 \\ 
0 & 1 & 0 & \omega^2 & \omega^2 & 1 \\ 
0 & 0 & 1 & \omega^2 & 1 & \omega^2 
\end{pmatrix}.
\end{equation*}
The code has $1$-dimensional Hermitian hull. Since there is a $[6,3,4]_4$ MDS code on record \cite{G2026-linear}, the $[6, 3, 3]_4$ NMDS code $\C^{\perp_{\rm H}}$ is \emph{not} optimal. Let $\C({\bf u}, 1)$ be the generalized extended code of $\C$. Fixing $a=1$, we have Table \ref{tab:extension_results}. 

\begin{table}[htbp]
\caption{Vectors that increase the Hermitian hull dimension or Hermitian dual distance of the $[6,3,3]_{4}$ NMDS code in Example \ref{ex.classical}}
\label{tab:extension_results}
\centering
\begin{threeparttable}
\renewcommand{\arraystretch}{1.1}
\resizebox{\linewidth}{!}{
\begin{tabular}{l | c|c|c|c|c|c}
\toprule
Increase in & $\mathbf{u}$ & $\dim(\Hull_{\rm H}(\C({\bf u}, 1)))$ & $d^{\perp_{\rm H}}(\C({\bf u}, 1))$ & ${\bf u}\notin \mathcal{Z}_H$ & ${\bf u}\notin \bigcup_{{\bf e}\in \mathbf W_{\min}(C^{\perp_H})} S_{{\bf e}}$ & Optimal\tnote{1} \\
\midrule
Hull & $(1, 1, 1, 0, \omega, \omega)$ & 2 & 3 & Yes & No & No\\
Distance & $(1, 0, \omega^2, 0, \omega, \omega^2)$ & 1 & 4 & No & Yes & Yes \\
Both & $(0, 1, \omega, 0, 0, 1)$ & 2 & 4 & Yes & Yes & Yes \\
\bottomrule
\end{tabular}
}
\begin{tablenotes}
\footnotesize
\item[1] The answer to whether the Hermitian dual code
$(\C({\bf u},1))^{\perp_{\rm H}}$ with parameters $[7,3]_4$ has the optimal minimum distance $4$. 
\end{tablenotes}
\end{threeparttable}
\end{table}

The optimal $[7,3,4]_4$ linear code documented in Grassl's table \cite{G2026-linear} happens to be Hermitian LCD. Therefore, we start with a non-optimal code $\C$ and finally obtain two optimal linear codes  
$\C({\bf u}, 1)^{\perp_{\rm H}}$ with larger Hermitian hulls. 
In addition, although the first generalized extended code 
$(\C((1, 1, 1, 0, \omega, \omega),1))^{\perp_{\rm H}}$ in Table \ref{tab:extension_results} is not optimal compared to \cite{G2026-linear}, 
it has a larger Hermitian hull dimension than that of $\C$. 
\end{example}

\section{Applications in EAQECCs}\label{sec.application}

Luo, Ezerman, Grassl, and Ling gave three propagation rules in \cite{LEGL2023} and applied them to optimal and best-known linear codes in \cite{G2026-linear}.
Together with the EAQECCs obtained in \cite{G2021-PRA}, Grassl's online
tables \cite{G-binary-2025,G-ternary-2025} collect best-known qubit and
qutrit codes of lengths up to $64$ and $36$, respectively. Motivated by
these tables, several sporadic improvements have been obtained by different
methods in, {\it e.g.}, \cite{LLS2025,K2023,FCLLC2025,LLZS2024,LZ2024-DM,CLZ2025,LZ2024-JAMC}.

In this section, we use the results in Sections \ref{sec.hull} and
\ref{sec.improvment} to construct EAQECCs from generalized extended codes.
{\em The goal is to obtain quantum codes whose parameters improve on those in \cite{G-binary-2025,G-ternary-2025,LLS2025,K2023,FCLLC2025,LLZS2024,LZ2024-DM,CLZ2025,LZ2024-JAMC}.} 

Given an $[n,k]_{q^2}$ linear code $\C$, Lemma
\ref{lem.Hermitian-Construction} yields an EAQECC $\mathcal{Q}(\C)$ with parameters 
\begin{align}\label{eq.EAQECC_parameters}
[[n,\kappa:=n-k-\dim(\Hull_{\rm H}(\C)),\delta\geq d(\C^{\perp_{\rm H}});
c:=k-\dim(\Hull_{\rm H}(\C))]]_q.
\end{align}
We call $\mathcal{Q}(\C)$ {\em pure} if $\delta=d(\C^{\perp_{\rm H}})$ 
and {\em impure} if $\delta>d(\C^{\perp_{\rm H}})$. Since $(\C^{\perp_{\rm H}})^{\perp_{\rm H}}=\C$ and
$\Hull_{\rm H}(\C^{\perp_{\rm H}})=\Hull_{\rm H}(\C)$, one can also 
obtain a $\mathcal{Q}(\C^{\perp_{\rm H}})$ by applying Lemma \ref{lem.Hermitian-Construction} to $\C^{\perp_{\rm H}}$. 
Regarding $\C^{\perp_{\rm H}}$ as a new linear code, we can apply the results in Sections \ref{sec.hull} and \ref{sec.improvment} to $\C^{\perp_{\rm H}}$ and then apply Lemma \ref{lem.Hermitian-Construction} to the resulting codes. 
Here, we keep our focus on the quantum codes obtained from the generalized extended code $\C({\bf u},a)$.

\begin{theorem}\label{th:application-search-principle}
Let $\C$ be an $[n,k]_{q^2}$ linear code with
$\ell=\dim(\Hull_{\rm H}(\C))$ and
$d^{\perp_{\rm H}}=d(\C^{\perp_{\rm H}})$. For
${\bf u}\in\F_{q^2}^n\setminus \C$ and $a\in\F_{q^2}^*$, the generalized
extended code $\C({\bf u},a)$ is an $[n+1,k+1]_{q^2}$ linear code and
\[
d(\C({\bf u},a)^{\perp_{\rm H}})
=
\begin{cases}
d^{\perp_{\rm H}},&
\text{if }{\bf e}{\bf u}^{\dagger}=0
\text{ for some }
{\bf e}\in\mathbf W_{\min}(\C^{\perp_{\rm H}}),\\
d^{\perp_{\rm H}}+1,&
\text{otherwise}.
\end{cases}
\]
The value of $\dim(\Hull_{\rm H}(\C({\bf u},a)))$, which can be $\ell-1$, $\ell$, or $\ell+1$, is determined according to Theorems \ref{th.hull} and \ref{th.hull222}. Lemma \ref{lem.Hermitian-Construction} gives an EAQECC with parameters
\[
[[n+1,n-k-\dim(\Hull_{\rm H}(\C({\bf u},a))),\delta\geq d(\C({\bf u},a)^{\perp_{\rm H}});
k+1-\dim(\Hull_{\rm H}(\C({\bf u},a)))]]_q.
\]
\end{theorem}
\begin{IEEEproof}
The length and dimension of $\C({\bf u},a)$ follow from Lemma
\ref{lem.generator_matrix}. The Hermitian dual distance is
from Theorem \ref{th.distance}. The statement on the Hermitian hull dimension
follows from Theorems \ref{th.hull} and \ref{th.hull222}. The final parameter
conversion is done according to Lemma \ref{lem.Hermitian-Construction}.
\end{IEEEproof}

We now record the corresponding parameter changes. 
Suppose that the $[n,k]_{q^2}$ initial code $\C$ has $\ell$-dimensional Hermitian hull 
and Hermitian dual distance $d^{\perp_{\rm H}}$. Let 
\[
\Delta_{\Hull_{\rm H}}
:=\dim(\Hull_{\rm H}(\C({\bf u},a)))-\ell \mbox{ and } 
\Delta_{d^{\perp_{\rm H}}}
:= d(\C({\bf u},a)^{\perp_{\rm H}})-d^{\perp_{\rm H}} .
\]
Let $\mathsf H({\bf u})$ stands for the proposition: 
${\bf e}{\bf u}^{\dagger}\neq0$ for all
${\bf e}\in\mathbf W_{\min}(\C^{\perp_{\rm H}})$. By Theorem
\ref{th.distance}, $\mathsf H({\bf u})$ is necessary and sufficient for
$d(\C({\bf u},a)^{\perp_{\rm H}})=d^{\perp_{\rm H}}+1$. Otherwise, 
$d(\C({\bf u},a)^{\perp_{\rm H}})=d^{\perp_{\rm H}}$. Using Theorems \ref{th.hull} and \ref{th.hull222} as well as Statement \ref{state.1}, we deduce that 
\[
\Delta_{\Hull_{\rm H}}=
\begin{cases}
-1, & \mbox{if }{\bf u}\notin\C+\C^{\perp_{\rm H}},\\
0,& \mbox{if }{\bf u}\in\C^{\perp_{\rm H}}\setminus\Hull_{\rm H}(\C) \mbox{ such that } {\bf u}{\bf u}^{\dagger}+a^{q+1}\neq0,\\
1,& \mbox{if } {\bf u}\in\C^{\perp_{\rm H}}\setminus\Hull_{\rm H}(\C) \mbox{ such that } {\bf u}{\bf u}^{\dagger}+a^{q+1}=0.
\end{cases}
\]
Applying  Lemma \ref{lem.Hermitian-Construction} to the results in Sections \ref{sec.hull} and \ref{sec.improvment}, we immediately get six families of EAQECCs from $\C({\bf u},a)$. Table \ref{tab:parameter-transformations} lists their parameters, with {\em the initial $[[n,\kappa,\delta;c]]_q$ EAQECCs obtained from $\C$ are assumed to be pure, for simplicity}. In the rare occasions that the initial EAQECCs are impure, one only needs to replace
$\delta$ by $d(\C^{\perp_{\rm H}})$ in the displayed lower bounds accordingly.

\begin{table}[htbp]
\centering
\caption{\protect\mbox{Six EAQECC families from an $[n+1,k+1]_{q^2}$ generalized extended code $\C({\bf u},a)$}}
\label{tab:parameter-transformations}
\renewcommand{\arraystretch}{1.1}
\resizebox{\linewidth}{!}{%
\begin{tabular}{c|l|c|l|l}
\toprule
Family
& Condition on $({\bf u},a)$
& $(\Delta_{\Hull_{\rm H}},\Delta_{d^{\perp_{\rm H}}})$
& Parameters of $\mathcal{Q}(\C({\bf u},a))$
& References\\
\midrule
$1$
& ${\bf u}\notin \left(\C+\C^{\perp_{\rm H}}\right)$ and $\neg\mathsf H({\bf u})$
& $(-1,0)$
& $[[n+1,\kappa+1,\ge \delta;c+2]]_q$
& Thms. \ref{th.hull} and \ref{th.distance}\\
$2$
& ${\bf u}\notin \left(\C+\C^{\perp_{\rm H}}\right)$ and $\mathsf H({\bf u})$
& $(-1,1)$
& $[[n+1,\kappa+1,\ge \delta+1;c+2]]_q$
& Thms. \ref{th.hull} and \ref{th.distance}\\
$3$
& ${\bf u}\in\C^{\perp_{\rm H}}\setminus\Hull_{\rm H}(\C)$, ${\bf u}{\bf u}^{\dagger}+a^{q+1}\neq0$, $\neg\mathsf H({\bf u})$
& $(0,0)$
& $[[n+1,\kappa,\ge \delta;c+1]]_q$
& Statement \ref{state.1}, Thms. \ref{th.hull222} and \ref{th.distance}\\
$4$
& ${\bf u}\in\C^{\perp_{\rm H}}\setminus\Hull_{\rm H}(\C)$, ${\bf u}{\bf u}^{\dagger}+a^{q+1}\neq0$, $\mathsf H({\bf u})$
& $(0,1)$
& $[[n+1,\kappa,\ge \delta+1;c+1]]_q$
& Statement \ref{state.1}, Thms. \ref{th.hull222} and \ref{th.distance}\\
$5$
& ${\bf u}\in\C^{\perp_{\rm H}}\setminus\Hull_{\rm H}(\C)$, ${\bf u}{\bf u}^{\dagger}+a^{q+1}=0$, $\neg\mathsf H({\bf u})$
& $(1,0)$
& $[[n+1,\kappa-1,\ge \delta;c]]_q$
& Statement \ref{state.1}, Thms. \ref{th.hull222} and \ref{th.distance}\\
$6$
& ${\bf u}\in\C^{\perp_{\rm H}}\setminus\Hull_{\rm H}(\C)$, ${\bf u}{\bf u}^{\dagger}+a^{q+1}=0$, $\mathsf H({\bf u})$
& $(1,1)$
& $[[n+1,\kappa-1,\ge \delta+1;c]]_q$
& Statement \ref{state.1}, Thms. \ref{th.hull222} and \ref{th.distance}\\
\bottomrule
\end{tabular}%
}
\parbox{\linewidth}{\small
\smallskip
{\em Note:} For simplicity, the initial $[[n,\kappa,\delta;c]]_q$ EAQECC $\mathcal{Q}(\C)$ is assumed to be \emph{pure}, that is, $\kappa = n-k-\ell$, $c=k-\ell$, and $\delta=d(\C^{\perp_{\rm H}})$. If $\mathcal{Q}(C)$ is impure, then we can replace $\delta$ in the table by $d(\C^{\perp_{\rm H}})$, while keeping the dimension and the number of entangled pairs unchanged.}
\end{table}

\begin{remark}
    The role of the preceding theory is summarized in Table
\ref{tab:parameter-transformations}. In each listed search region, the
conditions are necessary and sufficient for the corresponding changes
$\Delta_{\Hull_{\rm H}}$ and $\Delta_{d^{\perp_{\rm H}}}$. 
Note that the other results presented in Sections \ref{sec.hull} and \ref{sec.improvment} can further narrow down the search space for 
or confirm the existence of ${\bf u}$ and $a$. For example, Theorem \ref{th:necessary_condition_hull} and Corollary \ref{cor:degeneration_hull} 
can be used to exclude the cases where the Hermitian hull dimension or the Hermitian dual distance cannot be increased, respectively. 
In particular, Corollary \ref{cor:constructive_choice_u} provides 
a constructive sufficient condition for increasing the Hermitian hull dimension and the Hermitian dual distance and, hence, yields the sixth family of EAQECCs in Table \ref{tab:parameter-transformations}. 
We present more implementation details of these results in Table \ref{tab:criteria-for-search}.
\end{remark}

\begin{table}[htbp]
\caption{Criteria that control parameter changes. NS, N, and S stand for ``necessary and sufficient condition'', ``necessary condition'', and ``sufficient condition'', respectively}
\label{tab:criteria-for-search}
\centering
\begin{threeparttable}
\renewcommand{\arraystretch}{1.1}
\resizebox{\linewidth}{!}{
\begin{tabular}{l|c|c|l}
\hline
Tool & Controlled Quantity & Condition
& \begin{tabular}[c]{@{}l@{}} Operational Meaning for Table 
\ref{tab:parameter-transformations}\end{tabular}\\
\hline
Lemma \ref{lem:column-space-characterization}
& $\Delta_{\Hull_{\rm H}}$
& NS
& Tests whether ${\bf u}$ lies in the sum space.\\

Lemma \ref{lem.maximal_subcode_general}
& $\Delta_{d^{\perp_{\rm H}}}$ 
& S
& Produces maximal subcodes used in dual-distance tests.\\

Lemma \ref{lem.uu}
& $\Delta_{\Hull_{\rm H}}$
& NS
& Helps choose $a$ for Families $3$ to $6$.\\

Theorem \ref{th:necessary_condition_hull}
& $\Delta_{\Hull_{\rm H}}$
& NS
& Decides whether Families $3$ to $6$ occur.\\

Corollary \ref{coro.so}
& $\Delta_{\Hull_{\rm H}}$
& NS
& Controls repeated use of Families $3$ to $6$.\\

Theorem \ref{th.subcode_characterization}
& $\Delta_{d^{\perp_{\rm H}}}$
& NS
& Determines the existence of families with
$\Delta_{d^{\perp_{\rm H}}}=1$ in the full space.\\

Theorem \ref{th:dualspace_characterization}
& $\Delta_{d^{\perp_{\rm H}}}$
& NS
& Determines existence of Families $4$ to $6$ in the reduced space.\\

Corollary \ref{cor:degeneration_hull}
& $\Delta_{d^{\perp_{\rm H}}}$
& N
& Excludes Families $4$ and $6$ when it fails.\\

Theorem \ref{th.sufficient_increment}
& $\Delta_{d^{\perp_{\rm H}}}$
& S
& Guarantees some family with
$\Delta_{d^{\perp_{\rm H}}}=1$.\\

Theorem \ref{th.increasing_both}
& $(\Delta_{\Hull_{\rm H}},\Delta_{d^{\perp_{\rm H}}})$
& NS
& Characterizes Family $6$, namely
$(\Delta_{\Hull_{\rm H}},\Delta_{d^{\perp_{\rm H}}})=(1,1)$.\\

Corollary \ref{cor:constructive_choice_u}
& $(\Delta_{\Hull_{\rm H}},\Delta_{d^{\perp_{\rm H}}})$
& S
& Gives an explicit search criterion for Family $6$.\\
\hline
\end{tabular}
}
\end{threeparttable}
\end{table}

Next, we apply Tables \ref{tab:parameter-transformations} and \ref{tab:criteria-for-search} 
to construct improved EA qubit and qutrit codes. 
We start two examples that illustrate how to apply the results 
in Sections \ref{sec.hull} and \ref{sec.improvment} to obtain improved EAQECCs.

\begin{example}\label{ex:binary-MI-example}
Let $\F_4=\{0,1,\omega,\omega^2\}$ with $\omega^2=\omega+1$. 
Let $\C_1$ be the $[9,7]_4$ linear code generated by
\[
G=
\begin{pmatrix}
1 & 0 & 0 & 0 & 0 & 0 & 0 & \omega^2 & \omega\\
0 & 1 & 0 & 0 & 0 & 0 & 0 & \omega^2 & \omega^2\\
0 & 0 & 1 & 0 & 0 & 0 & 0 & \omega^2 & \omega^2\\
0 & 0 & 0 & 1 & 0 & 0 & 0 & \omega^2 & \omega\\
0 & 0 & 0 & 0 & 1 & 0 & 0 & 0 & \omega\\
0 & 0 & 0 & 0 & 0 & 1 & 0 & 1 & 0\\
0 & 0 & 0 & 0 & 0 & 0 & 1 & \omega^2 & 1
\end{pmatrix}.
\]
We quickly confirm that $\dim(\Hull_{\rm H}(\C_1))=1$ and $d(\C_1^{\perp_{\rm H}})=7$. Upon selecting ${\bf u}_1=(\omega^2,0,\omega^2,\omega^2,1,1,\omega,1,1)$ and $a=1$, a direct computation yields $\dim(\Hull_{\rm H}(\C_1({\bf u}_1,1)))=0$ and $d(\C_1({\bf u}_1,1)^{\perp_{\rm H}})=7$. Thus, 
$(\Delta_{\Hull_{\rm H}},\Delta_{d^{\perp_{\rm H}}})=(-1,0)$, which corresponds
to Family $1$ in Table~\ref{tab:parameter-transformations}. 

Applying Lemma \ref{lem.Hermitian-Construction} to $\C_1({\bf u}_1,1)$ and computing the actual minimum distance yield $\mathcal{Q}(\C_1({\bf u}_1,1))$ with parameters $[[10,2,7;8]]_2$. Compared with the best-known $[[12,2, 7;9]]_2$ code in \cite{G-binary-2025}, our $\mathcal{Q}(\C_1({\bf u}_1,1))$ maintains the same dimension and minimum distance. It features a \emph{shorter length} and requires \emph{fewer pre-shared entangled pairs}. 
\end{example}

\begin{example}\label{ex:improved-eaqecc-small}
Let $\F_9=\{0,1,\omega,\omega^2,\omega^3,\omega^4,\omega^5,\omega^6,\omega^7\}$ 
with $\omega^2=\omega+1$. 
Let $\C_2$ be the $[7,6]_9$ linear code generated by
\[
G=
\begin{pmatrix}
1 & 0 & 0 & 0 & 0 & 0 & 1\\
0 & 1 & 0 & 0 & 0 & 0 & \omega^2\\
0 & 0 & 1 & 0 & 0 & 0 & \omega\\
0 & 0 & 0 & 1 & 0 & 0 & \omega^5\\
0 & 0 & 0 & 0 & 1 & 0 & \omega^2\\
0 & 0 & 0 & 0 & 0 & 1 & \omega
\end{pmatrix}.
\]
The code has $\dim(\Hull_{\rm H}(\C_2))=0$ and $d(\C_2^{\perp_{\rm H}})=7$. Choosing ${\bf u}_2=(\omega^7,\omega^5,\omega^2,\omega^6,\omega^5,\omega^2,\omega^3)$ and $a=1$ yields $\dim(\Hull_{\rm H}(\C_2({\bf u}_2,1)))=1$ and $d(\C_2({\bf u}_2,1)^{\perp_{\rm H}})=8$. Thus, 
$(\Delta_{\Hull_{\rm H}},\Delta_{d^{\perp_{\rm H}}})=(1,1)$, which corresponds
to Family $6$ in Table \ref{tab:parameter-transformations}.  

Applying Lemma \ref{lem.Hermitian-Construction} to $\C_2({\bf u}_2,1)$ produces $\mathcal{Q}(\C_2({\bf u}_2,1))$ with parameters $[[8,0,8;6]]_3$. Here, $\mathcal{Q}(\C_2({\bf u}_2,1))$ is pure since it has zero dimension. Compared with the best-known $[[9,0,8;6]]_3$ code in \cite{G-ternary-2025}, our $\mathcal{Q}(\C_2({\bf u}_2,1))$ has the same dimension, minimum distance, and number of entangled pairs, but has a \emph{shorter length}. 
\end{example}

In manners similar to Examples \ref{ex:binary-MI-example} and \ref{ex:improved-eaqecc-small}, we obtain numerous improved and new EAQECCs by applying the results in Sections \ref{sec.hull} and \ref{sec.improvment}. 
\begin{enumerate}
\item After comparing with the known parameters of qubit and qutrit codes in \cite{G-binary-2025,G-ternary-2025} and 
the latest records in \cite{LLS2025,K2023,FCLLC2025,LLZS2024,LZ2024-DM,CLZ2025,LZ2024-JAMC}, 
we present $281$ qubit and qutrit codes in the specified ranges of lengths, 
including $267$ qubit parameters for lengths up to $40$ and $14$ qutrit parameters for lengths up to $25$. Among them, $244$ parameter sets improve on the best-known records. The remaining $37$ are new in the sense that, to the best of our knowledge, they have not been reported before.

\item Tables \ref{tab:breakthroughs2140} and \ref{tab:breakthroughs0118} list our new and improved parameters, obtained from $\C({\bf u},a)$. The following abbreviations in the tables signify the exact parameter(s) in which our quantum codes are better than the best-known ones in \cite{G-binary-2025,G-ternary-2025,LLS2025,K2023,FCLLC2025,LLZS2024,LZ2024-DM,CLZ2025,LZ2024-JAMC}. The codes of type SL have \emph{shorter length}. Those of Type HD have \emph{higher dimension}. The Type LD is for codes with \emph{larger minimum distance}. If the codes require \emph{less number of pre-shared entangled pairs}, then they are of Type LE. The codes of Type MI have \emph{multiple simultaneous improvements}, having two or more of the above four improvement types. We use NP to indicate that the parameters are \emph{new}.

\item The ``Source'' column records the source of the corresponding known qubit or qutrit entanglement-assisted codes. The ``$\Delta_{\Hull_{\rm H}}$'' and ``$\Delta_{d^{\perp_{\rm H}}}$'' columns record, respectively, the changes in the Hermitian hull dimension and the Hermitian dual distance from $\C$ to the generalized extended code $\C({\bf u},a)$. To save space, we do not list the detailed information about the initial linear code $\C$, extension vector ${\bf u}$, 
and scalar $a$. Interested readers are invited to inspect our online database \cite{Github} for the information.
\end{enumerate}

\newcounter{ternarycount}
\begin{center}
\setcounter{ternarycount}{0}
\setlength{\tabcolsep}{6pt}
\begin{longtable}{@{}c@{}}
\caption{New and improved EA qubit codes of lengths $7 \leq n \leq 40$}
\label{tab:breakthroughs2140}\\
\endfirsthead

\caption{\mbox{New and improved EA qubit codes of lengths $7 \leq n\leq 40$ (cont.)}}\\
\endhead

\resizebox{\linewidth}{!}{%
\setlength{\tabcolsep}{3pt}%
\begin{tabular}{c|llc|cc|c||c|llc|cc|c}
\toprule
\mbox{No.} & \mbox{Our Qubit Code} & \mbox{Known Code} & \mbox{Source} & \mbox{$\Delta_{\Hull_{\rm H}}$} & \mbox{$\Delta_{d^{\perp_{\rm H}}}$} & \mbox{Type} & \mbox{No.} & \mbox{Our Qubit Code} & \mbox{Known Code} & \mbox{Source} & \mbox{$\Delta_{\Hull_{\rm H}}$} & \mbox{$\Delta_{d^{\perp_{\rm H}}}$} & \mbox{Type} \\
\midrule
\stepcounter{ternarycount}\theternarycount & $[[7,2,5;5]]_2$ & $[[8,1,5;5]]_2$ & \cite{LSL2023} & 0 & 1 & MI & \stepcounter{ternarycount}\theternarycount & $[[8,2,6;6]]_2$ & $[[10,2,6;6]]_2$ & \cite{LSL2023} & 0 & 1 & SL \\
\stepcounter{ternarycount}\theternarycount & $[[9,3,5;4]]_2$ & $[[9,2,5;4]]_2$ & \cite{G-binary-2025} & 0 & 1 & HD & \stepcounter{ternarycount}\theternarycount & $[[10,2,7;8]]_2$ & $[[12,2,7;9]]_2$ & \cite{G-binary-2025} & -1 & 0 & MI \\
\stepcounter{ternarycount}\theternarycount & $[[11,2,7;7]]_2$ & $[[12,2,7;9]]_2$ & \cite{G-binary-2025} & 1 & 1 & MI & \stepcounter{ternarycount}\theternarycount & $[[11,2,8;9]]_2$ & $[[12,2,7;9]]_2$ & \cite{G-binary-2025} & 0 & 1 & MI \\
\stepcounter{ternarycount}\theternarycount & $[[11,3,7;8]]_2$ & $[[12,3,7;9]]_2$ & \cite{G-binary-2025} & 0 & 1 & MI & \stepcounter{ternarycount}\theternarycount & $[[12,2,9;10]]_2$ & $-$ & $-$ & -1 & 1 & NP \\
\stepcounter{ternarycount}\theternarycount & $[[13,2,10;11]]_2$ & $[[14,2,10;12]]_2$ & \cite{FCLLC2025} & -1 & 1 & MI & \stepcounter{ternarycount}\theternarycount & $[[13,6,6;7]]_2$ & $[[13,3,6;8]]_2$ & \cite{K2023} & -1 & 0 & MI \\
\stepcounter{ternarycount}\theternarycount & $[[13,7,5;6]]_2$ & $[[15,7,5;6]]_2$ & \cite{G-binary-2025} & -1 & 1 & SL & \stepcounter{ternarycount}\theternarycount & $[[14,5,7;9]]_2$ & $[[15,5,7;10]]_2$ & \cite{G-binary-2025} & -1 & 1 & MI \\
\stepcounter{ternarycount}\theternarycount & $[[15,2,11;13]]_2$ & $-$ & $-$ & 0 & 1 & NP & \stepcounter{ternarycount}\theternarycount & $[[15,3,10;12]]_2$ & $-$ & $-$ & 0 & 1 & NP \\
\stepcounter{ternarycount}\theternarycount & $[[15,7,6;8]]_2$ & $[[15,6,6;9]]_2$ & \cite{G-binary-2025} & 0 & 0 & MI & \stepcounter{ternarycount}\theternarycount & $[[15,12,3;3]]_2$ & $[[16,12,3;3]]_2$ & \cite{G-binary-2025} & -1 & 1 & SL \\
\stepcounter{ternarycount}\theternarycount & $[[16,11,4;5]]_2$ & $[[36,11,4;25]]_2$ & \cite{ZK2026} & 0 & 1 & MI & \stepcounter{ternarycount}\theternarycount & $[[17,2,13;15]]_2$ & $-$ & $-$ & 0 & 1 & NP \\
\stepcounter{ternarycount}\theternarycount & $[[17,4,7;5]]_2$ & $[[17,4,7;8]]_2$ & \cite{G-binary-2025} & 1 & 1 & LE & \stepcounter{ternarycount}\theternarycount & $[[17,4,10;13]]_2$ & $-$ & $-$ & 0 & 1 & NP \\
\stepcounter{ternarycount}\theternarycount & $[[17,5,7;6]]_2$ & $[[17,5,7;8]]_2$ & \cite{G-binary-2025} & 0 & 1 & LE & \stepcounter{ternarycount}\theternarycount & $[[17,6,7;7]]_2$ & $[[17,5,7;8]]_2$ & \cite{G-binary-2025} & -1 & 1 & MI \\
\stepcounter{ternarycount}\theternarycount & $[[17,12,3;1]]_2$ & $[[17,12,3;2]]_2$ & \cite{G-binary-2025} & 0 & 1 & LE & \stepcounter{ternarycount}\theternarycount & $[[17,14,3;3]]_2$ & $[[18,14,3;3]]_2$ & \cite{G-binary-2025} & 0 & 1 & SL \\
\stepcounter{ternarycount}\theternarycount & $[[18,5,8;9]]_2$ & $[[20,5,8;9]]_2$ & \cite{G-binary-2025} & 1 & 1 & SL & \stepcounter{ternarycount}\theternarycount & $[[18,8,7;10]]_2$ & $[[36,8,4;28]]_2$ & \cite{ZK2026} & -1 & 1 & MI \\
\stepcounter{ternarycount}\theternarycount & $[[18,11,5;7]]_2$ & $[[36,11,4;25]]_2$ & \cite{ZK2026} & -1 & 1 & MI & \stepcounter{ternarycount}\theternarycount & $[[18,13,4;5]]_2$ & $[[36,12,4;24]]_2$ & \cite{ZK2026} & -1 & 1 & MI \\
\stepcounter{ternarycount}\theternarycount & $[[19,3,13;16]]_2$ & $-$ & $-$ & 0 & 1 & NP & \stepcounter{ternarycount}\theternarycount & $[[19,12,5;7]]_2$ & $[[36,12,4;24]]_2$ & \cite{ZK2026} & 0 & 1 & MI \\
\stepcounter{ternarycount}\theternarycount & $[[20,2,15;18]]_2$ & $[[21,2,15;19]]_2$ & \cite{FCLLC2025} & 0 & 1 & MI & \stepcounter{ternarycount}\theternarycount & $[[20,5,11;15]]_2$ & $-$ & $-$ & 0 & 1 & NP \\
\stepcounter{ternarycount}\theternarycount & $[[20,6,10;14]]_2$ & $[[36,6,6;30]]_2$ & \cite{ZK2026} & 0 & 1 & MI & \stepcounter{ternarycount}\theternarycount & $[[20,11,6;9]]_2$ & $[[36,11,4;25]]_2$ & \cite{ZK2026} & 0 & 1 & MI \\
\stepcounter{ternarycount}\theternarycount & $[[21,2,16;19]]_2$ & $[[21,2,15;19]]_2$ & \cite{FCLLC2025} & 0 & 1 & LD & \stepcounter{ternarycount}\theternarycount & $[[21,3,14;18]]_2$ & $-$ & $-$ & 0 & 1 & NP \\
\stepcounter{ternarycount}\theternarycount & $[[21,10,8;11]]_2$ & $[[26,10,8;16]]_2$ & \cite{G-binary-2025} & -1 & 1 & MI & \stepcounter{ternarycount}\theternarycount & $[[21,11,6;8]]_2$ & $[[36,11,4;25]]_2$ & \cite{ZK2026} & -1 & 1 & MI \\
\stepcounter{ternarycount}\theternarycount & $[[21,12,6;9]]_2$ & $[[36,12,4;24]]_2$ & \cite{ZK2026} & 0 & 1 & MI & \stepcounter{ternarycount}\theternarycount & $[[22,2,17;20]]_2$ & $-$ & $-$ & 0 & 1 & NP \\
\stepcounter{ternarycount}\theternarycount & $[[22,3,15;19]]_2$ & $-$ & $-$ & 0 & 1 & NP & \stepcounter{ternarycount}\theternarycount & $[[22,6,11;16]]_2$ & $[[29,6,11;16]]_2$ & \cite{G-binary-2025} & 0 & 1 & SL \\
\stepcounter{ternarycount}\theternarycount & $[[22,7,10;15]]_2$ & $[[28,7,10;16]]_2$ & \cite{G-binary-2025} & 0 & 1 & MI & \stepcounter{ternarycount}\theternarycount & $[[22,8,9;14]]_2$ & $[[27,8,9;16]]_2$ & \cite{G-binary-2025} & 0 & 1 & MI \\
\stepcounter{ternarycount}\theternarycount & $[[22,13,6;9]]_2$ & $[[36,12,4;24]]_2$ & \cite{ZK2026} & 0 & 1 & MI & \stepcounter{ternarycount}\theternarycount & $[[23,2,18;21]]_2$ & $[[24,2,18;22]]_2$ & \cite{FCLLC2025} & 0 & 1 & MI \\
\stepcounter{ternarycount}\theternarycount & $[[23,3,16;20]]_2$ & $-$ & $-$ & 0 & 1 & NP & \stepcounter{ternarycount}\theternarycount & $[[23,5,12;16]]_2$ & $[[30,5,12;16]]_2$ & \cite{G-binary-2025} & 0 & 1 & SL \\
\stepcounter{ternarycount}\theternarycount & $[[23,5,13;18]]_2$ & $-$ & $-$ & 0 & 1 & NP & \stepcounter{ternarycount}\theternarycount & $[[23,6,8;7]]_2$ & $[[23,6,8;9]]_2$ & \cite{G-binary-2025} & 0 & 0 & LE \\
\stepcounter{ternarycount}\theternarycount & $[[23,6,10;13]]_2$ & $[[36,6,6;30]]_2$ & \cite{ZK2026} & 1 & 1 & MI & \stepcounter{ternarycount}\theternarycount & $[[23,6,11;15]]_2$ & $[[29,6,11;16]]_2$ & \cite{G-binary-2025} & 1 & 1 & MI \\
\stepcounter{ternarycount}\theternarycount & $[[23,6,12;17]]_2$ & $[[36,6,6;30]]_2$ & \cite{ZK2026} & -1 & 1 & MI & \stepcounter{ternarycount}\theternarycount & $[[23,7,8;8]]_2$ & $[[23,7,8;9]]_2$ & \cite{G-binary-2025} & -1 & 0 & LE \\
\stepcounter{ternarycount}\theternarycount & $[[23,7,10;14]]_2$ & $[[28,7,10;16]]_2$ & \cite{G-binary-2025} & 0 & 1 & MI & \stepcounter{ternarycount}\theternarycount & $[[23,7,11;16]]_2$ & $[[28,7,10;16]]_2$ & \cite{G-binary-2025} & 0 & 1 & MI \\
\stepcounter{ternarycount}\theternarycount & $[[23,8,9;13]]_2$ & $[[27,8,9;16]]_2$ & \cite{G-binary-2025} & 1 & 1 & MI & \stepcounter{ternarycount}\theternarycount & $[[23,8,10;15]]_2$ & $[[27,8,9;16]]_2$ & \cite{G-binary-2025} & -1 & 1 & MI \\
\stepcounter{ternarycount}\theternarycount & $[[23,9,9;14]]_2$ & $[[27,9,9;16]]_2$ & \cite{G-binary-2025} & 0 & 1 & MI & \stepcounter{ternarycount}\theternarycount & $[[23,10,8;11]]_2$ & $[[26,10,8;16]]_2$ & \cite{G-binary-2025} & 1 & 0 & MI \\
\stepcounter{ternarycount}\theternarycount & $[[23,11,8;12]]_2$ & $[[26,10,8;16]]_2$ & \cite{G-binary-2025} & 0 & 0 & MI & \stepcounter{ternarycount}\theternarycount & $[[24,3,17;21]]_2$ & $[[36,3,17;23]]_2$ & \cite{G-binary-2025} & 0 & 1 & MI \\
\stepcounter{ternarycount}\theternarycount & $[[24,4,15;20]]_2$ & $[[34,4,15;23]]_2$ & \cite{G-binary-2025} & 0 & 1 & MI & \stepcounter{ternarycount}\theternarycount & $[[24,5,13;17]]_2$ & $-$ & $-$ & 1 & 1 & NP \\
\stepcounter{ternarycount}\theternarycount & $[[24,6,13;18]]_2$ & $[[36,6,6;30]]_2$ & \cite{ZK2026} & 0 & 1 & MI & \stepcounter{ternarycount}\theternarycount & $[[24,9,9;13]]_2$ & $[[27,9,9;16]]_2$ & \cite{G-binary-2025} & -1 & 1 & MI \\
\stepcounter{ternarycount}\theternarycount & $[[24,11,7;9]]_2$ & $[[24,10,7;9]]_2$ & \cite{G-binary-2025} & 1 & 1 & HD & \stepcounter{ternarycount}\theternarycount & $[[24,13,7;11]]_2$ & $[[36,12,4;24]]_2$ & \cite{ZK2026} & 0 & 1 & MI \\
\stepcounter{ternarycount}\theternarycount & $[[25,2,19;23]]_2$ & $-$ & $-$ & -1 & 0 & NP & \stepcounter{ternarycount}\theternarycount & $[[25,7,12;18]]_2$ & $[[36,6,6;30]]_2$ & \cite{ZK2026} & -1 & 1 & MI \\
\stepcounter{ternarycount}\theternarycount & $[[25,8,11;17]]_2$ & $[[36,8,4;28]]_2$ & \cite{ZK2026} & 0 & 1 & MI & \stepcounter{ternarycount}\theternarycount & $[[25,9,8;8]]_2$ & $[[25,9,8;9]]_2$ & \cite{G-binary-2025} & 1 & 1 & LE \\
\stepcounter{ternarycount}\theternarycount & $[[25,9,10;16]]_2$ & $[[27,9,9;16]]_2$ & \cite{G-binary-2025} & 0 & 1 & MI & \stepcounter{ternarycount}\theternarycount & $[[25,10,8;9]]_2$ & $[[25,9,8;9]]_2$ & \cite{G-binary-2025} & 0 & 1 & HD \\
\stepcounter{ternarycount}\theternarycount & $[[25,10,9;15]]_2$ & $[[26,10,8;16]]_2$ & \cite{G-binary-2025} & -1 & 1 & MI & \stepcounter{ternarycount}\theternarycount & $[[25,11,8;10]]_2$ & $[[26,10,8;16]]_2$ & \cite{G-binary-2025} & -1 & 1 & MI \\
\stepcounter{ternarycount}\theternarycount & $[[25,12,8;13]]_2$ & $[[26,10,8;16]]_2$ & \cite{G-binary-2025} & 0 & 0 & MI & \stepcounter{ternarycount}\theternarycount & $[[25,14,7;11]]_2$ & $[[29,14,7;13]]_2$ & \cite{G-binary-2025} & -1 & 1 & MI \\
\stepcounter{ternarycount}\theternarycount & $[[26,3,18;23]]_2$ & $[[36,3,17;23]]_2$ & \cite{G-binary-2025} & 0 & 1 & MI & \stepcounter{ternarycount}\theternarycount & $[[26,5,15;21]]_2$ & $[[34,5,15;23]]_2$ & \cite{G-binary-2025} & 0 & 1 & MI \\
\stepcounter{ternarycount}\theternarycount & $[[26,6,12;16]]_2$ & $[[29,6,11;16]]_2$ & \cite{G-binary-2025} & 1 & 1 & MI & \stepcounter{ternarycount}\theternarycount & $[[26,7,11;15]]_2$ & $[[28,7,10;16]]_2$ & \cite{G-binary-2025} & 1 & 1 & MI \\
\stepcounter{ternarycount}\theternarycount & $[[26,7,12;17]]_2$ & $[[36,6,6;30]]_2$ & \cite{ZK2026} & 1 & 1 & MI & \stepcounter{ternarycount}\theternarycount & $[[26,7,13;19]]_2$ & $[[36,6,6;30]]_2$ & \cite{ZK2026} & -1 & 1 & MI \\
\stepcounter{ternarycount}\theternarycount & $[[26,8,10;14]]_2$ & $[[27,8,9;16]]_2$ & \cite{G-binary-2025} & 1 & 1 & MI & \stepcounter{ternarycount}\theternarycount & $[[26,8,11;16]]_2$ & $[[27,8,9;16]]_2$ & \cite{G-binary-2025} & 0 & 1 & MI \\
\stepcounter{ternarycount}\theternarycount & $[[26,8,12;18]]_2$ & $[[36,8,4;28]]_2$ & \cite{ZK2026} & -1 & 1 & MI & \stepcounter{ternarycount}\theternarycount & $[[26,9,10;15]]_2$ & $[[27,9,9;16]]_2$ & \cite{G-binary-2025} & 1 & 1 & MI \\
\stepcounter{ternarycount}\theternarycount & $[[26,9,11;17]]_2$ & $[[36,9,4;27]]_2$ & \cite{ZK2026} & 0 & 1 & MI & \stepcounter{ternarycount}\theternarycount & $[[26,10,9;14]]_2$ & $[[26,10,8;16]]_2$ & \cite{G-binary-2025} & 1 & 1 & MI \\
\stepcounter{ternarycount}\theternarycount & $[[26,10,10;16]]_2$ & $[[26,10,8;16]]_2$ & \cite{G-binary-2025} & 0 & 1 & LD & \stepcounter{ternarycount}\theternarycount & $[[26,11,9;15]]_2$ & $[[26,10,8;16]]_2$ & \cite{G-binary-2025} & 0 & 1 & MI \\
\stepcounter{ternarycount}\theternarycount & $[[26,12,8;12]]_2$ & $[[26,10,8;16]]_2$ & \cite{G-binary-2025} & 0 & 1 & MI & \stepcounter{ternarycount}\theternarycount & $[[26,13,8;13]]_2$ & $[[30,13,8;13]]_2$ & \cite{G-binary-2025} & 0 & 1 & SL \\
\stepcounter{ternarycount}\theternarycount & $[[26,18,5;8]]_2$ & $[[36,12,4;24]]_2$ & \cite{ZK2026} & -1 & 1 & MI & \stepcounter{ternarycount}\theternarycount & $[[27,2,21;25]]_2$ & $-$ & $-$ & 0 & 1 & NP \\
\stepcounter{ternarycount}\theternarycount & $[[27,3,19;24]]_2$ & $-$ & $-$ & 0 & 1 & NP & \stepcounter{ternarycount}\theternarycount & $[[27,6,14;21]]_2$ & $[[33,5,14;23]]_2$ & \cite{G-binary-2025} & 0 & 1 & MI \\
\stepcounter{ternarycount}\theternarycount & $[[27,9,10;14]]_2$ & $[[27,9,9;16]]_2$ & \cite{G-binary-2025} & 1 & 1 & MI & \stepcounter{ternarycount}\theternarycount & $[[27,10,10;15]]_2$ & $[[27,10,9;16]]_2$ & \cite{G-binary-2025} & 1 & 1 & MI \\
\midrule
\end{tabular}%
}
\\[1ex]
\resizebox{\linewidth}{!}{%
\setlength{\tabcolsep}{3pt}%
\begin{tabular}{c|llc|cc|c||c|llc|cc|c}
\midrule
\mbox{No.} & \mbox{Our Qubit Code} & \mbox{Known Code} & \mbox{Source} & \mbox{$\Delta_{\Hull_{\rm H}}$} & \mbox{$\Delta_{d^{\perp_{\rm H}}}$} & \mbox{Type} & \mbox{No.} & \mbox{Our Qubit Code} & \mbox{Known Code} & \mbox{Source} & \mbox{$\Delta_{\Hull_{\rm H}}$} & \mbox{$\Delta_{d^{\perp_{\rm H}}}$} & \mbox{Type} \\
\midrule
\stepcounter{ternarycount}\theternarycount & $[[27,10,11;17]]_2$ & $[[36,10,4;26]]_2$ & \cite{ZK2026} & -1 & 1 & MI & \stepcounter{ternarycount}\theternarycount & $[[27,11,9;12]]_2$ & $[[27,11,9;16]]_2$ & \cite{G-binary-2025} & 1 & 0 & LE \\
\stepcounter{ternarycount}\theternarycount & $[[27,11,10;16]]_2$ & $[[27,11,9;16]]_2$ & \cite{G-binary-2025} & -1 & 1 & LD & \stepcounter{ternarycount}\theternarycount & $[[27,12,9;13]]_2$ & $[[27,11,9;16]]_2$ & \cite{G-binary-2025} & 1 & 0 & MI \\
\stepcounter{ternarycount}\theternarycount & $[[27,13,9;14]]_2$ & $[[27,11,9;16]]_2$ & \cite{G-binary-2025} & 0 & 0 & MI & \stepcounter{ternarycount}\theternarycount & $[[27,15,7;12]]_2$ & $[[29,15,7;13]]_2$ & \cite{G-binary-2025} & -1 & 1 & MI \\
\stepcounter{ternarycount}\theternarycount & $[[27,17,6;10]]_2$ & $[[36,12,4;24]]_2$ & \cite{ZK2026} & -1 & 1 & MI & \stepcounter{ternarycount}\theternarycount & $[[27,18,5;7]]_2$ & $[[36,12,4;24]]_2$ & \cite{ZK2026} & 0 & 1 & MI \\
\stepcounter{ternarycount}\theternarycount & $[[27,19,5;8]]_2$ & $[[36,12,4;24]]_2$ & \cite{ZK2026} & -1 & 1 & MI & \stepcounter{ternarycount}\theternarycount & $[[28,2,22;26]]_2$ & $-$ & $-$ & 0 & 1 & NP \\
\stepcounter{ternarycount}\theternarycount & $[[28,6,15;22]]_2$ & $[[34,5,15;23]]_2$ & \cite{G-binary-2025} & 0 & 1 & MI & \stepcounter{ternarycount}\theternarycount & $[[28,7,14;21]]_2$ & $[[33,7,14;23]]_2$ & \cite{G-binary-2025} & -1 & 1 & MI \\
\stepcounter{ternarycount}\theternarycount & $[[28,8,13;20]]_2$ & $[[36,8,4;28]]_2$ & \cite{ZK2026} & 0 & 1 & MI & \stepcounter{ternarycount}\theternarycount & $[[28,9,12;19]]_2$ & $[[36,9,4;27]]_2$ & \cite{ZK2026} & 0 & 1 & MI \\
\stepcounter{ternarycount}\theternarycount & $[[28,11,10;15]]_2$ & $[[28,11,10;16]]_2$ & \cite{G-binary-2025} & 1 & 1 & LE & \stepcounter{ternarycount}\theternarycount & $[[28,14,7;10]]_2$ & $[[29,14,7;13]]_2$ & \cite{G-binary-2025} & 1 & 1 & MI \\
\stepcounter{ternarycount}\theternarycount & $[[28,15,7;11]]_2$ & $[[29,15,7;13]]_2$ & \cite{G-binary-2025} & 0 & 1 & MI & \stepcounter{ternarycount}\theternarycount & $[[28,16,7;12]]_2$ & $[[29,16,7;13]]_2$ & \cite{G-binary-2025} & -1 & 1 & MI \\
\stepcounter{ternarycount}\theternarycount & $[[28,17,6;9]]_2$ & $[[36,12,4;24]]_2$ & \cite{ZK2026} & 0 & 1 & MI & \stepcounter{ternarycount}\theternarycount & $[[28,18,6;10]]_2$ & $[[36,12,4;24]]_2$ & \cite{ZK2026} & 0 & 1 & MI \\
\stepcounter{ternarycount}\theternarycount & $[[28,19,5;7]]_2$ & $[[36,12,4;24]]_2$ & \cite{ZK2026} & -1 & 1 & MI & \stepcounter{ternarycount}\theternarycount & $[[28,20,5;8]]_2$ & $[[36,12,4;24]]_2$ & \cite{ZK2026} & 0 & 1 & MI \\
\stepcounter{ternarycount}\theternarycount & $[[29,5,17;24]]_2$ & $-$ & $-$ & 0 & 1 & NP & \stepcounter{ternarycount}\theternarycount & $[[29,6,9;7]]_2$ & $[[29,6,9;10]]_2$ & \cite{G-binary-2025} & 0 & 1 & LE \\
\stepcounter{ternarycount}\theternarycount & $[[29,6,16;23]]_2$ & $[[33,6,14;23]]_2$ & \cite{G-binary-2025} & -1 & 1 & MI & \stepcounter{ternarycount}\theternarycount & $[[29,7,9;8]]_2$ & $[[29,7,9;10]]_2$ & \cite{G-binary-2025} & 1 & 0 & LE \\
\stepcounter{ternarycount}\theternarycount & $[[29,8,9;9]]_2$ & $[[29,7,9;10]]_2$ & \cite{G-binary-2025} & 0 & 0 & MI & \stepcounter{ternarycount}\theternarycount & $[[29,9,9;10]]_2$ & $[[29,7,9;10]]_2$ & \cite{G-binary-2025} & -1 & 0 & HD \\
\stepcounter{ternarycount}\theternarycount & $[[29,9,12;18]]_2$ & $[[36,9,4;27]]_2$ & \cite{ZK2026} & 1 & 1 & MI & \stepcounter{ternarycount}\theternarycount & $[[29,10,12;19]]_2$ & $[[36,10,4;26]]_2$ & \cite{ZK2026} & 0 & 1 & MI \\
\stepcounter{ternarycount}\theternarycount & $[[29,14,7;9]]_2$ & $[[29,14,7;13]]_2$ & \cite{G-binary-2025} & 0 & 1 & LE & \stepcounter{ternarycount}\theternarycount & $[[29,14,9;15]]_2$ & $[[30,9,9;15]]_2$ & \cite{G-binary-2025} & 0 & 0 & MI \\
\stepcounter{ternarycount}\theternarycount & $[[29,15,7;10]]_2$ & $[[29,15,7;13]]_2$ & \cite{G-binary-2025} & 1 & 1 & LE & \stepcounter{ternarycount}\theternarycount & $[[29,16,7;11]]_2$ & $[[29,16,7;13]]_2$ & \cite{G-binary-2025} & 0 & 1 & LE \\
\stepcounter{ternarycount}\theternarycount & $[[29,17,6;8]]_2$ & $[[36,12,4;24]]_2$ & \cite{ZK2026} & 0 & 1 & MI & \stepcounter{ternarycount}\theternarycount & $[[29,17,7;12]]_2$ & $[[29,16,7;13]]_2$ & \cite{G-binary-2025} & -1 & 1 & MI \\
\stepcounter{ternarycount}\theternarycount & $[[29,18,6;9]]_2$ & $[[36,12,4;24]]_2$ & \cite{ZK2026} & -1 & 1 & MI & \stepcounter{ternarycount}\theternarycount & $[[29,19,5;6]]_2$ & $[[36,12,4;24]]_2$ & \cite{ZK2026} & -1 & 1 & MI \\
\stepcounter{ternarycount}\theternarycount & $[[29,21,5;8]]_2$ & $[[36,12,4;24]]_2$ & \cite{ZK2026} & 0 & 1 & MI & \stepcounter{ternarycount}\theternarycount & $[[30,2,23;28]]_2$ & $-$ & $-$ & -1 & 0 & NP \\
\stepcounter{ternarycount}\theternarycount & $[[30,4,19;26]]_2$ & $-$ & $-$ & 0 & 1 & NP & \stepcounter{ternarycount}\theternarycount & $[[30,8,14;22]]_2$ & $[[33,8,14;23]]_2$ & \cite{G-binary-2025} & 0 & 1 & MI \\
\stepcounter{ternarycount}\theternarycount & $[[30,11,10;15]]_2$ & $[[30,9,9;15]]_2$ & \cite{G-binary-2025} & 1 & 1 & MI & \stepcounter{ternarycount}\theternarycount & $[[30,13,8;11]]_2$ & $[[30,13,8;13]]_2$ & \cite{G-binary-2025} & 1 & 1 & LE \\
\stepcounter{ternarycount}\theternarycount & $[[30,13,10;17]]_2$ & $[[34,13,10;19]]_2$ & \cite{G-binary-2025} & 0 & 1 & MI & \stepcounter{ternarycount}\theternarycount & $[[30,14,8;12]]_2$ & $[[30,14,8;13]]_2$ & \cite{G-binary-2025} & 0 & 1 & LE \\
\stepcounter{ternarycount}\theternarycount & $[[30,16,7;10]]_2$ & $[[36,12,4;24]]_2$ & \cite{ZK2026} & 1 & 1 & MI & \stepcounter{ternarycount}\theternarycount & $[[30,17,7;11]]_2$ & $[[36,12,4;24]]_2$ & \cite{ZK2026} & 0 & 1 & MI \\
\stepcounter{ternarycount}\theternarycount & $[[30,18,7;12]]_2$ & $[[33,18,7;13]]_2$ & \cite{G-binary-2025} & -1 & 1 & MI & \stepcounter{ternarycount}\theternarycount & $[[31,2,24;29]]_2$ & $[[32,2,24;30]]_2$ & \cite{FCLLC2025} & 0 & 1 & MI \\
\stepcounter{ternarycount}\theternarycount & $[[31,4,20;27]]_2$ & $-$ & $-$ & 0 & 1 & NP & \stepcounter{ternarycount}\theternarycount & $[[31,6,17;25]]_2$ & $[[36,6,6;30]]_2$ & \cite{ZK2026} & 0 & 1 & MI \\
\stepcounter{ternarycount}\theternarycount & $[[31,7,10;8]]_2$ & $[[31,7,9;8]]_2$ & \cite{G-binary-2025} & 1 & 1 & LD & \stepcounter{ternarycount}\theternarycount & $[[31,7,16;24]]_2$ & $[[36,6,6;30]]_2$ & \cite{ZK2026} & 0 & 1 & MI \\
\stepcounter{ternarycount}\theternarycount & $[[31,8,10;9]]_2$ & $[[31,8,9;10]]_2$ & \cite{G-binary-2025} & 1 & 1 & MI & \stepcounter{ternarycount}\theternarycount & $[[31,8,15;23]]_2$ & $[[33,8,14;23]]_2$ & \cite{G-binary-2025} & -1 & 1 & MI \\
\stepcounter{ternarycount}\theternarycount & $[[31,9,10;10]]_2$ & $[[31,8,9;10]]_2$ & \cite{G-binary-2025} & 0 & 1 & MI & \stepcounter{ternarycount}\theternarycount & $[[31,10,10;11]]_2$ & $[[31,10,9;13]]_2$ & \cite{G-binary-2025} & -1 & 1 & MI \\
\stepcounter{ternarycount}\theternarycount & $[[31,10,13;21]]_2$ & $[[36,10,4;26]]_2$ & \cite{ZK2026} & 0 & 1 & MI & \stepcounter{ternarycount}\theternarycount & $[[31,11,10;12]]_2$ & $[[31,11,9;13]]_2$ & \cite{G-binary-2025} & 1 & 0 & MI \\
\stepcounter{ternarycount}\theternarycount & $[[31,11,12;20]]_2$ & $[[36,11,4;25]]_2$ & \cite{ZK2026} & -1 & 1 & MI & \stepcounter{ternarycount}\theternarycount & $[[31,12,10;13]]_2$ & $[[31,12,9;13]]_2$ & \cite{G-binary-2025} & 1 & 0 & LD \\
\stepcounter{ternarycount}\theternarycount & $[[31,12,11;19]]_2$ & $[[35,12,11;19]]_2$ & \cite{G-binary-2025} & 0 & 1 & SL & \stepcounter{ternarycount}\theternarycount & $[[31,13,8;10]]_2$ & $[[36,12,4;24]]_2$ & \cite{ZK2026} & 1 & 0 & MI \\
\stepcounter{ternarycount}\theternarycount & $[[31,13,9;12]]_2$ & $[[31,13,9;13]]_2$ & \cite{G-binary-2025} & 1 & 0 & LE & \stepcounter{ternarycount}\theternarycount & $[[31,13,10;14]]_2$ & $[[31,12,10;16]]_2$ & \cite{G-binary-2025} & 0 & 0 & MI \\
\stepcounter{ternarycount}\theternarycount & $[[31,14,8;11]]_2$ & $[[36,12,4;24]]_2$ & \cite{ZK2026} & 0 & 0 & MI & \stepcounter{ternarycount}\theternarycount & $[[31,14,10;15]]_2$ & $[[31,12,10;16]]_2$ & \cite{G-binary-2025} & -1 & 0 & MI \\
\stepcounter{ternarycount}\theternarycount & $[[31,15,8;12]]_2$ & $[[36,12,4;24]]_2$ & \cite{ZK2026} & -1 & 0 & MI & \stepcounter{ternarycount}\theternarycount & $[[32,2,25;30]]_2$ & $[[32,2,24;30]]_2$ & \cite{FCLLC2025} & 0 & 1 & LD \\
\stepcounter{ternarycount}\theternarycount & $[[32,11,12;19]]_2$ & $[[36,11,12;19]]_2$ & \cite{G-binary-2025} & 0 & 1 & SL & \stepcounter{ternarycount}\theternarycount & $[[32,11,13;21]]_2$ & $[[36,11,4;25]]_2$ & \cite{ZK2026} & 0 & 1 & MI \\
\stepcounter{ternarycount}\theternarycount & $[[32,12,11;18]]_2$ & $[[35,12,11;19]]_2$ & \cite{G-binary-2025} & 1 & 1 & MI & \stepcounter{ternarycount}\theternarycount & $[[32,12,12;20]]_2$ & $[[36,12,4;24]]_2$ & \cite{ZK2026} & 0 & 1 & MI \\
\stepcounter{ternarycount}\theternarycount & $[[32,13,11;19]]_2$ & $[[34,13,10;19]]_2$ & \cite{G-binary-2025} & 0 & 1 & MI & \stepcounter{ternarycount}\theternarycount & $[[33,2,26;31]]_2$ & $-$ & $-$ & -1 & 1 & NP \\
\stepcounter{ternarycount}\theternarycount & $[[33,3,23;30]]_2$ & $-$ & $-$ & 0 & 1 & NP & \stepcounter{ternarycount}\theternarycount & $[[33,8,16;25]]_2$ & $[[36,8,4;28]]_2$ & \cite{ZK2026} & 0 & 1 & MI \\
\stepcounter{ternarycount}\theternarycount & $[[33,9,15;24]]_2$ & $[[36,9,4;27]]_2$ & \cite{ZK2026} & 0 & 1 & MI & \stepcounter{ternarycount}\theternarycount & $[[33,12,11;17]]_2$ & $[[35,12,11;19]]_2$ & \cite{G-binary-2025} & 1 & 1 & MI \\
\stepcounter{ternarycount}\theternarycount & $[[33,12,12;19]]_2$ & $[[35,12,11;19]]_2$ & \cite{G-binary-2025} & 1 & 1 & MI & \stepcounter{ternarycount}\theternarycount & $[[33,13,11;18]]_2$ & $[[34,13,10;19]]_2$ & \cite{G-binary-2025} & 0 & 1 & MI \\
\stepcounter{ternarycount}\theternarycount & $[[33,13,12;20]]_2$ & $[[36,12,4;24]]_2$ & \cite{ZK2026} & 0 & 1 & MI & \stepcounter{ternarycount}\theternarycount & $[[33,14,10;15]]_2$ & $[[33,8,10;15]]_2$ & \cite{G-binary-2025} & 1 & 0 & HD \\
\stepcounter{ternarycount}\theternarycount & $[[33,14,11;19]]_2$ & $[[34,14,10;19]]_2$ & \cite{G-binary-2025} & -1 & 1 & MI & \stepcounter{ternarycount}\theternarycount & $[[33,15,10;16]]_2$ & $[[33,12,10;16]]_2$ & \cite{G-binary-2025} & 1 & 0 & HD \\
\stepcounter{ternarycount}\theternarycount & $[[33,16,7;9]]_2$ & $[[36,12,4;24]]_2$ & \cite{ZK2026} & 1 & 1 & MI & \stepcounter{ternarycount}\theternarycount & $[[33,16,10;17]]_2$ & $[[34,15,10;19]]_2$ & \cite{G-binary-2025} & 0 & 0 & MI \\
\stepcounter{ternarycount}\theternarycount & $[[33,17,7;10]]_2$ & $[[36,12,4;24]]_2$ & \cite{ZK2026} & 1 & 1 & MI & \stepcounter{ternarycount}\theternarycount & $[[33,17,8;12]]_2$ & $[[33,17,8;13]]_2$ & \cite{G-binary-2025} & 0 & 1 & LE \\
\stepcounter{ternarycount}\theternarycount & $[[33,20,6;9]]_2$ & $[[36,12,4;24]]_2$ & \cite{ZK2026} & 1 & 1 & MI & \stepcounter{ternarycount}\theternarycount & $[[33,21,6;10]]_2$ & $[[36,12,4;24]]_2$ & \cite{ZK2026} & 0 & 1 & MI \\
\stepcounter{ternarycount}\theternarycount & $[[33,22,6;11]]_2$ & $[[36,12,4;24]]_2$ & \cite{ZK2026} & -1 & 1 & MI & \stepcounter{ternarycount}\theternarycount & $[[34,3,24;31]]_2$ & $-$ & $-$ & 0 & 1 & NP \\
\stepcounter{ternarycount}\theternarycount & $[[34,4,22;30]]_2$ & $-$ & $-$ & -1 & 1 & NP & \stepcounter{ternarycount}\theternarycount & $[[34,6,19;28]]_2$ & $[[36,6,6;30]]_2$ & \cite{ZK2026} & -1 & 1 & MI \\
\stepcounter{ternarycount}\theternarycount & $[[34,10,13;20]]_2$ & $[[36,10,4;26]]_2$ & \cite{ZK2026} & 1 & 1 & MI & \stepcounter{ternarycount}\theternarycount & $[[34,10,14;22]]_2$ & $[[38,10,14;22]]_2$ & \cite{G-binary-2025} & 1 & 1 & SL \\
\stepcounter{ternarycount}\theternarycount & $[[34,11,12;19]]_2$ & $[[36,11,12;19]]_2$ & \cite{G-binary-2025} & 1 & 1 & SL & \stepcounter{ternarycount}\theternarycount & $[[34,12,13;22]]_2$ & $[[36,12,4;24]]_2$ & \cite{ZK2026} & 0 & 1 & MI \\
\midrule
\end{tabular}%
}
\\[1ex]
\resizebox{\linewidth}{!}{%
\setlength{\tabcolsep}{3pt}%
\begin{tabular}{c|llc|cc|c||c|llc|cc|c}
\hline
\mbox{No.} & \mbox{Our Qubit Code} & \mbox{Known Code} & \mbox{Source} & \mbox{$\Delta_{\Hull_{\rm H}}$} & \mbox{$\Delta_{d^{\perp_{\rm H}}}$} & \mbox{Type} & \mbox{No.} & \mbox{Our Qubit Code} & \mbox{Known Code} & \mbox{Source} & \mbox{$\Delta_{\Hull_{\rm H}}$} & \mbox{$\Delta_{d^{\perp_{\rm H}}}$} & \mbox{Type} \\
\hline\hline
\stepcounter{ternarycount}\theternarycount & $[[34,18,7;10]]_2$ & $[[36,12,4;24]]_2$ & \cite{ZK2026} & 1 & 1 & MI & \stepcounter{ternarycount}\theternarycount & $[[34,18,9;16]]_2$ & $[[36,12,4;24]]_2$ & \cite{ZK2026} & 0 & 1 & MI \\
\stepcounter{ternarycount}\theternarycount & $[[35,2,27;33]]_2$ & $-$ & $-$ & 0 & 1 & NP & \stepcounter{ternarycount}\theternarycount & $[[35,11,12;18]]_2$ & $[[36,11,12;19]]_2$ & \cite{G-binary-2025} & 1 & 1 & MI \\
\stepcounter{ternarycount}\theternarycount & $[[35,12,12;19]]_2$ & $[[35,12,11;19]]_2$ & \cite{G-binary-2025} & 0 & 1 & LD & \stepcounter{ternarycount}\theternarycount & $[[35,14,12;21]]_2$ & $[[36,12,4;24]]_2$ & \cite{ZK2026} & -1 & 1 & MI \\
\stepcounter{ternarycount}\theternarycount & $[[35,20,8;15]]_2$ & $[[37,20,8;17]]_2$ & \cite{G-binary-2025} & -1 & 1 & MI & \stepcounter{ternarycount}\theternarycount & $[[36,13,11;17]]_2$ & $[[36,12,4;24]]_2$ & \cite{ZK2026} & 1 & 1 & MI \\
\stepcounter{ternarycount}\theternarycount & $[[36,20,8;14]]_2$ & $[[37,20,8;17]]_2$ & \cite{G-binary-2025} & 1 & 1 & MI & \stepcounter{ternarycount}\theternarycount & $[[36,21,8;15]]_2$ & $[[37,21,7;16]]_2$ & \cite{G-binary-2025} & 0 & 1 & MI \\
\stepcounter{ternarycount}\theternarycount & $[[36,22,7;12]]_2$ & $[[37,21,7;16]]_2$ & \cite{G-binary-2025} & 1 & 1 & MI & \stepcounter{ternarycount}\theternarycount & $[[36,23,7;13]]_2$ & $[[37,21,7;16]]_2$ & \cite{G-binary-2025} & -1 & 1 & MI \\
\stepcounter{ternarycount}\theternarycount & $[[37,2,29;35]]_2$ & $-$ & $-$ & 0 & 1 & NP & \stepcounter{ternarycount}\theternarycount & $[[37,3,26;34]]_2$ & $-$ & $-$ & -1 & 1 & NP \\
\stepcounter{ternarycount}\theternarycount & $[[37,8,18;27]]_2$ & $[[60,6,10;54]]_2$ & \cite{ZK2026} & 1 & 1 & MI & \stepcounter{ternarycount}\theternarycount & $[[37,9,11;10]]_2$ & $[[37,9,10;13]]_2$ & \cite{G-binary-2025} & 0 & 1 & MI \\
\stepcounter{ternarycount}\theternarycount & $[[37,9,18;28]]_2$ & $[[60,6,10;54]]_2$ & \cite{ZK2026} & 0 & 1 & MI & \stepcounter{ternarycount}\theternarycount & $[[37,10,11;11]]_2$ & $[[37,10,10;13]]_2$ & \cite{G-binary-2025} & 0 & 0 & MI \\
\stepcounter{ternarycount}\theternarycount & $[[37,11,11;12]]_2$ & $[[37,11,10;13]]_2$ & \cite{G-binary-2025} & -1 & 0 & MI & \stepcounter{ternarycount}\theternarycount & $[[37,12,11;13]]_2$ & $[[37,12,10;13]]_2$ & \cite{G-binary-2025} & 0 & 1 & LD \\
\stepcounter{ternarycount}\theternarycount & $[[37,13,9;10]]_2$ & $[[37,13,9;12]]_2$ & \cite{G-binary-2025} & 1 & 1 & LE & \stepcounter{ternarycount}\theternarycount & $[[37,13,11;14]]_2$ & $[[37,11,11;16]]_2$ & \cite{G-binary-2025} & 1 & 0 & MI \\
\stepcounter{ternarycount}\theternarycount & $[[37,13,14;24]]_2$ & $[[60,12,4;48]]_2$ & \cite{ZK2026} & -1 & 1 & MI & \stepcounter{ternarycount}\theternarycount & $[[37,14,9;11]]_2$ & $[[37,14,9;12]]_2$ & \cite{G-binary-2025} & 1 & 1 & LE \\
\stepcounter{ternarycount}\theternarycount & $[[37,14,11;15]]_2$ & $[[37,11,11;16]]_2$ & \cite{G-binary-2025} & 1 & 0 & MI & \stepcounter{ternarycount}\theternarycount & $[[37,15,11;16]]_2$ & $[[37,11,11;16]]_2$ & \cite{G-binary-2025} & 0 & 0 & HD \\
\stepcounter{ternarycount}\theternarycount & $[[37,16,11;17]]_2$ & $[[37,16,11;19]]_2$ & \cite{G-binary-2025} & -1 & 0 & LE & \stepcounter{ternarycount}\theternarycount & $[[37,17,8;10]]_2$ & $[[60,15,4;45]]_2$ & \cite{ZK2026} & -1 & 1 & MI \\
\stepcounter{ternarycount}\theternarycount & $[[37,18,8;11]]_2$ & $[[37,18,8;12]]_2$ & \cite{G-binary-2025} & -1 & 1 & LE & \stepcounter{ternarycount}\theternarycount & $[[37,18,10;17]]_2$ & $[[39,18,10;17]]_2$ & \cite{G-binary-2025} & 0 & 0 & SL \\
\stepcounter{ternarycount}\theternarycount & $[[37,19,8;12]]_2$ & $[[37,18,8;12]]_2$ & \cite{G-binary-2025} & 1 & 1 & HD & \stepcounter{ternarycount}\theternarycount & $[[37,19,9;16]]_2$ & $[[38,19,9;17]]_2$ & \cite{G-binary-2025} & -1 & 1 & MI \\
\stepcounter{ternarycount}\theternarycount & $[[37,19,10;18]]_2$ & $[[60,18,4;42]]_2$ & \cite{ZK2026} & -1 & 0 & MI & \stepcounter{ternarycount}\theternarycount & $[[37,20,8;13]]_2$ & $[[37,19,8;13]]_2$ & \cite{G-binary-2025} & 0 & 1 & HD \\
\stepcounter{ternarycount}\theternarycount & $[[37,20,9;17]]_2$ & $[[37,20,8;17]]_2$ & \cite{G-binary-2025} & 0 & 1 & LD & \stepcounter{ternarycount}\theternarycount & $[[37,21,8;14]]_2$ & $[[37,21,7;16]]_2$ & \cite{G-binary-2025} & 0 & 1 & MI \\
\stepcounter{ternarycount}\theternarycount & $[[37,22,8;15]]_2$ & $[[38,22,8;16]]_2$ & \cite{G-binary-2025} & 0 & 1 & MI & \stepcounter{ternarycount}\theternarycount & $[[37,23,7;12]]_2$ & $[[37,21,7;16]]_2$ & \cite{G-binary-2025} & 0 & 1 & MI \\
\stepcounter{ternarycount}\theternarycount & $[[37,24,7;13]]_2$ & $[[41,24,7;14]]_2$ & \cite{G-binary-2025} & 0 & 1 & MI & \stepcounter{ternarycount}\theternarycount & $[[38,8,19;30]]_2$ & $[[60,6,10;54]]_2$ & \cite{ZK2026} & 0 & 1 & MI \\
\stepcounter{ternarycount}\theternarycount & $[[38,9,18;29]]_2$ & $[[60,6,10;54]]_2$ & \cite{ZK2026} & 0 & 1 & MI & \stepcounter{ternarycount}\theternarycount & $[[38,10,17;28]]_2$ & $[[60,6,10;54]]_2$ & \cite{ZK2026} & -1 & 1 & MI \\
\stepcounter{ternarycount}\theternarycount & $[[38,19,9;15]]_2$ & $[[38,19,9;17]]_2$ & \cite{G-binary-2025} & 1 & 1 & LE & \stepcounter{ternarycount}\theternarycount & $[[38,20,9;16]]_2$ & $[[38,20,8;16]]_2$ & \cite{G-binary-2025} & 1 & 1 & LD \\
\stepcounter{ternarycount}\theternarycount & $[[38,21,8;13]]_2$ & $[[38,19,8;13]]_2$ & \cite{G-binary-2025} & 0 & 1 & HD & \stepcounter{ternarycount}\theternarycount & $[[38,22,8;14]]_2$ & $[[38,22,8;16]]_2$ & \cite{G-binary-2025} & -1 & 1 & LE \\
\stepcounter{ternarycount}\theternarycount & $[[38,23,8;15]]_2$ & $[[38,22,8;16]]_2$ & \cite{G-binary-2025} & -1 & 1 & MI & \stepcounter{ternarycount}\theternarycount & $[[38,25,7;13]]_2$ & $[[41,25,7;14]]_2$ & \cite{G-binary-2025} & -1 & 1 & MI \\
\stepcounter{ternarycount}\theternarycount & $[[38,34,3;4]]_2$ & $[[60,9,2;51]]_2$ & \cite{ZK2026} & 0 & 1 & MI & \stepcounter{ternarycount}\theternarycount & $[[39,6,22;33]]_2$ & $[[60,6,10;54]]_2$ & \cite{ZK2026} & 0 & 1 & MI \\
\stepcounter{ternarycount}\theternarycount & $[[39,10,11;11]]_2$ & $[[39,10,10;11]]_2$ & \cite{G-binary-2025} & -1 & 1 & LD & \stepcounter{ternarycount}\theternarycount & $[[39,12,15;25]]_2$ & $[[60,12,4;48]]_2$ & \cite{ZK2026} & 1 & 1 & MI \\
\stepcounter{ternarycount}\theternarycount & $[[39,13,9;10]]_2$ & $[[60,12,4;48]]_2$ & \cite{ZK2026} & 0 & 1 & MI & \stepcounter{ternarycount}\theternarycount & $[[39,17,11;18]]_2$ & $[[39,16,11;19]]_2$ & \cite{G-binary-2025} & 1 & 0 & MI \\
\stepcounter{ternarycount}\theternarycount & $[[39,18,8;11]]_2$ & $[[39,18,8;12]]_2$ & \cite{G-binary-2025} & -1 & 1 & LE & \stepcounter{ternarycount}\theternarycount & $[[39,22,8;13]]_2$ & $[[39,19,8;13]]_2$ & \cite{G-binary-2025} & 1 & 1 & HD \\
\stepcounter{ternarycount}\theternarycount & $[[39,35,3;4]]_2$ & $[[60,9,2;51]]_2$ & \cite{ZK2026} & 0 & 1 & MI & \stepcounter{ternarycount}\theternarycount & $[[40,2,31;38]]_2$ & $-$ & $-$ & 0 & 1 & NP \\
\stepcounter{ternarycount}\theternarycount & $[[40,4,26;36]]_2$ & $-$ & $-$ & 0 & 1 & NP & \stepcounter{ternarycount}\theternarycount & $[[40,5,24;35]]_2$ & $-$ & $-$ & 0 & 1 & NP \\
\stepcounter{ternarycount}\theternarycount & $[[40,6,22;32]]_2$ & $[[60,6,10;54]]_2$ & \cite{ZK2026} & 1 & 1 & MI & \stepcounter{ternarycount}\theternarycount & $[[40,12,15;24]]_2$ & $[[60,12,4;48]]_2$ & \cite{ZK2026} & 1 & 1 & MI \\
\stepcounter{ternarycount}\theternarycount & $[[40,13,15;25]]_2$ & $[[60,12,4;48]]_2$ & \cite{ZK2026} & 0 & 1 & MI & \stepcounter{ternarycount}\theternarycount & $[[40,14,15;26]]_2$ & $[[46,14,15;26]]_2$ & \cite{G-binary-2025} & -1 & 1 & SL \\
\stepcounter{ternarycount}\theternarycount & $[[40,21,8;13]]_2$ & $[[40,19,8;13]]_2$ & \cite{G-binary-2025} & 1 & 1 & HD & \stepcounter{ternarycount}\theternarycount & $[[40,24,7;12]]_2$ & $[[41,24,7;14]]_2$ & \cite{G-binary-2025} & 1 & 1 & MI \\
\stepcounter{ternarycount}\theternarycount & $[[40,26,7;14]]_2$ & $[[41,26,7;14]]_2$ & \cite{G-binary-2025} & 0 & 1 & SL & & & & & & & \\
\bottomrule
\end{tabular}%
}
\end{longtable}
\end{center}

\begin{table}[htbp]
\caption{New and improved EA qutrit codes of lengths $7 \leq n \leq 25$}
\label{tab:breakthroughs0118}
\centering
\begin{tabular}{c|l|l|c|c|c|c}
\toprule
No. & Our Qutrit Code & Known Code & Source & $\Delta_{\Hull_{\rm H}}$ & $\Delta_{d^{\perp_{\rm H}}}$ & Type \\
\midrule
$1$ & $[[7,1,7;6]]_3$ & $-$ & $-$ & 0 & 1 & NP \\
$2$ & $[[8,0,8;6]]_3$ & $[[9,0,8;6]]_3$ & \cite{G-ternary-2025} & 1 & 1 & SL \\
$3$ & $[[8,1,8;7]]_3$ & $[[11,0,8;7]]_3$ & \cite{G-ternary-2025} & 0 & 1 & MI \\
$4$ & $[[13,3,7;4]]_3$ & $[[13,3,7;5]]_3$ & \cite{G-ternary-2025} & 0 & 1 & LE \\
$5$ & $[[17,2,15;15]]_3$ & $[[26,0,14;18]]_3$ & \cite{G-ternary-2025} & 0 & 1 & MI \\
$6$ & $[[18,2,16;16]]_3$ & $[[21,0,16;16]]_3$ & \cite{G-ternary-2025} & 0 & 1 & MI \\
$7$ & $[[19,2,17;17]]_3$ & $[[26,0,14;18]]_3$ & \cite{G-ternary-2025} & 0 & 1 & MI \\
$8$ & $[[19,3,15;16]]_3$ & $[[26,0,14;18]]_3$ & \cite{G-ternary-2025} & 0 & 1 & MI \\
$9$ & $[[21,2,18;19]]_3$ & $-$ & $-$ & 0 & 1 & NP \\
$10$ & $[[22,2,19;20]]_3$ & $-$ & $-$ & 0 & 1 & NP \\
$11$ & $[[23,2,20;21]]_3$ & $-$ & $-$ & 0 & 1 & NP \\
$12$ & $[[24,2,21;22]]_3$ & $-$ & $-$ & 0 & 1 & NP \\
$13$ & $[[25,2,22;23]]_3$ & $[[26,2,22;24]]_3$ & \cite{LZ2024-DM} & 0 & 1 & MI \\
$14$ & $[[25,3,20;22]]_3$ & $-$ & $-$ & 0 & 1 & NP \\
\bottomrule
\end{tabular}
\end{table}

\begin{remark}\label{rem.known-dual}
    Table \ref{tab:change-distribution} summarizes the distribution of the listed entries 
according to the parameter changes in Table \ref{tab:parameter-transformations} and construction source. 
In particular, $10$ {\em extra} EAQECCs obtainable from our generalized extended codes are already known 
to improve Grassl's tables \cite{G-binary-2025,G-ternary-2025} by different methods 
in \cite{LLS2025,LSL2023,CLZ2025,LZ2024-DM,LLZS2024}. 
We emphasize that these $10$ EAQECCs are not included in Tables~\ref{tab:breakthroughs2140} and \ref{tab:breakthroughs0118}, 
and they are constructed in a unified way via our generalized extended codes, 
which is different from the methods used in the cited papers.
\end{remark}

\begin{table}[htbp]
\caption{Distribution of the new and improved EAQECCs listed in Tables~\ref{tab:breakthroughs2140} and \ref{tab:breakthroughs0118}}
\label{tab:change-distribution}
\centering
\begin{tabular}{c|c|cc|cc|c}
\toprule
No. & $(\Delta_{\Hull_{\rm H}},\Delta_{d^{\perp_{\rm H}}})$
& \# Qubit Codes & \# Qutrit Codes & \# Improved & \# New & Total \\
\midrule
1 & $(-1,0)$ & $11$ & $0$ & $9$ & $2$ & $11$ \\
2 & $(-1,1)$ & $57$ & $0$ & $53$ & $4$ & $57$ \\
3 & $(0,0)$ & $13$ & $0$ & $13$ & $0$ & $13$ \\
4 & $(0,1)$ & $121$ & $13$ & $104$ & $30$ & $134$ \\
5 & $(1,0)$ & $13$ & $0$ & $13$ & $0$ & $13$ \\
6 & $(1,1)$ & $52$ & $1$ & $52$ & $1$ & $53$ \\
\midrule
Total & $-$ & $267$ & $14$ & $244$ & $37$ & $281$ \\
\bottomrule
\end{tabular}
\end{table}

\begin{table}[htbp]
\caption{Reproducing some known EAQECCs with parameters improving Grassl's tables}
\label{tab:literature-separated-eaqeccs}
\centering
\begin{tabular}{c|l|l|cc|c|l}
\toprule
No. & Our Qubit Code & Known Code in \cite{G-binary-2025} & $\Delta_{\Hull_{\rm H}}$ & $\Delta_{d^{\perp_{\rm H}}}$ & Type & Also obtained in \\
\toprule
$1$ & $[[9,2,5;3]]_2$ & $[[9,2,5;4]]_2$ & $-1$ & $1$ & LE & \cite[Table~2]{LLS2025} \\
$2$ & $[[9,3,5;4]]_2$ & $[[9,2,5;4]]_2$ & $0$ & $1$ & HD & \cite[Table~6]{LSL2023} \\
$3$ & $[[11,2,7;7]]_2$ & $[[12,2,7;9]]_2$ & $1$ & $1$ & MI & \cite[Table~6]{LSL2023} \\
$4$ & $[[11,5,4;2]]_2$ & $[[11,5,4;3]]_2$ & $1$ & $1$ & LE & \cite[Table~2]{LLS2025} \\
$5$ & $[[13,3,6;4]]_2$ & $[[13,3,6;5]]_2$ & $1$ & $1$ & LE & \cite[Table~4]{LLS2025} \\ 
\midrule
No. & Our Qutrit Code & Known Code in \cite{G-ternary-2025} & $\Delta_{\Hull_{\rm H}}$ & $\Delta_{d^{\perp_{\rm H}}}$ & Type & Also obtained in \\
\midrule
$6$ & $[[7,0,7;5]]_3$ & $[[11,0,7;5]]_3$ & $1$ & $1$ & SL & \cite[Table~V]{CLZ2025} \\
$7$ & $[[7,2,5;3]]_3$ & $[[7,2,5;4]]_3$ & $0$ & $1$ & LE & \cite[Example~5]{LLS2025} \\
$8$ & $[[9,0,9;7]]_3$ & $[[9,0,9;8]]_3$  & $1$ & $1$ & LE & \cite[Table~1]{LLZS2024} \\
$9$ & $[[15,2,13;13]]_3$ & $[[23,0,13;13]]_3$ & $0$ & $1$ & MI & \cite[Table~4]{LZ2024-DM} \\
$10$ & $[[16,2,14;14]]_3$ & $[[26,0,14;18]]_3$ & $0$ & $1$ & MI & \cite[Table~V]{CLZ2025} \\
\bottomrule
\end{tabular}
\end{table}

\section{Concluding Remarks}\label{sec.conclusion}
In this paper, we propose generalized extended codes $\C({\bf u},a)$ of$q^2$-ary linear codes and use them to construct EAQECCs via the
Hermitian construction. This class of codes is obtained by considering the Hermitian
dual of certain second kind of extended codes discussed in \cite{SDC2024FFA,SDC2024DM} together with some monomial variations. We prove that every $[n+1,k+1]_{q^2}$ linear code whose Hermitian dual distance is greater than
one is monomially equivalent to a generalized extended code. This shows that generalized extended codes provide a unified representation for a large class of linear codes while retaining an explicit algebraic form suitable for controlling the Hermitian hulls and Hermitian dual distances.

We have determined the Hermitian hull dimension and the Hermitian dual distance of
$\C({\bf u},a)$. For the Hermitian hull, we establish that, if
$\ell=\dim(\Hull_{\rm H}(\C))$, then
$\dim(\Hull_{\rm H}(\C({\bf u},a)))$ can only be $\ell-1$, $\ell$, or
$\ell+1$. The exact case is classified by the position of the extension vector ${\bf u}$ relative to $\C+\C^{\perp_{\rm H}}$. 
We also show that, in the relevant case ${\bf u}\in (\C+\C^{\perp_{\rm H}})\setminus \C$, it suffices to study the representatives
in $\C^{\perp_{\rm H}}\setminus\Hull_{\rm H}(\C)$. For the Hermitian dual distance, we prove that
$d(\C({\bf u},a)^{\perp_{\rm H}})$ is either
$d(\C^{\perp_{\rm H}})$ or $d(\C^{\perp_{\rm H}})+1$. We then give necessary and
sufficient conditions for the latter case in terms of the interaction between
${\bf u}$ and the minimum weight codewords of $\C^{\perp_{\rm H}}$. These criteria were further reformulated in terms of maximal subcodes and finite-geometric conditions. In particular, we obtain explicit conditions under which the Hermitian hull dimension and the Hermitian dual distance increase simultaneously.

Applying these structural results to the Hermitian construction, we obtain 
six families of EAQECCs arising from $\C({\bf u},a)$. The examples and tables show that generalized extended codes $\C({\bf u},a)$ can produce new or improved EAQECCs. After detailed comparisons with best-known results collected in Grassl's online tables \cite{G-ternary-2025,G-binary-2025} 
and the latest records reported in \cite{LLS2025,K2023,FCLLC2025,LLZS2024,LZ2024-DM,CLZ2025,LZ2024-JAMC},
our construction produces $281$ new or improved parameters of qubit and qutrit EA codes, including $267$ qubit parameters for lengths $7 \leq n \leq 40$ and $14$ qutrit parameters for lengths $7 \leq n \leq 25$. 
Among them, $244$ codes have improved parameters when compared with the previous best-known while $37$ codes have are new parameters. Improvements occur in terms of length, dimension, minimum distance, as well as the amount of entanglement, with $180$ instances having simultaneous improvements in two or more parameters. Our construction also reproduces $10$ further improvements over Grassl's tables previously obtained by different methods in \cite{LLS2025,LSL2023,CLZ2025,LZ2024-DM,LLZS2024}, showing that generalized extended codes indeed provide a unified framework for obtaining excellent EAQECCs.

In addition to the Hermitian construction, the CSS and symplectic constructions constitute important methods for constructing EAQECCs. Extend the results in this paper to the Euclidean inner product setting is straightforward. The symplectic case, however, appears more complicated and may be worthy of a separate investigation.

\vfill
\end{document}